\documentclass{article}
\usepackage[utf8]{inputenc}

\usepackage[utf8]{inputenc}

\usepackage{amsmath, amsfonts, amssymb, amsthm}

\usepackage{mathtools}

\usepackage{graphicx}

\usepackage{subcaption}

\usepackage{bbm}

\usepackage{booktabs}

\usepackage{bm}

\usepackage{xcolor}

\usepackage{enumerate}
\usepackage[inline]{enumitem}
\usepackage{tikz}
\usetikzlibrary{arrows}%
\usetikzlibrary{arrows.meta}
\usepackage{ifthen} %

\usepackage{natbib}

\usepackage{algorithm, algpseudocode, algorithmicx}

\usepackage[normalem]{ulem}
\newcommand{\st}[1]{\sout{#1}}

\usepackage{centernot}

\usepackage{thmtools}
\usepackage{thm-restate} 
\usepackage[preprint,abbrvbib]{jmlr2e}
\title{
Identifying Causal Effects using Instrumental Time Series: Nuisance IV and Correcting for the Past}
\editor{}
\author{Nikolaj Thams \email thams@math.ku.dk \AND Rikke Søndergaard \email rsn@math.ku.dk \AND Sebastian Weichwald \email sweichwald@math.ku.dk \AND Jonas Peters \email jonas.peters@math.ku.dk \AND 
        \addr{Department of Mathematical Sciences} \\ \addr{University of Copenhagen} \\ \addr{Denmark}}
\date{\today}

\usepackage[noabbrev,capitalize]{cleveref}

\crefname{assumption}{}{}%
\Crefname{assumption}{}{}%

\newtheorem*{assumption*}{Assumption} %

\newcommand{\creflastconjunction}{, and~}

\setlist[enumerate,1]{label={\roman*)}}

\newlist{assumpenum}{enumerate}{1} %
\setlist[assumpenum]{label={(A\arabic*)},ref=(A\arabic*)}
\crefalias{assumpenumi}{assumption} 

\newlist{alt-assumpenum}{enumerate}{1} %
\setlist[alt-assumpenum]{label={(A\arabic*')},ref=(A\arabic*')}
\crefalias{alt-assumpenumi}{assumption} 

\newlist{alt-alt-assumpenum}{enumerate}{1} %
\setlist[alt-alt-assumpenum]{label={(A\arabic**)},ref=(A\arabic**)}
\crefalias{alt-alt-assumpenumi}{assumption}

\crefname{requirement}{requirement}{requirements}
\Crefname{requirement}{Requirement}{Requirements}

\newlist{civenum}{enumerate}{1} %
\setlist[civenum]{label={(CIV\arabic*)},ref=(CIV\arabic*),leftmargin=\widthof{(CIV3)}+\labelsep}
\crefalias{civenumi}{requirement} 

\newlist{alt-civenum}{enumerate}{1} %
\setlist[alt-civenum]{label={(CIV\arabic*')},ref=(CIV\arabic*'),leftmargin=\widthof{(CIV3)}+\labelsep}
\crefalias{alt-civenumi}{requirement} 

\crefrangelabelformat{assumption}{#3#1#4--#5#2#6}

\newcommand{\R}{\mathbb{R}}
\newcommand{\N}{\mathbb{N}}
\newcommand{\Z}{\mathbb{Z}}
\newcommand{\X}{{\mathbf X}} %
\renewcommand{\S}{{\mathbf S}} %
\newcommand{\cX}{\mathcal{X}}
\newcommand{\cY}{\mathcal{Y}}
\newcommand{\cI}{\mathcal{I}}
\newcommand{\cB}{\mathcal{B}}
\newcommand{\cZ}{\mathcal{Z}}
\newcommand{\bX}{\mathbf{X}}
\newcommand{\bY}{\mathbf{Y}}
\newcommand{\bI}{\mathbf{I}}
\newcommand{\bB}{\mathbf{B}}
\newcommand{\bZ}{\mathbf{Z}}
\newcommand{\br}{\mathbf{r}}
\newcommand{\G}{\mathcal{G}}
\newcommand{\Gfull}{\mathcal{G}_{\text{full}}}

\DeclareMathOperator{\CIV}{CIV}
\DeclareMathOperator{\IV}{IV}
\DeclareMathOperator{\NIV}{NIV}
\newcommand{\iv}[2]{\IV_{#1\rightarrow #2}}
\newcommand{\civ}[2]{\CIV_{#1\rightarrow #2}}
\newcommand{\niv}[2]{\NIV_{#1\rightarrow #2}}

\newcommand{\ind}[1]{\mathbbm{1}_{#1}} %
\newcommand{\E}{\mathbb{E}}
\newcommand{\mat}[1]{\begin{pmatrix} #1 \end{pmatrix}}
\newcommand{\convP}{\stackrel{P}{\longrightarrow}}
\renewcommand{\vec}[1]{\left[ #1 \right]}
\newcommand{\vecin}[1]{[ #1 ]}
\newcommand{\cor}[1]{\langle #1 \rangle}
\newcommand{\indep}{\mbox{${}\perp\mkern-11mu\perp{}$}}
\newcommand{\notindep}{\centernot{\indep}\hspace{-2pt}}
\renewcommand{\epsilon}{\varepsilon}
\DeclareMathOperator*{\cov}{\operatorname{cov}}
\DeclareMathOperator*{\var}{\operatorname{var}}
\DeclareMathOperator*{\argmin}{{\arg\min}}
\DeclareMathOperator*{\rank}{\operatorname{rank}}
\DeclareMathOperator*{\diag}{{\operatorname{diag}}}
\DeclareMathOperator{\HAC}{HAC}
\DeclareMathOperator{\ID}{Id}

\definecolor{MI_red}{HTML}{FC4F2C}
\definecolor{MXY_blue}{HTML}{135BB3}
\definecolor{highlight_green}{HTML}{429D6F}
\definecolor{highlight_yellow}{HTML}{E9B444}

\newcommand{\PA}{\operatorname{PA}}
\newcommand{\AN}{\operatorname{AN}}
\newcommand{\ND}[1]{\operatorname{ND}(#1)}
\newcommand{\DE}[1]{\operatorname{DE}(#1)}

\makeatletter
\newcommand{\mylabel}[2]{#2\def\@currentlabel{#2}\label{#1}}
\makeatother 
\hypersetup{
    colorlinks,
    linkcolor=[HTML]{449B52},
    citecolor=[HTML]{67001F},
    urlcolor=[HTML]{008080}
}

\begin{document}

\maketitle
\begin{abstract}
    Instrumental variable (IV) regression relies on instruments to infer causal effects from observational data with unobserved confounding. We consider IV regression in time series models, such as vector auto-regressive (VAR) processes. Direct applications of i.i.d.\ techniques are generally inconsistent as they do not correctly adjust for dependencies in the past. In this paper, we {outline the difficulties that arise due to time structure and} propose methodology for constructing identifying equations that can be used for consistent parametric estimation of causal effects in time series data. 
    One method uses extra nuisance covariates to obtain identifiability (an idea that can be of interest even in the i.i.d.\ case).
    We further propose a graph marginalization framework that allows us to apply nuisance {IV} and other IV methods in a principled way to time series. Our methods make use of a version of the global Markov property, which we prove holds for VAR{(p)}  processes. For VAR(1) processes, we prove identifiability conditions that relate to Jordan forms and are different from the well-known rank conditions in the i.i.d.\ case (they do not require as many instruments as covariates, for example). We provide methods, prove their consistency, and show how the inferred causal effect can be used for distribution generalization. Simulation experiments corroborate our theoretical results. We provide ready-to-use Python code.
\end{abstract}
\begin{keywords}
    causality, time series, instrumental variables, VAR processes, Markov property, distribution generalization 
\end{keywords}

\section{Introduction}\label{sec:intro}
Predicting a response variable $Y$ from observations of covariates $X$ may be insufficient to answer a scientific question at hand. 
Instead, we may wish to model how the response variable $Y$ reacts to an intervention on $X$.
Such modeling requires causal knowledge.
For example, for i.i.d.\ data from a linear model 
$Y := \beta X + g(H, \epsilon^Y)$, it is well-known that an ordinary least squares (OLS) regression of Y on X generally yields a biased estimator of the linear causal effect $\beta$ from $X$ on $Y$ when an unobserved {variable} $H$ confounds $X$ and $Y$.
Instead, we may obtain unbiased estimates of $\beta$ by utilising instrumental variables (IVs) $I$ that correlate with the covariates $X$, are independent of $H$, and affect $Y$ only indirectly through $X$. 
IV regression, pioneered by \citet{wright1928tariff} and \citet{reiersol1945confluence}, is well-established in econometrics \citep{angrist1996identification,staiger1997instrumental,angrist2001instrumental}, statistics 
\citep{bowden1990instrumental} and epidemiology \citep{hernan2006instruments,didelez2010assumptions}. 
One approach for IV estimation in the linear i.i.d.\ model is the two-stage least squares (TSLS) estimator \citep{angrist1995two}, 
which first estimates the effect from $I$ to $X$ (stage 1) and then regresses $Y$ on the fitted values from the first regression (stage 2). 
Another formulation, used by \citet{hansen1982large}, is the generalized method of moments (GMM), which uses the independence of the residual $Y - \beta X = g(H, \epsilon^Y)$ from the instrument $I$:
One can estimate $\beta$ by selecting $\hat{\beta}$ such that the empirical correlation between $Y - \hat{\beta} X$ and $I$ is minimized. If the dimension of $I$ is greater than or equal to the dimension of $X$, these estimators are consistent \citep[e.g.,][]{hall2005generalized}.

In more recent approaches, causality and directed acyclic graph (DAG) representations have proved fruitful for studying instrumental variables for i.i.d.\ data \citep{pearl2009causality, hernan2006instruments,didelez2010assumptions}. \citet{brito2002generalized} proposed `generalized IV', a graphical framework that 
enlarges the class of graphical models, in which IV methods can be used to identify causal effects. Similarly, `conditional IV' \citep{pearl2009causality} relaxes the assumptions of IV by considering a conditional moment equation \citep[see also][]{henckel2021graphical}. 
Moreover, \citet{kang2016instrumental} demonstrate that identification and consistent estimation are possible when at least half of the instruments are valid, even without knowing which ones are invalid.

In many real-world applications (see \citet{weigend2018time} for examples from various fields), the data are sampled not independently but rather as a time series that exhibits memory effects, with past values affecting present ones.
For example, 
both price and demand on an electricity market are confounded by several factors, which makes estimation of price 
elasticity of demand (that is, the response in the demand to changes in price) 
difficult to identify.
In this example, wind power penetration of the market may act as instrumental variable, since it affects the energy supply and in turn quantity purchased at that price \citep{neamtu2016wind, Hirth2022very}; 
it can be justified as a valid instrument since the daily amount of generated wind power depends on external weather conditions, but not other aspects of the market supply or demand.
This is similar to how yield per acre acted as 
`curve shifter'
in the seminal work of \citet{wright1928tariff}
for estimating the price elasticity of demand for flaxseed.

Using IV methods in time series data poses a number of challenges {that are not present when considering IV methods for i.i.d.\ data}. For example, memory effects in the observed processes $X$, $Y$ and $I$ can obfuscate the assumption that $I$ only affects present values of $Y$ through the present value of $X$, because $I$ and $Y$ are confounded by common ancestors in the memory of the process.
{Later in this introduction, we show with a concrete example how these challenges arise in linear models.}
Additionally, memory effects, or serially correlated errors, in the confounder process $H$ can make identification of the dependence on past states of the process difficult; 
for such settings, \citet{fair1970estimation} proposes a search-based method for a subclass of first order vector auto-regressive (VAR) processes. If one is provided with identifying equations with serially correlated errors (such as the ones proposed in this paper), \citet{newey1987simple} construct confidence intervals by using heteroskedasticity and auto-correlation consistent (HAC) estimators to estimate long-run covariance matrices. 

{
To overcome the challenges for IV methods in time series data,
we establish a link between graphical models and IV methods for time series, which we then exploit to construct estimators and prove consistency. To help build this connection, we require that the global Markov property \citep{lauritzen1996graphical} holds in VAR($p$) processes. 
We prove this statement as \cref{thm:gmp}.
To the best of our knowledge this result 
is more general than existing results \citep[e.g.,][]{dahlhaus2003, lauritzen1996graphical}; its proof requires technical arguments taking into account that the graphs contain infinitely many nodes.}

{Causal inference on time series data has been considered before. E.g., causal inference can then be done using the principle of 
Granger causality \citep[e.g.,][]{Wiener1956,Granger1969, Granger1980,dahlhaus2003, Eichler2007, didelez2010assumptions} 
but such methods usually fail when some of the variables are unobserved \citep[e.g.,][Chapter~10.3]{peters2017elements}. 
It is also possible to extend independence-based methods
\citep{spirtes2000causation, pearl2009causality} 
to time series 
\citep{DemiralpHoover,Chu2008, Entner2010, moneta11}
but as in the i.i.d.\ case, these methods cannot exploit identifiability that stems from IV conditions. 
Adjustment formulas (and its modifications such as the front-door criterion) have been extended to time series with known graph structure \citep{eichlerdidelez} but the specific form of confounding 
assumed in our setting
does not allow for consistent estimation of the causal effect using these techniques -- again, this is similar to the i.i.d.\ case.
{\citet{michael2020instrumental} use instrumental variables for longitudinal data; different from our work they consider binary instruments and primarily consider a setting where time series are observed multiple times (for example each corresponding to a patient), though the marginal structural mean model \citep{robins1997causal} that they employ also allows  for estimation within a single observation of a time-series. 
Their target of inference differs to ours in that \citet{michael2020instrumental} consider interventions on the whole treatment time series. 
The structural assumptions (such as independence of instrument and confounder) underlying their method are similar to those required by the estimators we develop below, though generally not identical.

Recently, \citet{mogensen2022instrumental} uses integrated covariances to conduct instrumental variable estimation of `normalized parameters', and shows that these parameters have causal interpretations in both discrete time series and continuous time Hawkes processes. 
}

Throughout this work we consider a joint process $S := \vec{I_t^\top, X_t^\top, H_t^\top, Y_t^\top}^\top_{t\in\Z}$ where $H$ is latent. 
{As a motivating example, let $S$ be a linear VAR($1$) process (though we generally consider VAR($p$) processes), with dependencies represented in \cref{fig:sub3}, and assume that we observe a finite subsample $S_1, \ldots, S_T$ from a single instantiation of this process.
Our goal is to estimate the coefficient $\beta$ with which $X_{t-1}$ linearly enters into $Y_t$. 
When $S$ is fully observed, estimators that are consistent and asymptotically normal for the standard form parameters of the VAR($1$) process $S$ exist \citep[e.g.,][]{hamilton1994time}. Yet, in our setting, $H$, the confounding variable between $X$ and $Y$, is unobserved, and such estimators are then generally not consistent.
In this work, we seek to overcome this hidden confounding by using $I_{t-2}$ as an instrument. 
Estimation in the time series setting is complicated by memory effects not present in the i.i.d. setting:
For example, the instrument $I_{t-2}$ is correlated to $Y_t$ not only through the path $I_{t-2} \rightarrow X_{t-1} \rightarrow Y_t$ but also through an infinite number of paths in the past, due to common ancestors $I_{t-j}, j \geq 3$ in the instrument process $I$. This correlation violates the assumption that the instrument $I_{t-2}$ only correlates with $Y_t$ through $X_{t-1}$. Similarly, there are also observed confounders of $X_{t-1}$ and $Y_t$, such as $X_{t-2}$ which is a common ancestor of the two. 
}
\begin{figure}[t]
    \centering
    \resizebox{\textwidth}{!}{%
\usetikzlibrary{math} %
\begin{tikzpicture}[>=latex,font=\sffamily]
    \tikzstyle{every path} = [thin, ->];

    \foreach \i in {-1, ..., 4}{
        \foreach \P [count = \j] in {I, filler, X, H, Y}{
            \ifthenelse{\j=2}{}{
            \node (\P\i) at (15 - 3*\i, 4-3/4*\j) {\ifthenelse{\i=-1 \OR \i=4}{}{
            $\P_{\ifthenelse{\i=0}{t}{t-\i}}$
            }}};
        }
    }

    \foreach \i [count = \j from 0] in {-1, ..., 3}{

        \foreach \P in {I, X, Y}{
            \path (\P\j) edge[opacity={(\j==4||\j==0) ? 0.7 : 1}] (\P\i);
        }
        \path[gray] (H\j) edge[opacity={(\j==4||\j==0) ? 0.7 : 0.9}, dashed] (H\i);

        \path (I\j) edge[opacity={(\j==4||\j==0) ? 0.7 : 1}] (X\i);
        \path[highlight_green] (X\j) edge[line width=1.15pt,opacity={(\j==4||\j==0) ? 0.7 : 1}] node[above]{$\beta$} (Y\i);
        \path[gray] (H\j) edge[opacity={(\j==4||\j==0) ? 0.7 : 0.9}, dashed] (X\i);
        \path[gray] (H\j) edge[opacity={(\j==4||\j==0) ? 0.7 : 0.9}, dashed] (Y\i);
        \path (Y\j) edge[opacity={(\j==4||\j==0) ? 0.7 : 1}] (X\i);
    }
    
\end{tikzpicture}
     }%
    \caption{{Finite excerpt of a} full time graph (formally defined in \cref{sec:graph_representation}) of a {VAR($1$)} process $S$ satisfying \cref{assump:iv} below. 
    {We assume that each time point is observed only once.}
    {Our methodology} estimate{s} the causal effect $\beta$ (highlighted in green) of $X_{t-1}$ on $Y_t$, where $X = [X_t]_{t\in\Z}$ and $Y=[Y_t]_{t\in\Z}$ are subprocesses of $S$, that are confounded by a latent process $H = [H_t]_{t\in\Z}$ {(so that, in general, standard regression yields an inconsistent estimator)}. 
    Motivated by instrumental variables, one may aim to exploit the subprocess $I = [I_t]_{t\in\Z}$ that is independent of $H$ and only acts on $Y$ through $X$. However, simply using $I_{t-2}$ as an instrument is generally inconsistent; the same holds when adding $X_{t-2}$ and $Y_{t-1}$ as a conditioning set, for example. This paper develops a graphical framework giving rise to several estimators {which consistently estimate $\beta$.} 
    }
    \label{fig:sub3}
\end{figure}

{
To obtain valid identifying IV methods for VAR($p$) processes $S$, such as the one shown in \cref{fig:sub3}, 
we establish a general graph marginalization technique
allowing us to read off the relevant separation statements from the reduced graphs.
}
Based on these results, we outline two solutions {that} identify {total} causal effects in {the considered} time series. The first solution ({which is based on} `conditional IV' or `CIV') identifies $\beta$ using IV conditioned on
{some subset of past states of the time series.} {In the example above, we will see that, e.g., $I_{t-2}\indep Y_t-\beta X_{t-1} | \{I_{t-3}, X_{t-2}, Y_{t-1}\}$ is an identifying equation.}

The second solution is based on
{a modification of IV that} can be used not only for time series but also for i.i.d.\ data. 
{It adds nuisance treatment variables to the target causal effect and thereby allows for stronger identifiability results. This is a straightforward idea but we are not aware that this has been discussed explicitly, so we suggest to call it `nuisance IV' (or `NIV').}
Applied to the time series setting, nuisance IV yields a consistent estimator for the target of inference $\beta$ by 
{including nuisance regressors, such as $Y_{t-1}$, that also affect $Y_t$.
In the example above, it yields the identifying equation $\{I_{t-2}, \ldots, I_{t-m-1}\} \indep Y_t - \beta X_{t-1} - \kappa Y_{t-1}$ for some $m$. As for the first solution, we will detail the reasoning behind these equations and conditions for identifiability.} 

Similar to the i.i.d.\ case, these two approaches 
induce identifying moment equations that are satisfied by the {total} causal effect $\beta$. Rank conditions guarantee that their solution is unique, allowing us to identify the causal effect. 
Unlike in the i.i.d.\ case, however, the standard conditions are not easily interpretable in the time series setting. 
We therefore develop sufficient and necessary conditions on the parameters of the data-generating process that provide insight on when identifiability holds. 
Our results imply that identifiability with nuisance IV depends on geometric multiplicities of eigenvalues of the parameter matrix in the VAR{($1$)} process, and we show that if parameters are drawn at random from a continuous distribution, the causal effect $\beta$ of $X_{t-1}$ on $Y_t$ is almost surely identifiable.
In particular, it is possible to identify the causal effect even if the instrument $I$ is univariate and the regressor $X$ is multivariate.
For both of the approaches (conditional IV and nuisance IV), we propose estimators and prove that, in case of identifiability, these estimators consistently estimate the causal effect.

Finally, we apply our findings to the task of distribution generalization \citep[e.g.,][]{Christiansen2020DG,jakobsen2020distributional,
meinshausen2018causality,rothenhausler2018anchor}.
In many systems, the causal effects are of value in themselves because they contribute to the understanding of the system but it also serves a purpose when predicting $Y_{t+1}$ under an intervention on $X_t$. 
In a linear setting, the OLS estimator has the smallest expected mean squared error (MSE) among all linear predictors when predicting new test data from the observational distribution. However, as is known for the i.i.d.\ setting \citep[e.g.,][]{Rojas2016}, causal estimators can have better worst-case predictive performance when {there may be} interventions on the covariates. Similarly, we show that in time series, under arbitrary interventions on $X_t$, our IV estimators are worst-case prediction optimal for $Y_{t+1}$.

Our work is structured as follows.
\cref{sec:setup} introduces the model and the assumptions considered in this paper; we review graphical representations of time series models and prove that the global Markov property holds for VAR($p$) processes.
In \cref{sec:niv}, we review theory on conditional instrumental variables and introduce the concept of nuisance IV. 
Our main results for instrumental variable regression for time series are presented in \cref{sec:iv-for-time-series}. We propose two approaches to overcome confounding from past values yielding identifying equations for the {total} causal effect: the first one is based on CIV and the second one uses NIV. For the latter, we characterize identifiability of the causal parameter in terms of parameters of the data-generating process. 
We also discuss how to use the causal effect to perform optimal prediction of $Y_{t+1}$ under interventions on $X_t$. In \cref{sec:simulations} we empirically evaluate our method. All proofs are provided in \cref{sec:all_proofs}.
Code can be found at~\url{https://github.com/nikolajthams/its-time}.

\section{Causal Time Series Models with Confounding}\label{sec:setup}
\subsection{Definitions and Notation}\label{sec:def_notation}
We consider multivariate time series $X\coloneqq [X_t]_{t\in\Z}$ and $I \coloneqq [I_t]_{t\in\Z}$, 
a univariate process $Y\coloneqq [Y_t]_{t\in\Z}$,
and an unobserved multivariate process $H \coloneqq [H_t]_{t\in \Z}$. Let $d_X$ be the dimensionality of $X_t$, that is $X_t\in \R^{d_X}$, and similarly for $d_I$, $d_Y$ and $d_H$, with $d_Y = 1$. Let $S \coloneqq [S_t]_{t\in\Z} = [I_t^\top, H_t^\top, X_t^\top, Y_t^\top]^\top_{t\in\Z}$, $S_t \in \R^{d}$, with $d := d_X + d_Y + d_I + d_H$.

{In general,} our results are presented for VAR($p$) processes \citep[e.g.,][]{Brockwell1991}, {though we apply our theory to the particular example of a VAR($1$) process in \cref{se:meth1,se:meth2,sec:prediction_intervention}}. Many of the results hold more generally {than in VAR($p$) processes, which we discuss in} \cref{sec:relaxing-varp}.
Given a (known) $p \in \N$, we say that a (weakly) stationary process $S$ is a VAR($p$) process if the following assumptions hold: 
\begin{assumpenum}
\addtocounter{assumpenumi}{-1} %
\item\label{assump:varp-woutGaussian}  There are coefficient matrices $A_1, \ldots, A_p \in \R^{d\times d}$ 
    such that for all $t \in \Z$:
    \begin{equation}
            S_t = A_1 S_{t-1} + \ldots + A_p S_{t-p} + \epsilon_t, \label{eq:varp}
    \end{equation}
    where $A_1, \ldots, A_p$ are such that\footnote{A VAR($p$) process satisfying this condition is sometimes referred to as a causal VAR($p$) process \citep{Brockwell1991}; \citet{peters2013causal} discuss a relation between this causality property and the independent noise assumption in SCMs.}
    $\mathrm{det}(I_d\lambda^p-A_1\lambda^{p-1}-A_2\lambda^{p-2}-...-A_p)=0$ implies $|\lambda|<1$ and 
    the $\epsilon_t$ constitute an i.i.d.\ process with finite second moments.
    \item\label{assump:varp}  \ref{assump:varp-woutGaussian} is satisfied and, in addition $\epsilon_t \sim \mathcal{N}(0, \Gamma)$, where $\Gamma$ is a diagonal matrix. 
\end{assumpenum}
(That is, unless stated otherwise, we assume that the VAR processes in this paper are Gaussian.) We use the notation $\alpha_{X, I}^1$ to refer to the submatrix of $A_1$ with rows corresponding to $X$ and columns corresponding to $I$ (see \cref{fig:sub1} for an example), and similarly $\alpha_{I,I}^2$ etc.
We use superscripts to denote individual components of $\varepsilon$, e.g., $\varepsilon^Y$.
{
We consider $Y_t$ as the \emph{response variable}, and aim to estimate the total causal effect (TCE) of covariates on $Y_t$. 
For  $1\leq i, j\leq d$ and $l\in\N$, the TCE of $S^i_{t-l}$ on $S^j_t$ is defined as
\begin{equation*}
    \beta\coloneqq\bigg(\sum_{\substack{1 \leq l_1, \ldots, l_m \leq p \\ l_1 + \cdots + l_m = l}} A_{l_1}\cdots A_{l_m}\bigg)_{j,i};
\end{equation*}
see \cref{sec:appendix-TCE} for more details on defining the TCE.
An important example is when we use $X_{t-1}$ as covariates; in this case, the TCE equals $\beta = \alpha^1_{Y,X}$, which is also called the \emph{direct causal effect} of $X_{t-1}$ on $Y_t$. Unless specified otherwise, all causal effects in this paper refer to total causal effects.\footnote{The notions of causal effect and total causal effect are motivated by interpreting the VAR{($p$)} equations as a structural causal model (SCM), which we explain in detail in \cref{sec:interventions}. The interventional
interpretation of an SCM is not required for any results of the paper, except for the ones presented in \cref{sec:prediction_intervention}, where we discuss optimal predictions under interventions. {Considering the model as an interventional model implies that $H$ is the only unobserved process of relevance \citep[see `interventional sufficiency'][Chapter~9]{peters2017elements}.}}
}

Both $H$ and the noise $\epsilon$ are assumed to be unobserved; while the sequence of innovations $\epsilon^Y_t$ is assumed to be i.i.d.\ and independent of $X_t$, $H$ can act as a confounder between $X$ and $Y$ and can have an autoregressive structure. 
Similar to the i.i.d.\ case \citep[e.g.,][]{hernan2006instruments,pearl2009causality,peters2017elements}, the existence of the confounder $H$ implies that we cannot identify $\beta$ by simply regressing $Y_t$ on $X_{t-1}$. 
In \cref{sec:obs_equiv}, we provide an example of two VAR($1$) processes with two different parameter matrices that generate the same distribution over the observed time series. 

We assume that the process has zero mean\footnote{Since we can always subtract empirical means, the assumption of vanishing means does not come with any loss of generality.}, so for instance $\operatorname{cov}(X_t, Y_t) = \E\{X_tY_t^\top\} \in \R^{d_X \times d_Y}$.
We assume that the data are sampled as follows: We obtain a sample over time points $t = 1, \ldots, T$ of $S$ such that for $t=1$, $S_1$ follows the stationary distribution.
We denote the sample with boldface $\mathbf{S} = [\mathbf{S}_t]_{t=1}^T$ such that $\mathbf{S} \in \R^{d \times T}$ and each column $\mathbf{S}_t$ represents the process observation at time $t$.
We let $\hat\E \mathbf{S}_t \coloneqq \frac{1}{T}\sum_{t=1}^T \mathbf{S}_t$
denote the empirical mean of the {process} (here, the index $t$ in $\hat\E \mathbf{S}_t$ does not refer to any specific time point). 
From \cref{assump:varp} it follows \citep[Chapter~10]{hamilton1994time} that {for $T\to\infty$}
\begin{equation}\label{eq:assump_moments}
            \hat{\E} \mathbf{S}_t \stackrel{P}{\longrightarrow}
            \E S_t
            \quad \text{and}\quad
            \hat{\cov}\{\mathbf{S}_{t-j}, \mathbf{S}_t\} \coloneqq \frac{1}{T-j} \sum_{t=j+1}^T\mathbf{S}_{t-j}\mathbf{S}_t^{\top}
            \stackrel{P}{\longrightarrow}
            \cov\{S_{t-j}, S_t\}.
\end{equation}
{
Even though we assume that only a finite time window $\mathbf{S}$ is observed, we assume for simplicity that the data are generated (but not observed) over all time points $t\in\Z$, to ensure that subprocesses of $S$ are correlated in the same way throughout time. The same assumption is common in the literature \citep[e.g.][]{Brockwell1991,hamilton1994time}.

}
Finally, in the case where $S$ is a VAR($1$) process, we will sometimes assume additional structure on the coefficient matrix. 
\begin{assumpenum}[resume]
    \item\label{assump:iv} Assume that $S$ satisfies \cref{assump:varp} for $p=1$ and that $A_1$ has the sparsity structure displayed in \cref{fig:sub1}.
\end{assumpenum}
Under \cref{assump:iv}, none of the other time series components enters the assignment for $I$ {and $I$ itself neither appears in the assignments for $Y$ nor in the one of $H$; the assignment for $Y$ reduces to 
\begin{equation*}
    Y_t = \alpha_{Y,Y} Y_{t-1}
    + \beta X_{t-1}
    + \eta_t
\end{equation*}
with 
$\eta_t := \alpha_{Y,H} H_{t-1}
+ \varepsilon^{Y}_t$. } We refer to $I$ as an \emph{instrumental time series}.
{
As in the i.i.d.\ case, \cref{assump:iv} cannot in general be tested from data and must therefore rest on background knowledge.}

\begin{figure}[t]
\centering
\begin{subfigure}[t]{.47\textwidth}
\centering
\begin{tikzpicture}
    \node at (0,0) {
$A_1 = \; \begin{matrix} \text{{\tiny I}} \\ \text{{\tiny H}} \\ \text{{\tiny X}} \\ \text{{\tiny Y}}\end{matrix}
    \mat{\alpha_{I,I} & 0 & 0 & 0 \\
	    0 & \alpha_{H,H} & 0 & 0 \\
	 \alpha_{X,I} & \alpha_{X,H} & \alpha_{X,X} & \alpha_{X,Y}\\
	 0 & \alpha_{Y,H} & \beta & \alpha_{Y,Y}
	 }$
	 };
    \end{tikzpicture}
    \caption{
    Matrix block structure of $A_1$ assumed in \cref{assump:iv}. 
    }
    \label{fig:sub1}
\end{subfigure}\quad%
\begin{subfigure}[t]{.47\textwidth}
  \centering
  		\begin{tikzpicture}[>=latex,font=\sffamily]
			\node[draw, circle] (I) at (-1.5, 0) {$I$};
		    \node[draw, circle] (X) at (0,0) {$X$};
		    \node[draw, circle, gray, dashed] (H) at (1.5, 1.5) {$H$};
		    \node[draw, circle] (Y) at (3, 0) {$Y$};
		    \path [->] (I) edge node[below]{$\alpha_{X,I}$} (X);
		    \path [->, dashed, gray] (H) edge node[above left, gray]{$\alpha_{X,H}$} (X);		  
		    
		    \path [->, dashed, gray] (H) edge node[above right, gray]{$\alpha_{Y,H}$} (Y);		    
		    \path [->] (X) edge[bend right=10] node[below]{$\beta$} (Y);
		    \path [->] (Y) edge[bend right=10] node[above]{$\alpha_{X,Y}$} (X);
		    \path [->] (X) edge [loop below] node[below]{$\alpha_{X,X}$} (X);
		    \path [->] (Y) edge [loop below] node[below]{$\alpha_{Y,Y}$} (Y);
		    \path [->] (I) edge [loop below] node[below]{$\alpha_{I,I}$} (I);
		    \path [->, dashed, gray] (H) edge [loop left] node[left, gray]{$\alpha_{H,H}$} (H);
		\end{tikzpicture}   \caption{Summary time graph of a VAR($1$) process satisfying \cref{assump:iv}.}
  \label{fig:sub2}
\end{subfigure}%
\caption{The sparsity structure on the parameter matrix $A_1$ assumed in \cref{assump:iv}, and a representation of the graphical structure induced by $A_1$.
Zeros in panel a) correspond to absent edges in panel b).
}
\label{fig:representation_A}
\end{figure}
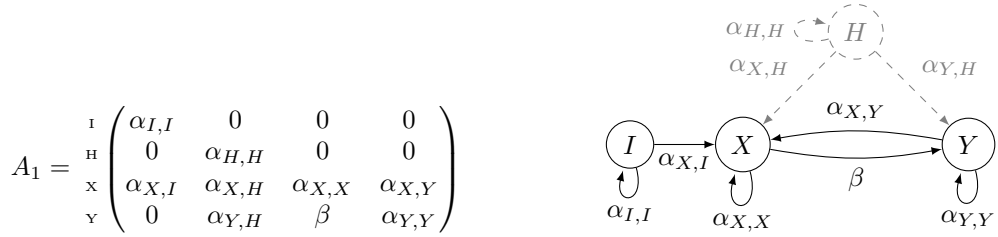

\subsection{Graph Representations of VAR\texorpdfstring{{($p$)}}{} Processes}\label{sec:graph_representation}
\cref{eq:varp} can be represented {by a directed graph. We choose a representation with infinitely many nodes that follows directly from the structural coefficients and is similar or identical to other representations that have been discussed in the literature \citep[see, e.g.,][]{dahlhaus2003, peters2013causal}. The graph} will prove helpful when establishing identifying equations for causal effects and constructing consistent estimators. 
The \emph{full time graph} \citep[e.g.,][]{peters2013causal} is defined as an infinite directed graph with nodes $I_t$, $H_t$, $X_t$, $Y_t$, for any $t \in \Z$.
For $k\in\N$, it contains a directed edge from $(j, t)$ to $(i, t+k)$, $i, j \in \{I, H, X, Y\}$, if $\alpha^k_{ij} \neq 0$.
For a {VAR($1$)} process satisfying \cref{assump:iv} an extract of this graph is shown in \cref{fig:sub3}. We define the full time graph for {higher-order} VAR{($p$)} processes accordingly. 
The \emph{summary time graph} only has a single node per time series component. It contains a directed edge from $i$ to $j$, for some $i, j \in \{I, H, X, Y\}$ if and only if the full time graph contains an edge from $(j, t)$ to $(i, t+k)$ for some $k \in \N$.
For a {VAR($1$)} process satisfying \cref{assump:iv}, such a graph is visualized in \cref{fig:sub2}. {(The difference between the full time graph and a summary time graph is also sometimes referred to as `unrolling' or `unfolding'.)}
Neither of these two graphical representations requires Gaussianity (that is, the constructions only require~\cref{assump:varp-woutGaussian}).
We now introduce some standard graph terminology 
\citep[e.g.,][]{lauritzen1996graphical,koller2009probabilistic,pearl2009causality}.
A \emph{path} $p$, $p = (v_1, e_1, v_2, \ldots, e_{n-1}, v_n)$, is an alternating sequence of distinct vertices $v_i$ and edges $e_i$ such that $v_i$ and $v_{i+1}$ are connected by $e_i$.
We say that $p$ is a \emph{directed path} from $v_1$ to $v_n$ if for every $i$, $e_i$ points from $v_i$ to $v_{i+1}$.
For two nodes $v$ and $u$, we say that $u$ is a \emph{descendant} of $v$ if there exists a directed path from $v$ to $u$ and otherwise $u$ is a \emph{non-descendant} of $v$.
We write $\ND{v}$ and $\DE{v}$ for the sets of non-descendants and descendants of $v$, respectively, using the convention that neither of them contain $v$ itself.
For a path $p$ and any $i \in \{2, \ldots, n-1\}$, we say that $v_i$ is a \emph{collider on $p$} if $(v_{i-1}, e_{i-1}, v_i, e_i, v_{i+1})$ is of the form $v_{i-1}\rightarrow v_i \leftarrow v_{i+1}$ and else $v_i$ is a \emph{non-collider on $p$}. 
We say that the path $p$ is \emph{unblocked, given the set $B$} if for every non-collider $v_i$ in $p$, $v_i \notin B$ and for every collider $v_i$ on $p$, $(v_i \cup \DE{v_i})\cap B \neq \emptyset$. Otherwise, we say that $p$ is \emph{blocked} by $B$. If all paths between distinct vertices $v$ and $u$ are blocked by a set $B$ neither containing $v$ nor $u$, we say that $v$ and $u$ are \emph{$d$-separated} by $B$. 
Similarly, we say that disjoint sets $V$ and $U$ are $d$-separated by $B$ if all nodes $v\in V$ and $u\in U$ are $d$-separated by $B$.

In \cref{sec:iv-civ-var}, we also consider marginalized graphs, which are acylic directed mixed graphs (ADMGs), containing both directed ($\rightarrow$) and bidirected ($\leftrightarrow$) edges. If we define $v$ to be a collider on a path whenever two surrounding edges have arrowheads at $v$ (e.g. $u_1\leftrightarrow v \leftarrow u_2$), and define descendants only with respect to directed edges, $d$-separation also extends to ADMGs. See \citet{richardson2003markov} for details.

\subsubsection{Markov Properties of VAR\texorpdfstring{{($p$)}}{} processes}
The representation described above satisfies a Markov property, which enables us to read off conditional independences from the full time graph. This will be an important tool in our theory (and will be used in many of the proofs), because it enables the use of graphical models to develop IV methodology in time series. 
{
Markov properties of time series have been discussed before
\citep[e.g.,][]{dahlhaus2003, EichlerPTRF} 
but to the best of our knowledge, the statement of Theorem~\ref{thm:gmp}
has not been proved before, so we add a proof in this paper.
For example, some of the results in the above works consider finite graphs while the statement of Theorem~\ref{thm:gmp} 
is about graphs with infinitely many nodes.
For this reason, the result does not formally follow from standard results in, for example, \citet{lauritzen1996graphical}, either.
\begin{restatable}{theorem}{GMP}
\label{thm:gmp}
    Consider $p \in \N$,  a (weakly) stationary time series $S$ satisfying \cref{assump:varp-woutGaussian} for that $p$, and 
    consider disjoint, finite sets
    $A, B, C$ of nodes in the full time graph $\Gfull$. 
    If $A$ and $C$ are $d$-separated given $B$ in $\Gfull$, then $A \indep C | B$.
\end{restatable}
The proof of \cref{thm:gmp} 
is inspired by \citet[][Sec.\ 6]{lauritzen1990independence}.
As all other proofs in this paper, it can be found in~\cref{sec:all_proofs}. 
}

\section{Nuisance Effects in Instrumental Variable Regression}\label{sec:niv}
In this work, we establish two identifying equations for causal effect estimation in time series:
The first one is based on conditional instrumental variables (CIV) and the second one on a generalization that we term nuisance instrumental variables (NIV). 
We regard the idea of NIV as interesting in its own right, as it can be applied in the i.i.d.\ setting, too.
In this section, we therefore first review CIV regression for i.i.d.\ data, and then introduce NIV regression; instrumental variable (IV) regression is a special case of CIV regression, where the conditioning set is empty. 
In \cref{sec:iv-for-time-series}, we extend the CIV and NIV estimators to VAR{($p$)} processes via a reduction of the full time graph to a marginalized graph.

\subsection{Instrumental Variables and Conditional Instrumental Variables}\label{sec:intro-civ}
Consider a linear SCM (see \cref{sec:interventions}) over variables $V$, and let $\cI, \cX, \cB, \{Y\} \subseteq V$ be disjoint collections of variables\footnote{
    Below, the different variables will take different roles (such as instruments or regressors). We use the calligraphic notation $\cI, \cX,$ and $\cB$ to denote collections of observed variables, being used as instruments, regressors, and conditioning sets, respectively. Individual variables are denoted by non-calligraphic letters, such as $I$.
}
from $V$, and let $\G$ be the corresponding DAG. Assume that $\cI, \cX$ and $Y$ have zero mean 
and finite second moment and 
let $\beta$ be the causal coefficient with which $\cX$ enters the structural equation for $Y$, that is,
\begin{equation*}
    Y = \beta \cX + \gamma W + \epsilon^Y,
\end{equation*}
for some variables $W \subseteq V\setminus\cX$;
(some of the entries of $\beta$ can be zero, so not all variables in $\cX$ have to be parents of $Y$).
{In this setup, $\beta$ is both the total and the direct causal effect of $\cX$ on $Y$.}
We consider the following three requirements on $\cI, \cX, \cB$ and $Y$:
\begin{civenum}
    \item \label{assump:civ-d-sep} $\cI$ and $Y$ are $d$-separated given $\cB$ in the graph $\G_{\cX\not\rightarrow Y}$, 
    that is the graph $\G$ where all direct edges from $\cX$ to $Y$ are removed,
    \item \label{assump:civ-descendants} $\cB$ does not contain a descendant of 
    $\cX \cup \{Y\}$ in $\G$, and
    \item \label{assump:civ-relevance} the matrix $\E[\cov(\cX, \cI | \cB)]$ has rank $d_\cX$, that is, full row rank.
\end{civenum}
If \cref{assump:civ-d-sep,assump:civ-descendants} are met, $Y - \beta\cX \indep \cI | \cB$, and in particular $\beta$ satisfies the \emph{CIV moment equation}\footnote{
Here we use the definition $\cov(A, C|B) \coloneqq \E[AC^\top|B] - \E[A|B]\E[C^\top|B] = \cov(A - \E[A|B], C - \E[C|B])$, which even accommodates for nonlinear relationships between the variables and $\cB$.
}
\begin{equation}\label{eq:civ-moment-condition}
    \E[\cov(Y - \beta \cX, \cI|\cB)] = 0.
\end{equation}
If, additionally, \cref{assump:civ-relevance} is met,
$\beta$ is the unique solution to \cref{eq:civ-moment-condition},
\begin{equation*}
    \E[\cov(Y - b \cX, \cI|\cB)] = 0 \implies b = \beta.
\end{equation*}
(Conditional IV with univariate $\cX$ has been discussed in the literature \citep{pearl2009causality, henckel2021graphical,brito2002generalized}.
Since we allow $d_{\cX}>1$, we add a short proof in \cref{app:proofciviid}.)
In this case, we say that $\beta$ is \emph{identified by CIV} or, more precisely, \emph{identified by $\civ{\cX}{Y}(\cI | \cB)$}.
If \cref{assump:civ-relevance,assump:civ-d-sep} are satisfied for $\mathcal{B}=\emptyset$ (\cref{assump:civ-descendants} is trivially satisfied for $\cB = \emptyset$), $\civ{\cX}{Y}(\cI|\emptyset)$ reduces to instrumental variables (IV) regression
\citep{reiersol1945confluence,anderson1949estimation,bowden1990instrumental,angrist1996identification}, which we refer to as $\iv{\cX}{Y}(\cI)$.
We use the term \emph{proper} CIV when $\mathcal{B}\neq \emptyset$.

For a finite sample $\bX$, $\bY$, $\bI$, and $\bB$, we consider an empirical counterpart of \cref{eq:civ-moment-condition} which, 
however, may not have a solution in the overidentified setting, that is when $d_\cI>d_\cX$; 
to overcome this, for any positive definite weight matrix $W$, we define the estimator $\hat b(W)$ as
\begin{equation}\label{eq:civ-argmin}
    \hat b(W) \coloneqq \argmin_b \|\hat\cov(\bY - b\, \bX, \bI|\bB)\|_W^2,
\end{equation} 
where $\|x\|_W^2 := x^\top W x$ and $\hat\cov{(\bY - b\, \bX, \bI|\bB)}$ is the empirical covariance of the residuals after regressing out $\bB$.
We refer to this estimator as $\civ{\X}{\bY}(\bI| \bB)$.
If $\cI, \cX, Y$ and $\cB$ are zero mean random vectors, the minimizer of \cref{eq:civ-argmin} is given by
\begin{align}\label{eq:civ-closed-solution}
	\hat b(W) = \hat\E[r_\bY r_\bI^\top] \, \, W\, \,  \hat\E[r_\bI r_\bX^\top]\bigg(\hat\E[r_\bX r_\bI^\top]\, \, W\, \, \hat\E[r_\bI r_\bX^\top]\bigg)^{-1},
\end{align}
where $r_{\bY} \coloneqq \bY - \hat\E[\bY|\bB]$ are the residuals after regressing $\bY$ on $\bB$, and similarly $r_\bX \coloneqq \bX - \hat\E[\bX|\bB]$ and $r_\bI\coloneqq \bI - \hat\E[\bI|\bB]$. 

Choosing the two-stage least squares (TSLS) weight matrix $W_{\textrm{TSLS}} := \E[r_\bI r_\bI^\top]^{-1}$ corresponds to the procedure where one regresses $r_\bX$ on $r_\bI$ and then returns the regression coefficient of $r_\bY$ on the fitted values $\hat{r}_\bX$. In a linear Gaussian model, the IV estimator $\hat b(W_{\textrm{TSLS}})$ has the lowest asymptotic variance among all positive definite weight matrices $W$ \citep{hall2005generalized} {(see \cref{sec:asymptotic-variance}, where we review the asymptotic theory for IV estimators in i.i.d.\ data)}.

\subsection{Nuisance instrumental variables}\label{sec:nuisance-iv}
\begin{figure}[t]
    \centering
    \begin{subfigure}{0.3\textwidth}
        \begin{tikzpicture}[>=latex, 
                            obs/.style={draw, circle}, 
                            hidden/.style={draw, gray, inner sep=5pt}]
            \node[obs] (I) at (0, 0) {$I$};
            \node[obs] (B) at (1.5, -1) {$B$};
            \node[obs] (X) at (1.5, 0) {$X$};
            \node[obs] (Y) at (3, 0) {$Y$};
            \node[hidden] (H) at (2.25, 1) {$H$};
            \draw [->] (I) -- (X);
            \draw [->] (X) -- (Y);
            \draw [->] (H) -- (X);
            \draw [->] (H) -- (Y);
            \draw [->] (B) -- (I);
            \draw [->] (B) -- (Y);
            \draw [->] (H) to [bend left=20] (B);
        \end{tikzpicture}
        \label{fig:only-civ}
    \end{subfigure}
    \begin{subfigure}{0.3\textwidth}
        \begin{tikzpicture}[>=latex, 
                            obs/.style={draw, circle}, 
                            hidden/.style={draw, gray, inner sep=5pt}]
            \node[obs] (I) at (0, 0) {$I$};
            \node[obs] (Z) at (1.5, -1) {$Z$};
            \node[obs] (X) at (1.5, 0) {$X$};
            \node[obs] (Y) at (3, 0) {$Y$};
            \node[hidden] (H) at (2.25, 1) {$H$};
            \draw [->] (I) -- (X);
            \draw [->] (X) -- (Y);
            \draw [->] (H) -- (X);
            \draw [->] (H) -- (Y);
            \draw [->] (I) -- (Z);
            \draw [->] (Z) -- (Y);
            \draw [->] (H) to [bend left=20] (Z);
        \end{tikzpicture}
        \label{fig:only-iv}
    \end{subfigure}
    \begin{subfigure}{0.3\textwidth}
        \begin{tikzpicture}[>=latex, 
                            obs/.style={draw, circle}, 
                            hidden/.style={draw, gray, inner sep=5pt}]
            \node[obs] (I) at (0, 0) {$I$};
            \node[obs] (B) at (0, -1) {$B$};
            \node[obs] (X) at (1.5, 0) {$X$};
            \node[obs] (Z) at (1.5, -1) {$Z$};
            \node[obs] (Y) at (3, 0) {$Y$};
            \node[hidden] (H) at (2.25, 1) {$H$};
            \draw [->] (I) -- (X);
            \draw [->] (X) -- (Y);
            \draw [->] (H) -- (X);
            \draw [->] (H) -- (Y);
            \draw [->] (B) -- (I);
            \draw [->] (B) -- (Z);
            \draw [->] (B) -- (X);
            \draw [->] (H) to [bend left=20] (Z);
            \draw [->] (Z) to (Y);
        \end{tikzpicture}
        \label{fig:both-works}
    \end{subfigure}
     \caption{
    CIV and (proper) NIV are complementary for identifying causal effects in that they can be used in different settings.
    (\textit{Left}) A graph where the effect $X\rightarrow Y$ {can be} identified by $\civ{X}{Y}(I|B)$ (provided that \cref{assump:civ-relevance} holds) but not by any IV or proper NIV method.
    (\textit{Middle}) A graph where $X\rightarrow Y$ {can be} identified by $\niv{X}{Y}(I, Z)$, but not by any proper CIV method. 
    (\textit{Right}) A graph where $X\rightarrow Y$ {can be} identified by both $\civ{X}{Y}(I|B)$ and $\niv{X}{Y}(\cI, Z)$ with $\cI = \{I, B\}$. 
    }
    \label{fig:3-structures}
\end{figure}
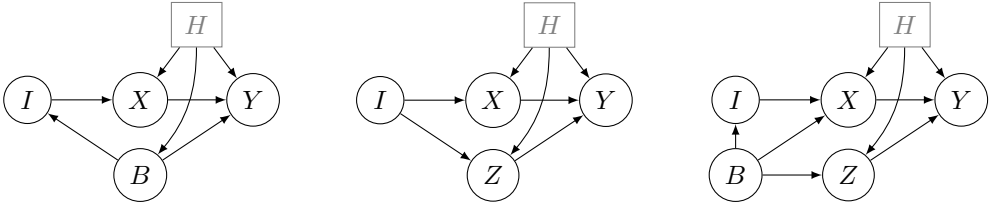
IV estimation is a special case of CIV with the empty set as conditioning set. 
The example in \cref{fig:3-structures} (left) shows a graph where an effect between $X$ and $Y$ cannot be identified using the IV estimator: 
no variable is $d$-separated from $H$, and hence no valid instruments for (unconditional) IV exist. Yet, the causal effect is identified by\footnote{
    In a slight abuse of notation, we sometimes omit parantheses indicating sets and write 
    $\civ{\cX}{Y}(I|\cB)$ instead of 
    $\civ{\cX}{Y}(\cI|\cB)$ if $\cI = \{I\}$, 
    for example.
} 
$\civ{X}{Y}(I| B)$
because $B$ satisfies \cref{assump:civ-relevance,assump:civ-d-sep,assump:civ-descendants}.

In the case shown in \cref{fig:3-structures} (middle),  the effect from $X$ to $Y$ cannot be identified 
by $\iv{X}{Y}(I)$ because of the unblocked path $I \rightarrow Z \rightarrow Y$. 
We cannot use proper CIV, either, because the path $I \rightarrow Z \leftarrow H \rightarrow Y$ is unblocked given $Z$, violating \cref{assump:civ-d-sep}. Nevertheless, we can identify the effect from $X$ to $Y$ by adding an additional regressor variable.
If $d_I \geq d_X + d_Z$, then $\iv{\{X,Z\}}{Y}(I)$ satisfies the assumption for identifying the effect $\{X,Z\} \rightarrow Y$.
In particular, from this we can extract the effect of interest, $X \rightarrow Y$. 

We refer to this approach as \emph{nuisance instrumental variables (NIV)}. More formally, consider collections of variables $\cI, \cX, \cZ$ and a response variable $Y$. We say that $\beta$ satisfies the \emph{NIV moment equation} if there exist $\alpha\in\R^{d_Y \times d_{\cZ}}$ such that
\begin{equation}\label{eq:niv-identifying-equation}
    \cov(Y - \beta \cX - \alpha \cZ, \cI) = 0.
\end{equation}
We say that $\beta$ is \emph{identified by NIV} or, more formally 
\emph{identified by $\niv{\cX}{\cY}(I, Z)$}
 if additionally $\beta$ is the only solution to the moment equation; that is for all $a\in\R^{d_Y \times d_{\cZ}}$ and $b\in\R^{d_Y \times d_{\cX}}$
\begin{equation}
    \cov(Y - b \cX - a \cZ, \cI) = 0 \implies b = \beta.
\end{equation}
We refer to $\cZ$ as a \emph{nuisance regressor}.
If we use both a nuisance regressor $\cZ$ and condition on a variable $\cB$, 
the conditions become
\begin{equation*}
\text{there exists $\alpha$ s.t. }\E[\cov(Y - \beta \cX - \alpha \cZ, \cI|\cB)] = 0 \quad\text{and}\quad \E[\cov(Y - b \cX - a \cZ, \cI|\cB)] = 0 \implies b = \beta,
\end{equation*}
and we write $\niv{\cX}{Y}(\cI,\cZ|\cB)$; this corresponds to extracting the entries relevant for $\cX$ from the output of $\civ{\cX\cup\cZ}{Y}(\cI|\cB)$.\footnote{As such, one could also write 
$\civ{\underline{\cX}\cup\cZ}{Y}(\cI|\cB)$
instead of 
$\niv{\cX}{Y}(\cI,\cZ|\cB)$,
where the underline indicates which causal effect one is interested in.} 
By choosing $\cZ = \emptyset$, NIV extends CIV. When $\cZ \neq \emptyset$, we use the term \emph{proper NIV}.
The following theorem proves that \cref{assump:civ-relevance,assump:civ-d-sep,assump:civ-descendants} are sufficient to establish identifiability of NIV.
\begin{restatable}[Nuisance IV]{theorem}{nuisanceIV}
\label{prop:nuisance-iv}
    Consider a linear SCM (see \cref{sec:interventions}) over variables $V$, and let $\cI$, $\cX$, $\cZ$, $\cB$, $\{Y\} \subseteq V$ be disjoint collections of variables from $V$, and let $\G$ be the corresponding DAG.
    Assume that $\cI, \cX, \cZ$ and $Y$ have zero mean and finite second moment and let $\beta$ and $\alpha$ be the causal coefficients with which $\cX$ and $\cZ$ enter the structural equation for $Y$, respectively (some of the entries of $\beta$ and $\alpha$ can be zero, so not all variables in $\cX$ and $\cZ$ have to be parents of $Y$). 
    Let $\tilde{\cX} := \cX\cup \cZ$. 
    If \cref{assump:civ-relevance,assump:civ-d-sep,assump:civ-descendants} are satisfied in $\G$ for $\cI, \tilde{\cX}, \cB$ and $Y$, the causal effect $\beta$ of $\cX$ on $Y$ is identified by $\niv{\cX}{Y}(\cI, \cZ|\cB)$.
\end{restatable}
Even though this is a straight-forward extension of IV regression, we are not aware of any work describing the idea of NIV. 
It will prove useful in the time series setting and even in the i.i.d.\ setting it is a strict generalization of CIV:
there are graphs, such as the one in \cref{fig:3-structures} (middle), where the causal effect $X\rightarrow Y$ is neither identified by IV nor by  CIV.  

For some graphs the effect $X \rightarrow Y$ can be identified by (proper) NIV and (proper) CIV.
For example, in the graph in \cref{fig:3-structures} (right) the effect $X \rightarrow Y$ can be identified both by $\civ{X}{Y}(I| B)$ and by $\niv{X}{Y}(\{I, B\},Z)$.
When estimated from a finite sample, the two resulting estimators are not identical, and, as the following proposition establishes, the two approaches cannot in general be sorted in terms of asymptotic variance.\footnote{
In the i.i.d.\ setting, the asymptotic variances of both NIV and CIV estimators can be described by closed form expressions, see \cref{sec:asymptotic-variance}.
}
\begin{restatable}{proposition}{neitherOptimal}
\label{prop:neither-civ-nor-niv-optimal}
    If an effect can be identified by CIV and by NIV, then the estimators cannot be strictly sorted in terms of asymptotic variance. More specifically, there exist data generating processes, for which CIV has strictly smaller asymptotic variance and others, for which NIV has  strictly smaller asymptotic variance.
\end{restatable}
The idea of NIV can be naturally applied in time series settings, too. In \cref{sec:iv-civ-var} we show that causal effects in VAR($1$) processes as described in \cref{sec:setup} can be estimated both by CIV and NIV. To establish this result, we first develop a marginalization technique of time series graphs, which allows us to apply the above results.

\section{Instrumental Time Series Regression}\label{sec:iv-for-time-series}
\subsection{Time Series Reduction}\label{sec:time-series-reduction}
In \cref{sec:graph_representation}, a VAR{($p$)} process is represented by its full time graph
(see, e.g, \cref{fig:sub3}).
We now show that instruments and conditioning sets for VAR{($p$)} processes can be found by
considering marginalized time graphs, which are obtained by marginalization of the full time graph to a finite set of nodes; 
they resemble latent projections \citep[e.g.,][]{verma1992invariant} but are projections of graphs that are not finite. 
\begin{definition}\label{def:marginalzed_graph}
    Consider a process $S = [S_t]_{t\in\Z}$ satisfying \cref{assump:varp} and let $\Gfull$ be the full time graph of $S$ as defined in \cref{sec:graph_representation}.
    Let $M=\{S^{i_1}_{t_1}, \ldots, S^{i_m}_{t_m} \}$ be a finite collection of nodes in $\Gfull$.
    The \emph{marginalized time graph}, $\mathcal{G}_M$, is the graph over nodes $M$ where for all $i, j \in M$ there is:
    \begin{enumerate}
        \item a directed edge $i \rightarrow j$ if and only if $i \rightarrow j$ in $\Gfull$ or 
        there exists $m_1 \in \mathbb{N}$, $v_1, \ldots, v_{m_1} \notin M$ and a directed path $i\rightarrow v_1 \rightarrow \cdots \rightarrow v_{m_1}\rightarrow j$ in $\Gfull$, and
        \item a bidirected edge $i \leftrightarrow j$ if and only if there exists $m_1, m_2 \in \mathbb{N}$, $v_1, \ldots, v_{m_1}, w_1, \ldots, w_{m_2}, U\notin M$ in $\Gfull$ such that there exists directed paths $U\rightarrow v_1\rightarrow \cdots v_{m_1}\rightarrow i$ and $U\rightarrow w_1\rightarrow \cdots w_{m_2} \rightarrow j$.
    \end{enumerate}
\end{definition}
The following theorem establishes that the CIV conditions being satisfied in a marginalized time graph implies a moment condition that can be used for identifying the causal effect.
In the previous section, we stated identifiability results in terms of the causal effect $\beta$.
The following theorem is stated in terms of the total causal effect (see \cref{sec:def_notation}); this generalizes the causal effect and may for instance be interesting if some predictors $X$ are unobserved or are observed but cannot be intervened on. To do so, we state a slight modification of \cref{assump:civ-d-sep}.
\begin{alt-civenum}
\item \label{assump:civ-d-sep-tce} $\cI$ and $Y$ are $d$-separated given $\cB$ in the graph where we remove 
all edges from a node $i \in \cX$ to a node $j \notin \cX$ if the edge lies
on a directed path from $\cX$ to $Y$ that contains only a single node in $\cX$ (namely $i$).
\end{alt-civenum}
\begin{restatable}[Time series IV by marginalization]{theorem}{dreamTheorem}
\label{prop:dream-theorem}
    Consider a process $S = [S_t]_{t\in \Z}$ satisfying \cref{assump:varp} with full time graph $\Gfull$.
    Let $Y$ be some node in $\Gfull$ and let $\cX, \cI, \cZ$, and $\cB$ be disjoint collections of nodes from $\Gfull$. Let $\tilde{\cX} := \cX\cup \cZ$ and define $M := \{Y\} \cup \cX \cup \cI \cup \cZ \cup \cB$. Assume that \cref{assump:civ-d-sep-tce,,assump:civ-descendants} are satisfied for $\cI, \tilde{\cX}, \cB$ and $Y$ in $\G_M$ (see \cref{def:marginalzed_graph}). Then, the following three statements hold.
    (i) The total causal effect $\vecin{\beta, \alpha}$ of $\vecin{\cX^\top,\cZ^\top}^\top$ on $Y$ satisfies the NIV moment equation
    \begin{align} \label{eq:momcond}
        \E[\cov(Y - b \cX - a \cZ, \cI|\cB)] = 0.
    \end{align}
    (ii) Further, if \cref{assump:civ-relevance} is satisfied for $\cI, \tilde{\cX}, \cB$, then $\vecin{\beta, \alpha}$ is the unique solution to \cref{eq:momcond}.
    (iii) If, additionally, $\bX, \bY, \bI, \bZ$ and $\bB$ are observations of $\cX, Y, \cI, \cZ$ and $\cB$ at $T$ time points, $W$ is a positive definite matrix, and
    \begin{align}\label{eq:ts-niv-from-sample}
	    \vecin{\hat b, \hat a} := \argmin_{b,a} \|\hat\cov(\bY - b \bX- a \bZ, \bI|\bB)\|_{W}^2,
    \end{align}
    then $\hat{b}$ is a consistent estimator for $\beta$.
\end{restatable}
We now apply the above result to VAR($1$) processes satisfying \cref{assump:iv}.
In this case, the total causal effect coincides with the $\beta$ defined in \cref{sec:def_notation} (and \ref{assump:civ-d-sep-tce} and \ref{assump:civ-d-sep} become equivalent). The generality of \cref{prop:dream-theorem}, however, can be used to develop similar results for VAR($p$) processes with $p\neq 1$ and for total causal effects between arbitrary variables in the process.

\subsection{Instrumental \texorpdfstring{ VAR{($1$)}}{} Processes}\label{sec:iv-civ-var}
We now consider estimating the causal effect $\beta$ in VAR($1$)-processes, such as the one displayed in {\cref{fig:sub1}}.
A first attempt might be to estimate the effect $\beta$ from $X_{t-1}$ to $Y_t$ by directly adapting the i.i.d.\ case and using $\iv{X_{t-1}}{Y_t}(I_{t-2})$.
{
However, in general, this estimator is not consistent: In \cref{failure:naive_iv} in \cref{app:failure-naive-iv}, we show that the resulting estimator converges to $(1 - \alpha_{I,I}\alpha_{Y,Y})^{-1}\beta$. 
This naive IV approach fails due to the memory of the processes: The instrument, $I_{t-2}$, correlates with the response, $Y_t$, not only through the directed path $I_{t-2} \rightarrow X_{t-1} \rightarrow Y_t$, but also through paths reaching into the past, such as $I_{t-2}\leftarrow I_{t-3} \rightarrow X_{t-2} \rightarrow Y_{t-1}\rightarrow Y_t$. 
}
Instead, the concepts of NIV and time series reductions introduced above 
provide us with a principled approach to selecting which variables to include into the regression and allow us to construct estimators that adjust for the past. 

\subsubsection{Blocking the past using conditional IV} \label{se:meth1}
Consider a {VAR($1$)} process $S$ satisfying \cref{assump:iv}.
By \cref{prop:dream-theorem}, any set satisfying \cref{assump:civ-d-sep,assump:civ-descendants,assump:civ-relevance} in the corresponding marginalised time graph yields a consistent estimator for the causal effect of $X_{t-1}$ on $Y_t$. 
Thus, define $M:=\{I_{t-3}, I_{t-2}, X_{t-2}, X_{t-1}, Y_{t-1}, Y_t\}$ and consider the marginalization $\G_M$ of the full time graph with respect to $M$, see \cref{fig:marginalized-var1} (left). 
In this graph, every path from $I_{t-2}$ to $Y_t$ either goes through $X_{t-1}$ or $I_{t-3}$. 
Indeed, the assumptions for \cref{prop:dream-theorem} are satisfied when choosing
$\cI_t:=\{I_{t-2}\}$ and either $\cB_t = \{I_{t-3}\}$ or $\cB_t := \{I_{t-3},X_{t-2}, Y_{t-1}\}$. 
Formally, we have the following theorem.
\begin{restatable}[Identification with conditioning set]{theorem}{tsCiv}
    \label{prop:ts-civ}
    Consider a {VAR($1$)} process 
    $S = [S_t]_{t \in \Z}$ with
    $S = \vecin{I_t^\top, X_t^\top, H_t^\top, Y_t^\top}^\top_{t\in\Z}$ satisfying \cref{assump:iv}.
    Let either $\cB_t \coloneqq \{I_{t-3}\}$ or $\cB_t \coloneqq \{I_{t-3}, X_{t-2}, Y_{t-1}\}$. Then, the following three statements hold.
    (i)
    The causal effect $\beta$ of $X_{t-1}$ on $Y_t$ satisfies the CIV moment condition $\E[\cov(Y_t - \beta X_{t-1}, I_{t-2}|\cB_t)] = 0$.
    (ii) Furthermore, if $\E[\cov(X_{t-1}, I_{t-2} | \cB_t)]$ has rank $d_X$, then $\beta$ is identified by $\civ{X_{t-1}}{Y_t}(I_{t-2}| \cB_t)$.
    (iii) If, additionally,  $\bX_{t}, \bY_{t}, \bI_{t}$, and $\bB_{t}$ are observations of $X, Y, I$ and $\cB$ at $T$ time points, then $\beta$ can be consistently estimated as $T\to\infty$ by $\civ{\bX_{t-1}}{\bY_t}(\bI_{t-2}| \bB_{t})$, that is, the output of \cref{alg:main-method}. 
\end{restatable}
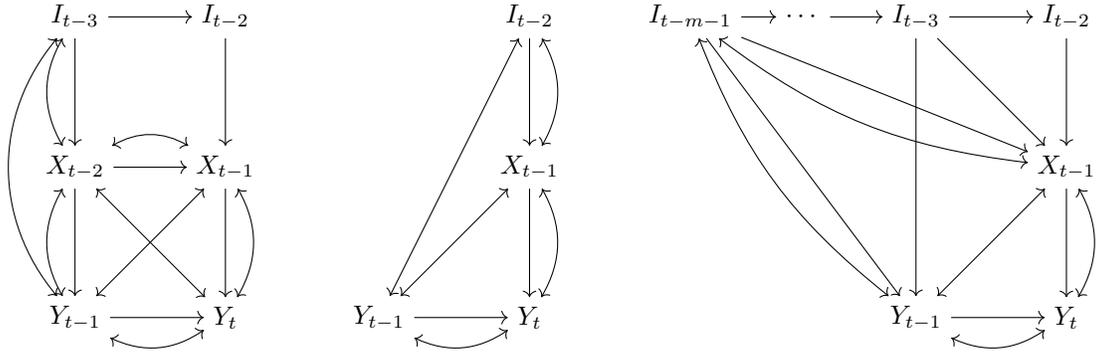
\begin{figure}
    \centering
    \begin{subfigure}{0.3\linewidth}
        \begin{tikzpicture}
        \node (I) at (2,4) {$I_{t-2}$};
        \node (X) at (2,2) {$X_{t-1}$};
        \node (Y) at (2,0) {$Y_{t}$};
        \node (I-) at (0,4) {$I_{t-3}$};
        \node (X-) at (0,2) {$X_{t-2}$};
        \node (Y-) at (0,0) {$Y_{t-1}$};
        \draw[->] (I) edge (X);
        \draw[->] (X) edge (Y);
        \draw[->] (I-) edge (I);
        \draw[->] (I-) edge (X-);
        \draw[->] (X-) edge (Y-);
        \draw[->] (Y-) edge (Y);
        \draw[->] (X-) edge (X);
        \draw[<->] (I-) edge[bend right=40] (Y-);
        \draw[<->] (I-) edge[bend right=30] (X-);
        \draw[<->] (I-) edge (X);
        \draw[<->] (Y-) edge[bend right=30] (Y);
        \draw[<->] (X-) edge[bend right=30] (Y-);
        \draw[<->] (X) edge[bend left=30] (Y);
        \draw[<->] (X-) edge[bend left=30] (X);
        \draw[<->] (X-) edge (Y);
        \draw[<->] (Y-) edge (X);
        \end{tikzpicture}
    \end{subfigure}
    \begin{subfigure}{0.25\linewidth}
        \begin{tikzpicture}
        \node (I) at (2,4) {$I_{t-2}$};
        \node (X) at (2,2) {$X_{t-1}$};
        \node (Y) at (2,0) {$Y_{t}$};
        \node (Y-) at (0,0) {$Y_{t-1}$};
        \draw[->] (I) edge (X);
        \draw[->] (X) edge (Y);
        \draw[->] (Y-) edge (Y);
        \draw[<->] (I) edge[bend left] (X);
        \draw[<->] (I) edge (Y-);
        \draw[<->] (Y-) edge[bend right=30] (Y);
        \draw[<->] (X) edge[bend left=30] (Y);
        \draw[<->] (Y-) edge (X);
        \end{tikzpicture}
    \end{subfigure}
    \begin{subfigure}{0.4\linewidth}
        \begin{tikzpicture}
        \node (I) at (2,4) {$I_{t-2}$};
        \node (X) at (2,2) {$X_{t-1}$};
        \node (Y) at (2,0) {$Y_{t}$};
        \node (I-) at (0,4) {$I_{t-3}$};
        \node (I--) at (-3,4) {$I_{t-m-1}$};
        \node (dots) at (-1.5,4) {$\cdots$};
        \node (Y-) at (0,0) {$Y_{t-1}$};
        \draw[->] (I) edge (X);
        \draw[->] (X) edge (Y);
        \draw[->] (I-) edge (I);
        \draw[->] (I-) edge (Y-);
        \draw[->] (Y-) edge (Y);
        \draw[->] (I-) edge (X);
        \draw[<->] (Y-) edge[bend right=30] (Y);
        \draw[<->] (X) edge[bend left=30] (Y);
        \draw[<->] (Y-) edge (X);
        \draw[->] (I--) edge (dots);
        \draw[->] (dots) edge (I-);
        \draw[->] (I--) edge (Y-);
        \draw[->] (I--) edge (X);
        \draw[<->] (I--) edge[bend right=15] (X);
        \draw[<->] (I--) edge[bend right=15] (Y-);
        \end{tikzpicture}
    \end{subfigure}     \caption{Three different marginalizations 
    of the full time graph $\Gfull$ of a {VAR($1$)} process satisfying \cref{assump:iv} (see \cref{fig:sub3}). 
    Using \cref{prop:dream-theorem}, we use these marginalizations for identification using CIV (left) and NIV (middle and right). (\textit{Left}) Marginalization to nodes $I_{t-2}, X_{t-1}$ and $Y_t$ and their lagged values. (\textit{Middle}) Marginalization to nodes $I_{t-2}, X_{t-1}$ and $Y_t$ and the lagged value of $Y_t$. (\textit{Right}) Marginalization to $m$ instrument nodes $I_{t-2}, \ldots, I_{t-m-1}$, and $X_{t-1}, Y_t,$ and $Y_{t-1}$. 
    }
    \label{fig:marginalized-var1}
\end{figure}
\Cref{prop:ts-civ} establishes that the bias due to  confounding from the past (see beginning of Section~\ref{sec:iv-civ-var})
can be overcome by choosing either $\cB_t = \{I_{t-3}\}$ or $\cB_t = \{I_{t-3}, X_{t-2}, Y_{t-1}\}$ as conditioning set. 
In \cref{sec:simulations}, we compare these two choices empirically.
The assumption that $\alpha_{XI}$ has full rank ensures the relevance condition, \cref{assump:civ-relevance}. 
{
In \cref{sec:asymptotic-variance-ts}, we state the asymptotic distribution of an estimator $\hat\beta$ constructed using \cref{prop:ts-civ}.

\Cref{prop:ts-civ} is robust to misspecification in the number of lags: Though we assume \cref{assump:iv}, which includes a VAR($1$) assumption, the proof of \cref{prop:ts-civ} is still valid if higher-order lags (where each coefficient matrix has the sparsity pattern of \cref{fig:sub1}) are present, as long as the $I$ process is marginally VAR($1$), i.e. $\alpha_{I,I}^j = 0$ for $j = 2, \ldots, p$. In this case, the conditioning set $\cB_t$ still yields 
a valid instrument setting.
}
Other choices of the marginalization set $M$ are possible, too, and yield alternative ways of estimating the causal effect $X_{t-1}\rightarrow Y_t$: We now illustrate an alternative strategy for identification using nuisance IV.

\subsubsection{Blocking the past using nuisance IV}  \label{se:meth2}
The graph in \cref{fig:marginalized-var1} (\textit{middle}) shows the marginalization of the full time graph in {\cref{fig:sub1}} to nodes $M := \{I_{t-2}, X_{t-1}, Y_{t-1}, Y_t\}$. 
The effect $X_{t-1}\rightarrow Y_{t}$ cannot be consistently estimated using only these variables in a CIV; if we condition on $Y_{t-1}$, for example, the path $I_{t-2} \leftrightarrow Y_{t-1}\leftrightarrow Y_t$ is unblocked, violating \cref{assump:civ-d-sep}. 
Instead, we can include $Y_{t-1}$ as a nuisance regressor: We identify the effect of $X_{t-1}$ on $Y_t$ using the instrument $\cI'_{t}:=\{I_{t-2}\}$ by $\niv{X_{t-1}}{Y_t}(\cI'_t, Y_{t-1})$, defined in \cref{sec:nuisance-iv}. This model satisfies the (nuissance) IV \cref{assump:civ-d-sep}, because
the only open paths between $I_{t-2}$ and $Y_t$ include either the edge $X_{t-1}\rightarrow Y_t$ or the edge $Y_{t-1}\rightarrow Y_t$.

\begin{algorithm}[t]
\begin{algorithmic}[1]
	\Statex \textbf{Input}: Sample $\bX = \vec{\bX_1, \ldots, \bX_{T}} \in \R^{d_X \times T}$, $\bY = \vec{\bY_1, \ldots, \bY_{T}} \in \R^{d_Y \times T}$, $\bI = \vec{\bI_1, \ldots, \bI_{T}} \in \R^{d_I \times T}$
	\Statex
	\State Align observations, by setting $\bY = [\bY_s, \ldots, \bY_T]$
	and either
	\begin{align*}
	    \text{(CIV)} \quad\quad & \bI \coloneqq [\bI_{s-2}, \ldots, \bI_{T-2}]\quad \bX \coloneqq [\bX_{s-1}, \ldots, \bX_{T-1}] \quad \text{and} \quad \bB \coloneqq [\bI_{s-3}, \ldots, \bI_{T-3}] \\
	    \text{(NIV)}  \quad\quad & \bI \coloneqq \left[\begin{matrix}\bI_{s-2} \\ \vdots \\ \bI_{s-m+1}\end{matrix} \ldots \begin{matrix}\bI_{T-2} \\ \vdots \\ \bI_{T-m+1}\end{matrix}\right] \quad \text{and}\quad 
	    \bX \coloneqq \left[\begin{matrix}\bX_{s-1} \\  \bY_{s-1} \end{matrix} \ldots
	    \begin{matrix}\bX_{T-1} \\  \bY_{T-1} \end{matrix}\right],
	\end{align*}
	where $s$ is chosen such that all indices are positive. 
	\State Compute regression estimates $\hat{E}[\bX|\bB], \hat{E}[\bY|\bB]$, and $\hat{E}[\bI|\bB]$.
	\State Compute residual processes $r_{\bX} \coloneqq \bX - \hat{E}[\bX|\bB], r_{\bY} \coloneqq \bY - \hat{E}[\bY|\bB]$, and $r_{\bI} \coloneqq \bI - \hat{E}[\bI|\bB]$.
	\State $W \coloneqq \left(\frac{1}{T-s+1} r_\bI r_\bI^\top\right)^{-1}$
	\State $\hat\beta \coloneqq r_\bY r_\bI^\top W  r_\bI r_\bX\bigg(r_\bX r_\bI^\top W r_\bI r_\bX^\top\bigg)^{-1}$
	\State If $\bI$ and $\bX$ are chosen according to NIV, set $\hat\beta \coloneqq \hat\beta_{1:d_X}$.
	\Statex 
	\Statex \textbf{Output}: 
	Estimate of causal effect $\hat\beta$
\end{algorithmic}
\caption{Estimating the causal effect $\beta$ of $X_{t-1}$ on $Y_t$ {in a VAR($1$) process satisfying} \cref{assump:iv}.}
\label{alg:main-method}
\end{algorithm}

Because the resulting $d_X$-dimensional estimate is extracted from the $d_X + d_Y$-dimensional solution $\iv{\{X_{t-1},Y_{t-1}\}}{Y_t}(\cI'_t)$, 
we require that $\rank\E\left\{\vecin{X_{t-1}^\top,Y_{t-1}^\top}^\top(\cI'_t)^\top\right\}=d_X + d_Y$. 
If $d_I < d_X + d_Y$, this rank condition is not met for $\cI'_{t}:=\{I_{t-2}\}$.
To overcome this, one can increase the instrument set to $\cI_t = \{I_{t-2}, I_{t-3},\ldots, I_{t-m-1}\}$ {for a user-specified $m$}: \cref{fig:marginalized-var1} (\textit{right}) shows the marginalization of the full time graph to $M \coloneqq \cI_t \cup \{X_{t-1},Y_{t-1},Y_t\}$. 
Again, \cref{assump:civ-d-sep,assump:civ-descendants} are satisfied in $\G_M$, but the instrument set now has dimension $|\cI| = m d_I$ 
{(implying that we should choose $m \geq (d_X + d_Y)/d_I$)}
and, provided the rank condition now holds, $\beta$ is identified by $\niv{X_{t-1}}{Y_t}(\cI_t, Y_{t-1})$. 
The following theorem formalizes this discussion.
\begin{restatable}[Identification with nuisance regressor]{theorem}{tsNiv}
    \label{prop:ts-niv}
    Consider a {VAR($1$)} process 
    $S = [S_t]_{t \in \Z}$ with
    $S = \vecin{I_t^\top, X_t^\top, H_t^\top, Y_t^\top}^\top_{t\in\Z}$ satisfying \cref{assump:iv}.
    Let $\cI_t \coloneqq \{I_{t-2},\ldots, I_{t-m-1}\}$ for an $m \geq 1$ and $\cZ_t \coloneqq \{Y_{t-1}\}$. Then, the following three statements hold.
    (i)
    There exists $\alpha \in \R$ such that the causal effect $\beta$ of $X_{t-1}$ on $Y_t$ satisfies the NIV moment condition {$\cov(Y_t - \beta X_{t-1}-\alpha \cZ_t, \cI_t) = 0$}.
    (ii)
    Further, if $\E[\vecin{X_{t-1}^\top,\cZ_t^\top}^\top\cI_t^\top]$ has rank $d_X + d_Y$, $\beta$ is identified by $\niv{X_{t-1}}{Y_t}(\cI_t, \cZ_t)$.
    (iii)
    If, additionally, $\bX_{t}, \bY_{t}, \bI_{t}$, and $\bZ_{t}$ are observations of $X, Y, \cI$ and $\cZ$ at $T$ time points, then $\beta$ can be consistently estimated as $T\to\infty$
    by $\niv{\bX_{t-1}}{\bY_t}(\bI_t, \bZ_{t})$, that is, the output of \cref{alg:main-method}. 
\end{restatable}
{
In \cref{sec:asymptotic-variance-ts}, we state the asymptotic distribution of an estimator $\hat\beta$ constructed using \cref{prop:ts-niv}.
Similar to the discussion in 
\cref{se:meth1}, \cref{prop:ts-niv} is robust to misspecification of the number of lags. While \cref{prop:ts-civ} is still valid in VAR($p$) processes if the $I$ process is a VAR($1$) process, \cref{prop:ts-niv} is valid in VAR($p$) processes 
if the $Y$ process is VAR($1$), since in that case $Y_{t-1}$ yields a valid nuisance regressor.

}
\Cref{prop:ts-niv} shows that identification is possible if $\rank\E\left\{\vecin{X_{t-1}^\top,Y_{t-1}^\top}^\top\cI_t^\top\right\}=d_X+d_Y$ is met, where we use the lagged instrument set $\cI_t = \{I_{t-2},\ldots, I_{t-m-1}\}$ (it is easy to see that for this choice of $\cI_t$, identifiability also holds for $\beta=0$ and the rank being $d_X$).
Satisfying this relevance criterion implies \cref{assump:civ-relevance}. For the CIV approach in \cref{se:meth1}, this is directly related to the rank of the parameter $\alpha_{XI}$. 
For NIV, the relevance criterion depends on the parameter matrix in an intricate way.
We now provide necessary and sufficient conditions for when this rank condition is satisfied. Moreover, we show in \cref{thm:almost_sure_identify} below that for almost all parameter matrices one can obtain sufficiently high rank to identify the effect $X_{t-1}\rightarrow Y_t$ by increasing the number of lags used.

Define the following submatrices of the parameter matrix $A$ from \cref{assump:iv}.
\begin{equation}\label{eq:submatrices}
    A_I := \mat{\alpha_{X,I} \\ 0} \in \R^{d_X + 1}\quad \text{and} \quad A_{XY} := \mat{\alpha_{X,X} & \alpha_{X,Y} \\ \beta & \alpha_{Y,Y}} \in \R^{(d_X + 1)\times (d_X + 1)}.
\end{equation}
The following theorem outlines conditions for $\E[\vecin{X_{t-1}^\top,Y_{t-1}^\top}^\top \cI^\top]$ to have full rank when $d_I = 1$ and we use $d_X + d_Y = d_X+1$ lags as instruments:
\begin{restatable}{theorem}{jordanIdentifiability}
\label{prop:identifiability}
	Consider a process $S = \vecin{I_t^\top, X_t^\top, H_t^\top, Y_t^\top}^\top_{t\in\Z}$ satisfying \cref{assump:iv}.
	Assume that $d_I = d_Y = 1$ and let $\cI_{t}:= \{I_{t-2},\ldots, I_{t-m-1}\}$, where $m = d_X + d_Y$. 
	Let $A_{XY}$ and $A_I$ be defined as in \cref{eq:submatrices}.
	The following three statements are equivalent:
    \begin{enumerate}
		\item $\rank \E[\vecin{X_{t-1}^\top,Y_{t-1}^\top}^\top \cI_t^\top] = d_X + d_Y$.
		\item The matrix $\vec{A_{XY}^0 A_I,\, A_{XY}^1 A_I, \ldots, A_{XY}^{d_X}A_I}$ is invertible, where $A_{XY}^0$ is the identity matrix of size $(d_X+d_Y) \times (d_X+d_Y)$.
		\label{item:controlability}
		\item Different Jordan blocks of $J$ have different eigenvalues and for all $q \in \{1, \ldots, k\}$, the coefficient $w_{\sum_{i=1}^q m_i}$ is non-zero;
		here, $J = Q^{-1}A_{XY} Q$ is the Jordan normal form\footnote{See \cref{sec:proof-jordan-forms} for the definition of Jordan normal forms and the notation that we use.} of $A_{XY}$, with $k$ Jordan blocks $J = \diag(J_{m_1}(\lambda_1), \ldots, J_{m_k}(\lambda_k))$, each with size $m_i$ and eigenvalue $\lambda_i$, and $w$ are the coefficients of $A_I$ in the basis of the generalized eigenvectors $Q$, that is, $w = Q^{-1}A_I$.
	\end{enumerate}
\end{restatable}
Intuitively, \cref{prop:identifiability} states that identification of $\beta$ is possible if the information is passed on in a sufficiently diverse way from $I_{t-k}$ to $(X_t, Y_t)$. 
For every $k \in \{0, \ldots, d_X\}$, $A^{k}_{XY}A_I$ corresponds to the path $I_{t-(k+1)}\stackrel{A_I}{\rightarrow}(X_{t-k},Y_{t-k})\stackrel{A_{XY}}{\rightarrow}\ldots\stackrel{A_{XY}}{\rightarrow} (X_{t},Y_{t})$ in \cref{fig:sub3}.
\cref{prop:identifiability}~\ref{item:controlability} requires these to be sufficiently different (that is, linearly independent), for the matrix $\E[\vecin{X_{t-1}^\top,Y_{t-1}^\top}^\top \cI_t^\top]$ to have full rank. 
While there are parameter matrices that do not satisfy \cref{prop:identifiability} (see \cref{appendix:ts-niv-examples} for examples), the following corollary shows that, when chosen randomly, almost all parameter matrices $A$ allow for using multiple lags $I_{t-2-j}$ as instruments.
\begin{restatable}{corollary}{almostSure}
\label{thm:almost_sure_identify}
Consider a VAR($1$) process $S$ with $d_I = 1$ and parameter matrix $A$, 
and assume that sparsity pattern of $A$ is given by \cref{assump:iv} and that the non-zero entries of $A$ are drawn from any distribution which has density with respect to Lebesgue measure.
Then $\beta$ is identifiable with probability $1$.
\end{restatable}

The following corollary provides a sufficient condition for identifiability of $\beta$, when instruments are multivariate.
\begin{restatable}{corollary}{moreThanOneInstrument}
\label{prop:identify_more_than_one_instrument}
    Consider a {VAR($1$)} process $S$ satisfying \cref{assump:iv} with $d_I > 1$ instrument processes $I^{(1)}, \ldots, I^{(d_I)}$. Assume that there is at least one instrument process $I^{(j)}$ such that both of the following conditions hold.
    \begin{enumerate}
        \item $I^{(j)}_{t}$ is independent of $I^{(i)}_{s}$ for all $t,s$ and $i \neq j$, and
        \item the requirements of \cref{prop:identifiability} are satisfied for the reduced process $(I^{(j)}, X, Y)$. 
    \end{enumerate}
    Then $\beta$ is identifiable.
\end{restatable}
Using a single instrument at $d_X + 1$ lags allows for a simple condition for identifiability. 
But in finite samples, using instruments with high time lags may come at a loss of efficiency as the estimation procedure may suffer from weak dependencies between the residual and the instrument due to the mixing of the time series, see \cref{sec:simulations} for an empirical investigation.

\subsubsection{Extension to non-VAR \texorpdfstring{{Nonlinear}}{} Processes}
\label{sec:relaxing-varp}
The results in \cref{se:meth1,se:meth2} assume that the entire process $S$ is a VAR($1$) process, see \cref{assump:iv}, 
which was done mostly for presentation purposes. 
The key arguments and statements hold more generally.
We have argued that similar arguments hold for VAR($p$) processes and we now consider a setting, where $Y_t$ 
satisfies a linear structural equation (as in a VAR($1$) process) 
but we do not assume any VAR{($p$)} structure or linearity on the remaining subprocesses.
We outline the assumptions needed to obtain the same identification results as in \cref{se:meth1,se:meth2} (a similar relaxation can be applied when $Y_t$ behaves like a VAR($p$) process).
Let $[S_t]_{t\in\Z}=[I_t^\top,X_t^\top,H_t^\top,Y_t^\top]^\top_{t\in\Z}$ and assume that for all $t\in\Z$
\begin{equation}\label{eq:structural-Y}
    Y_t \coloneqq \beta X_{t-1} + \alpha_{Y,Y} Y_{t-1}+g(\epsilon_t^Y, H_{t-1}),
\end{equation}
where $\epsilon_t^Y$ is a sequence of i.i.d.\ random variables that, for all $t$, are independent of $[Y_s]_{s<t}$ and $[X_s, I_s, H_s]_{s\leq t}$ and $g$ is a measurable function. 
Without assuming a VAR{($p$)} process we make the following assumptions on $S$.
\begin{alt-assumpenum}
    \item \label{assump:S_stable} $[S_t]_{t\in\Z}$ is covariance stationary.
    \item \label{assump:empirical_moments} $S$ satisfies \cref{eq:assump_moments}, that is, empirical first and second moments converge to their population version.
    \item \label{assump:CIV_general} There exists a $p\in\N$ such that for 
    all $i \in \{1, \ldots, d\}$, there exists a function $f^i$ such that for all $t\in\Z$ 
    \begin{align*}
         S^i_t=f^i(S_{t-1},...,S_{t-p})+\varepsilon^i_t.
     \end{align*}
     This induces a full time graph $\Gfull$ \citep{peters2013causal};
     we assume that for all finite disjoint collections of nodes $A, B, C$ from $\Gfull$ such that $A$ and $C$ are $d$-separated given $B$ in $\Gfull$, we have $A \indep C | B$.
     Furthermore, 
     for all $t\in\Z$, $H_{t-1}\in \ND{I_{t-2}}$, $Y_{t-1}\in \ND{I_{t-2}}$ and  $\varepsilon_Y^t$ is independent of any finite set $A \subseteq \ND{Y_t}$. 
     \item \label{assump:inst_indep_of_noise} For all $t\in \Z$ and $m\in \N$, we have $(\varepsilon_t^Y,H_{t-1})\indep (I_{t-2},..., I_{t-2-m})$.
\end{alt-assumpenum}
Using these assumptions, we can restate \cref{prop:ts-civ} without the VAR(1) assumption.

\begin{restatable}[Identification with conditioning set relaxing the VAR{($p$)} assumption]{proposition}{tsCivgen}
\label{thm:CIV_general}
Consider a process $S = \vecin{I_t^\top, X_t^\top, H_t^\top, Y_t^\top}^\top_{t\in\Z}$ satisfying \cref{eq:structural-Y,assump:S_stable,assump:empirical_moments,assump:CIV_general}.
    Let $\cB_t$ be a set of variables satisfying $\PA{(I_{t-2})}\subseteq\cB_t\subseteq \ND{Y_{t}}\cap \ND{I_{t-2}}$ in $\Gfull$.
    Then, (i), (ii) and (iii) from \cref{prop:ts-civ} hold. 
\end{restatable}

Similarly, we can extend 
\cref{prop:ts-niv} to more general time series models, too.
\begin{restatable}[Identification with nuisance regressor relaxing the VAR{($p$)} assumption]{proposition}{tsNivgen}
    \label{thm:NIV_general}
    Consider a process $S = \vecin{I_t^\top, X_t^\top, H_t^\top, Y_t^\top}^\top_{t\in\Z}$ satisfying \cref{eq:structural-Y,assump:S_stable,assump:inst_indep_of_noise,assump:empirical_moments}.
    Let 
    $\cZ_t \coloneqq \{Y_{t-1}\}$
    and $\cI_t \coloneqq \{I_{t-2},\ldots, I_{t-m-1}\}$ for an $m \geq 1$. 
    Then, (i), (ii), and (iii) from \cref{prop:ts-niv} hold.
\end{restatable}

\Cref{eq:structural-Y} allows for binary covariates or instruments. The assumed model in $\cref{eq:structural-Y}$ does not allow for the response $Y_t$ to be binary, though \cref{thm:CIV_general,thm:NIV_general} readily extend to other parameterizations of the parameter of interest.
{Another generalization that is possible, is to let $\beta_t$ change over time. We show in \cref{app:time-inhomogeneous-effect} that one can also develop estimating equations for this setting.}

So far, we have focused on estimating causal effects. Knowledge of such causal effects can be of interest in itself. In the following section, we discuss that they can also be used for prediction and forecasting under interventions.

\subsection{Optimal Prediction under Interventions} \label{sec:prediction_intervention}
Causal estimates may facilitate improved prediction under intervention.
Suppose that we have inferred the causal effect $X_{t-1}\rightarrow Y_t$ for instance using the methods presented in \cref{sec:iv-civ-var} above. 
How do we best predict $Y_{t+1}$ given that we 
perform the intervention $\operatorname{do}(X_t:=x)$ (see \cref{sec:interventions} for an introduction to do-interventions)
and that we observe the past values of the time series
(but the value of $X_t$ had it not been intervened on is not observed)?

Due to the hidden confounding, the conditional mean of $Y_{t+1}$ given its past (which could be consistently estimated by OLS regression, for example) is in general not optimal in terms of mean squared prediction error (MSPE). Intuitively, this is because the conditional mean also encompasses the effects of the latent process $H$ onto $Y_{t+1}$. 
In the i.i.d.\ setting, it has been observed that using the causal parameter yields a smaller MSPE and can be worst-case optimal under arbitrarily large interventions
\citep[e.g.,][]{Rojas2016, Christiansen2020DG}.

In VAR{($p$)} processes, the intervention $\operatorname{do}(X_t := x)$ partially breaks the confounding of $X_{t}$ and $Y_{t+1}$ from the past{. If an unobserved} process $H$ {confounds $X$ and $Y$}, the {marginalized} process $(X,Y)$ is not a Markov process {(neither with or without intervention), and therefore} the lagged observations $\{X_{t-k}, k = 1, \ldots, m\}$ and $\{Y_{t-j}, j = 1, \ldots, l\}$ {can} further improve prediction of $Y_{t+1}${, due to being informative of $H$, and therefore of $Y_{t+1}$}.
Fixing the number of lags, the following proposition shows that, under the intervention $\operatorname{do}(X_t := x)$, the optimal linear prediction consist of a mix of (population) regression parameters for non-intervened variables and causal parameters for the intervened variable $X_t$.
\begin{restatable}{proposition}{optimalPrediction}
\label{prop:optimal_prediction}
    Consider a {VAR($1$)} process $S=[S_t]_{t\in \mathbb{Z}}$ satisfying~\cref{assump:iv}.
    Let $\beta$ be the causal effect from $X_t$ to $Y_{t+1}$, and let, for an arbitrary $m, \ell \in \mathbb{N}$,
    $(\alpha_{Y,X}, \alpha_{Y,Y})$ be 
    the population vector of coefficients when regressing $Y_{s+1}-\beta X_s$ on $\{{X}_{s-k}, k = 1, \ldots, m\} \cup \{{Y}_{s-j}, j = 0, \ldots, l\}$.
    Then 
    \begin{align*}
        (\alpha_{Y,Y},\beta,\alpha_{Y,X})
        = \arg\min_{a, b, c}\E_{\operatorname{do}(X_t:=x)}\left\{Y_{t+1} - \sum_{j=0}^\ell a_j Y_{t-j} - bX_{t} - \sum_{k=1}^m c_k X_{t-k}\right\}^2.
    \end{align*}
\end{restatable}
That is, under the intervention $\operatorname{do}(X_t := x)$, the causal coefficient can be used to optimally predict $Y_{t+1}$. We state the corresponding finite sample algorithm in Appendix~\ref{sec:apppredunderinterv}.

\section{Simulation Experiments}\label{sec:simulations}
We test the empirical performance of our proposed estimators in simulation experiments. 
(\Cref{sec:exp-fixed-generating,sec:weak-instruments-experiment} discuss additional experiments on settings with weak instruments.)

\paragraph{Data generating process.} 
We generate data by first simulating a matrix $A$ with the sparsity structure of \cref{fig:sub1} and all non-zero entries being drawn independently uniformly at random from $(-0.9, -0.1)\cup(0.1,0.9)$. By \cref{thm:almost_sure_identify}, the causal effect under this sampling scheme is almost surely identifiable by NIV. Unless specified differently, we use $d_I = 3, d_X = 2$ and $d_Y = 1$. 
Any such randomly generated matrix is used to generate data only if it satisfies the eigenvalue condition in \cref{assump:varp} with a margin of $0.1$. Furthermore, all noise variables $\varepsilon_t^{i}$ are independently randomly drawn from a normal distribution with mean $0$ and standard deviation $1$.

\paragraph{Evaluation of the estimators.}
We simulate $m=1{,}000$ random matrices and for each we simulate $s=10$ data sets. 
We fit estimators $\hat\beta_j$ for each $j \in \{1, \ldots, s\}$, and for a given matrix, we compute the average error, 
$\operatorname{error}(\hat\beta) \coloneqq \operatorname{mean}( \|\hat\beta_1 - \beta\|_2^2, \ldots, \|\hat\beta_s - \beta\|_2^2)$.

\subsection{Identification of the causal effect}
\label{sec:Sim_compare}
\paragraph{Consistency}
The estimators CIV and NIV are consistent
(\cref{prop:ts-civ,prop:ts-niv}) and we perform an experiment to compare their finite sample properties for different sample sizes. 
We simulate data from the scheme described above and fit two $\civ{X_{t-1}}{Y_t}(I_{t-2}|\cB_t)$ estimators, where $\cB_t = \{I_{t-3}\}$ ($\CIV_{I}$) or $\cB_t = \{I_{t-3}, X_{t-2}, Y_{t-1}\}$ ($\CIV_{I,X,Y}$)
and two $\niv{X_{t-1}}{Y_t}(\cI_t, Y_{t-1})$ {estimators},
where $\cI_t = \{I_{t-2}\}$ ($\NIV_{\text{1 lag}}$) or $\cI_t = \{I_{t-2}, I_{t-3}, I_{t-4}\}$ ($\NIV_{\text{3 lag}}$).
All of these estimators are consistent (see \cref{sec:iv-civ-var}).
We plot the errors obtained for different sample sizes $T$ in \cref{plot:compare} (left).
For all estimators, the errors decrease with increasing sample size supporting the consistency results in \cref{prop:dream-theorem}.
In general, there is no empirical support for either of the NIV or CIV estimators to be strictly better than the others in terms of speed of convergence in terms of sample size.
As discussed in \cref{se:meth1}, both $\cB_t=\{I_{t-3}\}$ and $\cB_t=\{I_{t-3}, X_{t-2}, Y_{t-1}\}$ block confounding from past values. We observe that removing $X_{t-2}$ and $Y_{t-1}$ from $\cB_t$ increases the upper tail of the error distribution, supporting the intuition that while the conditioning variables $X_{t-2}$ and $Y_{t-1}$ are not necessary for identification, they reduce finite sample variance.
Similarly, for the NIV estimators, using $3$ lags instead of a single lag shrinks the upper tail of the error distribution. {In \cref{sec:error-correlation-matrix} we print the correlation matrices between the error of the estimators, and in \cref{sec:exp-fixed-generating} we conduct a similar experiment, but with fixed data generating mechanisms.}
\begin{figure}[t]
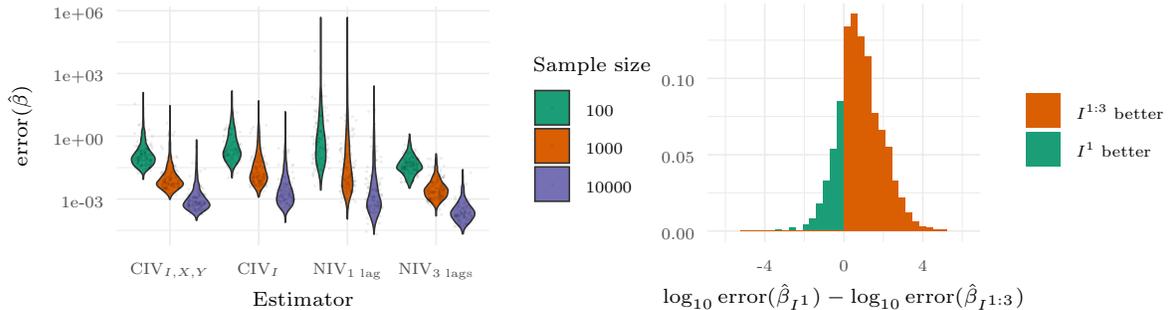

    \centering
    {\scriptsize
    \begin{subfigure}{0.55\textwidth}

     \end{subfigure}
    }
    \caption{
    (\textit{Left}) Distributions of the average error (in log scale) of different consistent CIV and NIV estimators for various sample sizes $T$, see Section~\ref{sec:Sim_compare} (`Consistency'). 
    Each point corresponds to the average over repeated draws from the same parameter matrix, and different points correspond to different parameter matrices.
    The $\NIV_{1\text{ lag}}$ estimators have heavier tails, corresponding to the fact that this is a just-identified case, where we use a $3$-dimensional instrument to estimate a $3$-dimensional causal effect.
    (\textit{Right})
    Histogram of log error ratios for two different NIV estimators in a model with a $3$-dimensional instrument. $I^{1}$ uses $6$ lags from a $1$-dimensional instrument process, while $I^{1:3}$ uses $2$ lags of a $3$-dimensional instrument process, see Section~\ref{sec:Sim_compare} (`Using more lags or additional instruments'). A value larger than zero indicates that $I^{1}$ yields a larger error than $I^{1:3}$. This is the case for the majority of the considered settings.
    }
    \label{plot:compare}
\end{figure}

\paragraph{Using more lags or additional instruments.} 
We compare using multiple lags of a single instrument to 
using multiple instrument processes. 
We consider the model described above with $d_I=3$ independent instrument processes, and use either the first instrument process ($I^{1}$) or all three instrument processes ($I^{1:3}$) for estimation in NIV. For $I^{1:3}$ we use 2 lags of each of the three processes (`recent instruments'), while for $I^{1}$ we use 6 lags (`distant instruments'), such that both models use in total $6$ instruments. In this way, we can inspect the benefit of using more recent instruments if available.

\Cref{plot:compare} (right) shows a histogram of the log error ratio of the two different estimators: A large value indicates that model $I^{1}$, which uses many lags of only a single instrument process, incurs a higher error. 
In the majority of parameter settings, using recent instruments yields a lower error, in some cases several orders of magnitude, when compared to the more distant instruments. 
Although adding lags of a univariate instrument process can yield identifiability of the causal effect of a regressor process that is not univariate (see \cref{prop:identifiability}), this simulation experiment supports the notion that using more instruments (if available) is preferred over using more distant lags. 

\paragraph{Estimation close to non-identifiability.}
\Cref{ex:not_identifiable} in \cref{appendix:ts-niv-examples} shows a setting, in which 
$\beta$ is not identifiable by NIV (in the setting of \cref{thm:almost_sure_identify}, this happens with probability zero).
We examine the behaviour of the NIV estimator in scenarios that are close to this non-identifiable setting.
We consider $d_I = 1$, $\alpha_{X,X}=\diag(-0.6, -0.6+\Delta)$ and draw the remaining parameters uniformly from $(-0.9,-0.1)\cup (0.1,0.9)$.
As per \cref{thm:almost_sure_identify}, the causal effect is identifiable, except for $\Delta = 0$, 
in which case $\alpha_{X,X}$ has two Jordan blocks with the same eigenvalue.
In \cref{plot:non_ident} we plot the median error for changing $\Delta$ and sample size $T$.\footnote{Here, we report the median, since the non-identifiability when $\Delta=0$ implies that the mean is ill-behaved; for those lines where $\Delta\neq 0$, plotting the mean instead of the median yields a similar plot (not shown).}
The further $\Delta$ is from $0$, the faster (in terms of sample size) the estimator converges to $\beta$.
In the non-identified setting $\Delta = 0$, we do not observe the error to decrease with increasing sample size. 
This observation is in line with \cref{prop:identifiability}.
\begin{figure}[t]
    \centering
    {\scriptsize
\begin{tikzpicture}[x=1pt,y=1pt]
\definecolor{fillColor}{RGB}{255,255,255}
\begin{scope}
\definecolor{drawColor}{gray}{0.92}

\path[draw=drawColor,line width= 0.3pt,line join=round] ( 42.95, 50.40) --
	(203.99, 50.40);

\path[draw=drawColor,line width= 0.3pt,line join=round] ( 42.95, 74.73) --
	(203.99, 74.73);

\path[draw=drawColor,line width= 0.3pt,line join=round] ( 42.95, 99.06) --
	(203.99, 99.06);

\path[draw=drawColor,line width= 0.6pt,line join=round] ( 42.95, 38.23) --
	(203.99, 38.23);

\path[draw=drawColor,line width= 0.6pt,line join=round] ( 42.95, 62.56) --
	(203.99, 62.56);

\path[draw=drawColor,line width= 0.6pt,line join=round] ( 42.95, 86.89) --
	(203.99, 86.89);

\path[draw=drawColor,line width= 0.6pt,line join=round] ( 42.95,111.23) --
	(203.99,111.23);

\path[draw=drawColor,line width= 0.6pt,line join=round] ( 58.54, 30.69) --
	( 58.54,120.97);

\path[draw=drawColor,line width= 0.6pt,line join=round] ( 84.51, 30.69) --
	( 84.51,120.97);

\path[draw=drawColor,line width= 0.6pt,line join=round] (110.49, 30.69) --
	(110.49,120.97);

\path[draw=drawColor,line width= 0.6pt,line join=round] (136.46, 30.69) --
	(136.46,120.97);

\path[draw=drawColor,line width= 0.6pt,line join=round] (162.44, 30.69) --
	(162.44,120.97);

\path[draw=drawColor,line width= 0.6pt,line join=round] (188.41, 30.69) --
	(188.41,120.97);
\definecolor{drawColor}{RGB}{27,158,119}

\path[draw=drawColor,line width= 0.6pt,line join=round] ( 58.54,113.43) --
	( 84.51,115.41) --
	(110.49,113.31) --
	(136.46,112.25) --
	(162.44,114.25) --
	(188.41,114.84);
\definecolor{drawColor}{RGB}{217,95,2}

\path[draw=drawColor,line width= 0.6pt,line join=round] ( 58.54,116.87) --
	( 84.51,112.42) --
	(110.49,113.71) --
	(136.46,110.80) --
	(162.44,112.99) --
	(188.41,111.14);
\definecolor{drawColor}{RGB}{117,112,179}

\path[draw=drawColor,line width= 0.6pt,line join=round] ( 58.54,114.20) --
	( 84.51,109.90) --
	(110.49,109.96) --
	(136.46, 99.82) --
	(162.44, 95.93) --
	(188.41, 80.79);
\definecolor{drawColor}{RGB}{231,41,138}

\path[draw=drawColor,line width= 0.6pt,line join=round] ( 58.54,107.69) --
	( 84.51, 94.18) --
	(110.49, 89.05) --
	(136.46, 72.02) --
	(162.44, 66.04) --
	(188.41, 49.20);
\definecolor{drawColor}{RGB}{102,166,30}

\path[draw=drawColor,line width= 0.6pt,line join=round] ( 58.54,100.43) --
	( 84.51, 81.41) --
	(110.49, 74.37) --
	(136.46, 58.73) --
	(162.44, 48.82) --
	(188.41, 34.79);
\end{scope}
\begin{scope}
\definecolor{drawColor}{gray}{0.30}

\node[text=drawColor,anchor=base east,inner sep=0pt, outer sep=0pt, scale=  0.88] at ( 38.00, 35.20) {0.001};

\node[text=drawColor,anchor=base east,inner sep=0pt, outer sep=0pt, scale=  0.88] at ( 38.00, 59.53) {0.010};

\node[text=drawColor,anchor=base east,inner sep=0pt, outer sep=0pt, scale=  0.88] at ( 38.00, 83.86) {0.100};

\node[text=drawColor,anchor=base east,inner sep=0pt, outer sep=0pt, scale=  0.88] at ( 38.00,108.20) {1.000};
\end{scope}
\begin{scope}
\definecolor{drawColor}{gray}{0.30}

\node[text=drawColor,anchor=base,inner sep=0pt, outer sep=0pt, scale=  0.88] at ( 58.54, 19.68) {100};

\node[text=drawColor,anchor=base,inner sep=0pt, outer sep=0pt, scale=  0.88] at ( 84.51, 19.68) {500};

\node[text=drawColor,anchor=base,inner sep=0pt, outer sep=0pt, scale=  0.88] at (110.49, 19.68) {1000};

\node[text=drawColor,anchor=base,inner sep=0pt, outer sep=0pt, scale=  0.88] at (136.46, 19.68) {5000};

\node[text=drawColor,anchor=base,inner sep=0pt, outer sep=0pt, scale=  0.88] at (162.44, 19.68) {10000};

\node[text=drawColor,anchor=base,inner sep=0pt, outer sep=0pt, scale=  0.88] at (188.41, 19.68) {50000};
\end{scope}
\begin{scope}
\definecolor{drawColor}{RGB}{0,0,0}

\node[text=drawColor,anchor=base,inner sep=0pt, outer sep=0pt, scale=  1.10] at (123.47,  7.64) {$n$};
\end{scope}
\begin{scope}
\definecolor{drawColor}{RGB}{0,0,0}

\node[text=drawColor,rotate= 90.00,anchor=base,inner sep=0pt, outer sep=0pt, scale=  1.10] at ( 13.08, 75.83) {$\operatorname{error}(\hat\beta)$};
\end{scope}
\begin{scope}
\definecolor{drawColor}{RGB}{0,0,0}

\node[text=drawColor,anchor=base west,inner sep=0pt, outer sep=0pt, scale=  1.10] at (220.49,110.93) {Eigenval. $\Delta$};
\end{scope}
\begin{scope}
\definecolor{drawColor}{RGB}{27,158,119}

\path[draw=drawColor,line width= 0.6pt,line join=round] (221.94, 97.13) -- (233.50, 97.13);
\end{scope}
\begin{scope}
\definecolor{drawColor}{RGB}{217,95,2}

\path[draw=drawColor,line width= 0.6pt,line join=round] (221.94, 82.68) -- (233.50, 82.68);
\end{scope}
\begin{scope}
\definecolor{drawColor}{RGB}{117,112,179}

\path[draw=drawColor,line width= 0.6pt,line join=round] (221.94, 68.22) -- (233.50, 68.22);
\end{scope}
\begin{scope}
\definecolor{drawColor}{RGB}{231,41,138}

\path[draw=drawColor,line width= 0.6pt,line join=round] (221.94, 53.77) -- (233.50, 53.77);
\end{scope}
\begin{scope}
\definecolor{drawColor}{RGB}{102,166,30}

\path[draw=drawColor,line width= 0.6pt,line join=round] (221.94, 39.31) -- (233.50, 39.31);
\end{scope}
\begin{scope}
\definecolor{drawColor}{RGB}{0,0,0}

\node[text=drawColor,anchor=base west,inner sep=0pt, outer sep=0pt, scale=  0.88] at (240.45, 94.10) {0.0};
\end{scope}
\begin{scope}
\definecolor{drawColor}{RGB}{0,0,0}

\node[text=drawColor,anchor=base west,inner sep=0pt, outer sep=0pt, scale=  0.88] at (240.45, 79.65) {0.01};
\end{scope}
\begin{scope}
\definecolor{drawColor}{RGB}{0,0,0}

\node[text=drawColor,anchor=base west,inner sep=0pt, outer sep=0pt, scale=  0.88] at (240.45, 65.19) {0.1};
\end{scope}
\begin{scope}
\definecolor{drawColor}{RGB}{0,0,0}

\node[text=drawColor,anchor=base west,inner sep=0pt, outer sep=0pt, scale=  0.88] at (240.45, 50.74) {0.5};
\end{scope}
\begin{scope}
\definecolor{drawColor}{RGB}{0,0,0}

\node[text=drawColor,anchor=base west,inner sep=0pt, outer sep=0pt, scale=  0.88] at (240.45, 36.28) {1.0};
\end{scope}
\end{tikzpicture}
     }
    \caption{Median error (log scale) for the NIV estimator as we vary $\Delta$, the difference between the eigenvalues of $\alpha_{X,X}$, see Section~\ref{sec:Sim_compare} (`Estimation close to non-identifiability'). The causal effect is identifiable if and only if $\Delta \neq 0$, see~\cref{prop:identifiability}. 
    Indeed, the error does not decrease for $\Delta = 0$ and decreases the faster, the further $\Delta$ is away from $0$.}
    \label{plot:non_ident}
\end{figure}
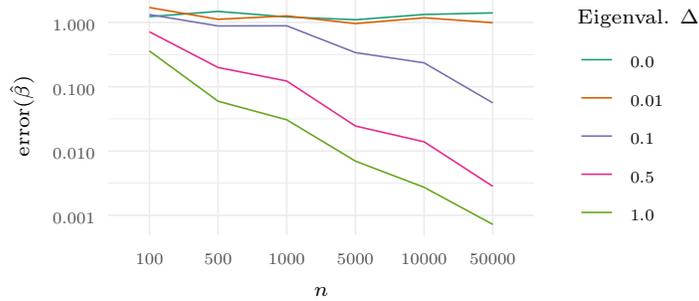

\subsection{Using the causal parameter for prediction under intervention on \texorpdfstring{$X$}{}}\label{sec:prediction_simulation}
We support empirically that a linear prediction using OLS parameters for non-intervened variables and causal parameters for the intervened variables achieves minimal square loss for prediction under intervention (see~\cref{prop:optimal_prediction}).
We consider the model with $d_X=d_I=1$ and, ensuring strong hidden confounding, draw the non-zero entries $\alpha_{i,H},\ i\in\{X,H,Y\}$ uniformly at random from $(-0.9,-0.5)\cup (0.5,0.9)$ (instead of $(-0.9, -0.1)\cup(0.1,0.9)$ as for the other non-zero entries of $A$).
The prediction task follows \cref{sec:prediction_intervention}:
Given observations
$\mathbf{X} = \vec{\mathbf{X}_1, \ldots, \mathbf{X}_{T-1}} \in \R^{d_X \times T-1}$
and
$\mathbf{Y} = \vec{\mathbf{Y}_1, \ldots, \mathbf{Y}_T} \in \R^{d_Y \times T}$, with $T = 3{,}000$,
the goal is to predict
$\mathbf{Y}_{T+1} \in \R^{d_Y \times 1}$
under an intervention $\operatorname{do}(X_T:= n\cdot\sigma)$ where $\sigma$ is the standard deviation {of} the process $[X_t]_{t=1}^{T-1}$ and $n \in \{1, 5\}$. 
In \cref{sec:prediction_intervention} we discuss that IV is prediction optimal under arbitrary interventions $\operatorname{do}(X_T \coloneqq x)$. In particular, one would expect that OLS becomes increasingly inferior to IV methods when one increases $n$ and thereby the intervention strength $\operatorname{do}(X_T\coloneqq n\cdot \sigma)$.

We compare prediction for $Y_{T+1}$ via \cref{alg:prediction} in \cref{sec:apppredunderinterv} (with $m=2$ and $l=1$) with
prediction based on the OLS regression 
$Y_{t+1}\sim X_t+X_{t-1}+X_{t-2}+Y_t+Y_{t-1}$ (`OLS'). 
For each repetition and matrix $A$, we obtain CIV and NIV estimates of $\beta$ on a separate sample first, and then obtain predictions for $Y_{T+1}$ following \cref{alg:prediction} using either $\hat\beta_\text{CIV}$ (`CIV') or $\hat\beta_\text{NIV}$ (`NIV').
For each of the $m=100$ random matrices $A$ we compute $\operatorname{error}(\hat Y_{T+1})$ as the mean of the squared prediction error $(\hat Y_{T+1} - Y_{T+1})^2$ over the $s=100$ repetitions.

In \cref{plot:predict}, we plot the error of the OLS estimate against the error of the IV estimate (either CIV or NIV).
When $\do(X_T \coloneqq \sigma)$, the intervention is within the normal range of $X_t$, and the errors of OLS and IV estimates are similar. As we perform a stronger intervention $\do(X_T \coloneqq 5\sigma)$, the OLS error exceeds the IV error in most simulations. This robustness in prediction under intervention is in line with the result in \cref{prop:optimal_prediction}. 

\begin{figure}[t!]
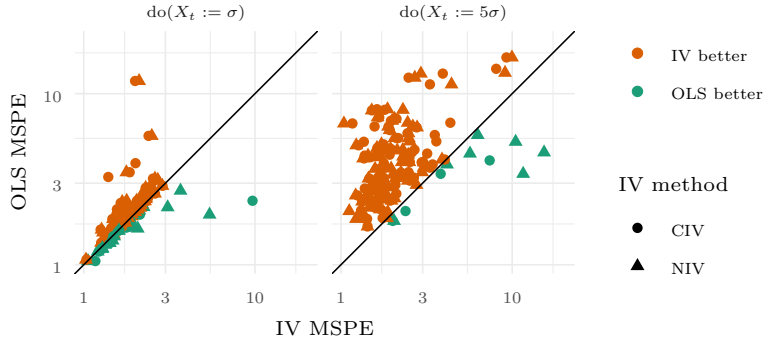

    \centering
    {\scriptsize

     }
    \caption{%
    The figure plots the loss $\operatorname{MSPE}(\hat{Y}_{T+1})$ under an intervention $\operatorname{do}(X_T\coloneqq x)$ when $Y_{T+1}$ is predicted using OLS against the loss when $Y_{T+1}$ is predicted using one of the IV methods we develop for time series. The OLS prediction is based on the regression $Y_{t+1}\sim X_t+X_{t-1}+X_{t-2}+Y_t+Y_{t-1}$ while the IV predictions of $Y_{T+1}$ are based on the procedure discussed in \cref{sec:prediction_intervention} (see \cref{alg:prediction} in \cref{sec:apppredunderinterv}).
    We plot this both for $\operatorname{do}(X_T\coloneqq \sigma)$ and $\operatorname{do}(X_T\coloneqq 5\sigma)$, where $\sigma$ is the standard deviation of $X_t$ in the unintervened distribution. The results show that as the intervention strength increases, the OLS prediction error increases at a faster rate than the IV prediction error. 17 outliers were removed from the right-hand side plot, of which 10 had a larger OLS MSPE than IV MSPE.
    }
    \label{plot:predict}
\end{figure}

\section{Conclusion and Future Work}
In this work, we have developed IV methods for time series data that allow us to identify the causal effect of a process $X$ on a response process $Y$, based on an instrument process $I$ that exhibits memory effects. 
Simple adaptations of ordinary IV estimators generally fail to identify the causal effect due to confounding from the past,
as we show in \cref{failure:naive_iv}.
We have developed the concept of nuisance IV (NIV), see \cref{prop:nuisance-iv}, a marginalization framework for time series graphs, see \cref{thm:gmp}, and have proved
a version of the global Markov property for VAR{($p$)} processes, see~\cref{prop:dream-theorem}.
Based on these principles, we propose two classes of estimation methods that properly adjust for confounding from the past: one based on choosing the correct conditioning set (CIV), see \cref{prop:ts-civ}, and another one based on nuisance regressors (NIV), see \cref{prop:ts-niv}. 
The procedures find solutions to moment conditions that, in their population version, are satisfied for the true causal parameters. Unlike in the i.i.d.\ case, the identifiability conditions (which are usually rank conditions) do not have a simple interpretation.
\cref{prop:identifiability} provides necessary and sufficient conditions on the parameters of the underlying data-generating process for the causal parameter being the unique solution to the corresponding moment equation.

{The experiments show that while several different choices of instruments, conditioning sets and nuisance regressors allow us to consistently identify the causal effect, the modelling choices impact the finite sample performance.}
For example, for identifiability in the case of NIV, we only need the number of lags used as instruments{, $m$, to be so} large that $d_I\cdot m=d_X+d_Y${,} but using more lags may help in finite samples in that this shrinks the upper tail of the error distribution, see \cref{plot:compare} (left).

We have further argued that identifying the causal effect may be of interest not only for causal inference, but also for prediction of $Y$ under the intervention $\operatorname{do}(X_t := x)$, where the minimal expected squared error can be obtained by a mix of causal parameters and regression coefficients, see \cref{prop:optimal_prediction} and \cref{sec:prediction_simulation}.

For future work, it may be fruitful to develop principled techniques for deciding which estimator yields the best finite sample performance \citep[see, e.g.,][Chapter~4]{henckel2021graphical} and to construct confidence statements, either based on Appendix~\ref{sec:asymptotic-variance-ts} or other techniques \citep[e.g.,][]{newey1987simple,shah2020hardness}.
{
    In settings where both NIV and CIV estimators are valid, it would be interesting to combine the estimating equations to obtain a single estimator and to analyze both the statistical and identifiability properties of such an approach.
    Our results assume that the response variable $Y_t$ is continuous. 
    However, we think one can generalize the NIV and CIV procedures, for example to a setting where $Y_t$ follows a generalized linear model.
    Further, it would be interesting to understand whether the identifiability obtained by using lagged instruments in \cref{prop:identifiability} can be achieved in processes that do not satisfy the linearity assumption.
    In addition,  it is known that overidentification can help to detect or even correct for certain types of model violations \citep[see, e.g.,][]{hansen1982large,chao2014testing}. It may be worthwhile to explore to which extent such methods can also be used in the time series setting -- in particular, as the NIV approach comes with a natural way of obtaining a certain degree of overidentification by including more of the previous time points (see also the discussion in Section~\ref{sec:Sim_compare}).
    It could further be interesting to extend the results to VARMA processes \citep[e.g.,][]{Scherrer2019}.}
Finally, as for the i.i.d.\ case \citep{Imbens2009,Chesher2003, Saengkyongam2022}, considering 
independence, rather than vanishing covariances may yield stronger identifiability results but may come with computational and statistical challenges.

\acks{%
We thank Aurélien Bibaut for pointing out typos in an earlier version of \cref{fig:marginalized-var1} and Ignacio Gonzalez Pérez for discovering and discussing potential fixes for errors in the proofs of \cref{thm:gmp,prop:dream-theorem}.
During parts of this project, NT, RS and JP were supported by  a research grant (18968) from VILLUM FONDEN and JP and SW by the Carlsberg Foundation. 
}

\bibliography{other_tex_files/references.bib}

\newpage
\appendix

\section{Additional details for \texorpdfstring{\cref{sec:setup}}{}}
\subsection{Relation to VARMA Processes}
In this section, we discuss that the partially observed VAR($1$) process can also be viewed as a VARMA($p,q$) process.
In this perspective,
the difficulty of identifying $\beta$ when $H$ is unobserved
is linked to the non-uniqueness of vector autoregressive moving average (VARMA) process representations.
The observed process $[I^\top,X^\top,Y^\top]^\top_{t\in\mathbb{Z}}$
can be obtained as a linear transformation of the VAR($1$) process $S$
and as such it
has a VARMA($p,q$) process representation
where
$p \leq d$
and
$q \leq (d - 1)$
\citep[Corollary 11.1.1]{lutkepohl2005new}.
Intuitively, the dependencies between $I, X, Y$ induced by the unobserved $H$ process can instead be modelled by serially correlated errors and higher-order memory. In contrast to a VAR{($p$)} process, however, the parameters of a VARMA process in standard form are not identified and different parameter settings may induce the same distribution over $[I^\top,X^\top,Y^\top]^\top_{t\in\mathbb{Z}}$
\citep[Chapter 12.1]{lutkepohl2005new}.
As such, it is not straight-forward to obtain $\beta$ from a VARMA representation of the observed process, even when choosing a canonical representation such as the echelon form or the final equations form.
In this work, we propose another approach and describe how to exploit the instrumental variables idea to identify $\beta$ when $H$ is unobserved, without needing to estimate all of $A$.

\subsection{Observational Equivalence}\label{sec:obs_equiv}
Without an instrument, the causal effect $\beta$ is, due to the hidden confounding, not identifiable in general.
In a fully observed Gaussian VAR($1$) process, the parameter matrix (which contains the causal effect) and the covariance matrix of the noises are uniquely determined by the distribution, 
and can be identified by least squares regression on the previous time step, for example \citep[Chapter~11]{hamilton1994time}. This is not the case if parts of the system are unobserved.
We consider a VAR($1$) process over $H = \vecin{H^1, H^2}^\top, X$, and $Y$, where $H$ is latent, and provide two different sets of parameters which entail the same observational distribution, that is, the same joint distribution of the observed process $[S_{XY,t}]_{t\in\Z} = [X_t^\top, Y_t^\top]^\top_{t\in\Z}$. 
Nevertheless, the causal effects in the two cases are different (one is $0$, the other is $b \neq 0$), and so are the induced interventional distributions when intervening on $X$, see \cref{sec:interventions} below.
Consider the two coefficient matrices 
\begin{align}
    \label{eq:obs_equiv_matrices}
    A_1 := 
    \begin{matrix} \text{{\tiny H}}^{\text{\tiny 1}} \\ \text{{\tiny H}}^{\text{\tiny 2}} \\ \text{{\tiny X}} \\ \text{{\tiny Y}}\end{matrix}
    \mat{a & 0 & 0 & 0 \\ c & 0 & 0 &0 \\ c & 0 & 0& 0 \\ 0 & b & 0 & 0}
    \quad \text{ and } \quad
    A_2 := 
    \begin{matrix} \text{{\tiny H}}^{\text{\tiny 1}}\\ \text{{\tiny H}}^{\text{\tiny 2}} \\ \text{{\tiny X}} \\ \text{{\tiny Y}}\end{matrix}
    \mat{a & 0 & 0 & 0 \\ 0 & a & 0 & 0 \\ 0& c & 0 & 0 \\ 0&0 & b & 0},
\end{align}
  with coefficients $a \in (-1, 1)$ and $b,c\in \R \setminus \{0\}$.
\begin{figure}[t]
\centering
\begin{tikzpicture}
    \tikzstyle{every path} = [thin, ->];
    \node (H12) at (0,1.5){$H^2_{t-2}$};
    \node (H11) at (2,1.5){$H^2_{t-1}$};
    \node (H10) at (4,1.5){$H^2_{t}$};
    \node (H22) at (0,0.75){$H^1_{t-2}$};
    \node (H21) at (2,0.75){$H^1_{t-1}$};
    \node (H20) at (4,0.75){$H^1_{t}$};
    \node (Y2) at (0,0){$Y_{t-2}$};
    \node (Y1) at (2,0){$Y_{t-1}$};
    \node (Y0) at (4,0){$Y_{t}$};
    \node (X2) at (0,2.25){$X_{t-2}$};
    \node (X1) at (2,2.25){$X_{t-1}$};
    \node (X0) at (4,2.25){$X_{t}$};

    \path [->] (H11) edge  (H10);
    \path [->] (H12) edge node[above] {$a$}(H11);
    \path [->] (H11) edge (H20);
    \path [->] (H12) edge node[below] {$c$} (H21);
    \path [->] (H21) edge node[below] {$b$} (Y0);
    \path [->] (H22) edge (Y1);
    \path [->] (H11) edge (X0);
    \path [->] (H12) edge node [above] {$c$} (X1);
    \draw[draw=black, fill opacity=0.5, dashed] (-1, 0.375) rectangle (5, 1.875);
    \node at (5.6, 1.5) {Latent};

    \node (H2b) at (8, 2.25){$H^2_{t-2}$};
    \node (H1b) at (10,2.25){$H^2_{t-1}$};
    \node (H0b) at (12,2.25){$H^2_{t}$};
    \node (H2) at (8, 1.5){$H^1_{t-2}$};
    \node (H1) at (10,1.5){$H^1_{t-1}$};
    \node (H0) at (12,1.5){$H^1_{t}$};
    \node (X2) at (8, 0.75){$X_{t-2}$};
    \node (X1) at (10,0.75){$X_{t-1}$};
    \node (X0) at (12,0.75){$X_{t}$};
    \node (Y2) at (8,0){$Y_{t-2}$};
    \node (Y1) at (10,0){$Y_{t-1}$};
    \node (Y0) at (12,0){$Y_{t}$};
    \path [->] (H1) edge (H0);
    \path [->] (H2) edge node[above] {$a$} (H1);
    \path [->] (H1) edge (X0);
    \path [->] (H2) edge node[below] {$c$}(X1);
    \path [->] (X1) edge node[above] {$b$}(Y0);
    \path [->] (X2) edge (Y1);
    \path [->] (H2b) edge node[above] {$a$} (H1b);
    \path [->] (H1b) edge (H0b);
    \draw[draw=black, fill opacity=0.5, dashed] (7, 1.125) rectangle (13, 2.625);
    
    \node at (13.6, 2.5) {Latent};
\end{tikzpicture}

\caption{Illustration of two different causal mechanisms 
($S^{(1)}$, left, and $\tilde{S^{(2)}}$, right)
which are observationally equivalent for any $b,c \in \R$ and $a \in [-1,1]$ and Gaussian noise distribution $\mathcal{N}(0, \operatorname{Id})$.}
\label{fig:obs_equiv}
\end{figure}
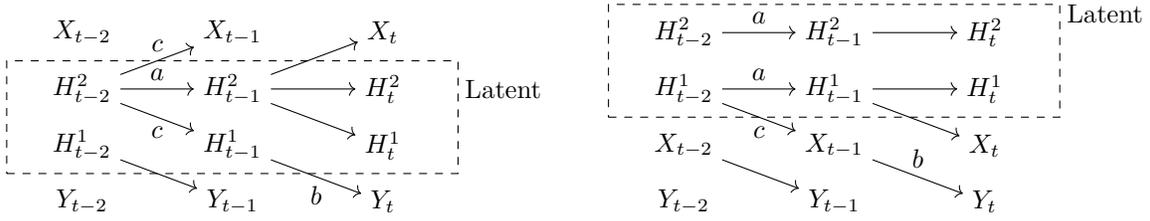
Consider the processes $S^{(1)}$ and $S^{(2)}$ satisfying
\begin{align}
    \label{eq:obs_equiv_example}
    S^{(i)}_{t+1} = A_i S^{(i)}_t + \epsilon^{i}_{t+1},
\end{align}
  with $\epsilon^{i}_t \sim \mathcal{N}(0,\operatorname{Id})$. 
  \Cref{fig:obs_equiv} depicts parts of the corresponding full time graphs.
  Let $S_{XY}^{(1)}$ and $S_{XY}^{(2)}$ denote the subprocesses where only $X$ and $Y$ are observed. 
  The following result shows, that $S_{XY}^{(1)}$ and $S_{XY}^{(2)}$ are identically distributed, that is
  the two models arising from $A_1$ and $A_2$ are \emph{observationally equivalent} \citep[e.g.,][]{rothenberg1971identification}.
\begin{restatable}{proposition}{equivalentDistributions}
\label{prop:two_equivalent_distributions}
    Let $S^{(1)}, S^{(2)}$ be the processes defined from \cref{eq:obs_equiv_example} with respective parameter matrices $A_1$ and $A_2$ from \cref{eq:obs_equiv_matrices} and $\mathcal{N}(0, \operatorname{Id})$-distributed noise.
    Then the observed subprocesses $S^{(1)}_{XY}$ and $S^{(2)}_{XY}$ are identically distributed. 
\end{restatable}
\begin{proof}
Since each of the processes are jointly Gaussian with zero mean, their distributions are uniquely determined by the autocovariance matrices.
This means that the observed processes are identically distributed if and only if $\E(S_{XY,t}^{(1)} S^{(1)^\top}_{XY,t-s}) = \E(S^{(2)}_{XY,t} S^{(2)^\top}_{XY,t-s})$ for all $s \geq 0$.
For $i \in \{1, 2\}$, the $s$-th autocovariance of $S^{(i)}$ is given by: 
$$\label{eq:autocovariance_formula}
    \E(S^{(i)}_t S^{(i)^\top}_{t-s}) = \sum_{k = 0}^\infty A_i^{k+s} \operatorname{Id} A_i^{k^\top}.
$$
Observe that
$$
A_1^k = \mat{a^k & 0 & 0 & 0 \\ a^{k-1}c & 0 & 0 &0 \\ a^{k-1}c & 0 & 0& 0 \\ a^{k-2} bc & 0 & 0 & 0}
\quad 
A_2^k = \mat{a^k & 0 & 0 & 0 \\ 0 & a^k & 0 & 0 \\ 0 & a^{k-1}c & 0 & 0 \\ 0 & a^{k-2}bc &0& 0}
$$
for $k \geq 2$. 
Consequently, 
\begin{align*}
A_1^{k+s} A_1^{k^\top} &= 
\mat{
a^{2k + s} & a^{2k + s - 1}c & a^{2k + s - 1}c & a^{2k+s-2}bc \\ 
* & a^{2k + s -2}c^2 & a^{2k + s-2}c^2 & a^{2k + s-3}bc^2 \\
* & * & a^{2k + s-2}c^2 & a^{2k + s-3}bc^2 \\
* & * & * & a^{2k + s-4}b^2c^2
}\ \text{and}
\\
A_2^{k+s} A_2^{k^\top} &=
\mat{
a^{2k+s} & 0 & 0 & 0 \\
* &a^{2k+s} & a^{2k+s-1}c & a^{2k+s-2}bc \\
*&* & a^{2k+s-2}c^2 & a^{2k+s-3}bc^2 \\
*&* & * & a^{2k+s-4}b^2c^2
}
\end{align*}
for $k \geq 2$ and $s \geq 0$, where the asterisks are given by symmetry of the matrices. 
For the case $k = 1, s = 0$, we have:
$$
A_1  A_1^{\top} = 
\mat{a^{2} & a^2 c & a^2c & 0 \\ 
* & c^2 & c^2 & 0 \\
* & * & c^2 & 0 \\
* & * & * & b^2
}
\quad\text{and}\quad
A_2 A_2^{\top} = 
\mat{
a^2 & 0 & 0 & 0 \\
*&a^{2} & ac & 0\\
*&* & c^2 & 0 \\
*&* & * & b^2
},
$$
and if $k = 1, s \geq 1$:
$$
A_1^{1+s}  A_1^{\top} = 
\mat{
a^{2 + s} & a^{1+s}c & a^{1+s}c & 0 \\ 
a^{1+s}c & a^{s}c^2 & a^{s}c^2 & 0 \\
a^{1+s}c & a^{s}c^2 & a^{s}c^2 & 0 \\
a^s bc & a^{s-1} c^2 b & a^{s-1} c^2 b & 0
}
\\\quad\text{and}\quad
A_2^{1+s}  A_2^{\top} =
\mat{
a^{2+s} & 0 & 0 & 0 \\
0 & a^{2 + s} & a^{1+s}c & 0 \\ 
0 & a^{1+s}c & a^{s}c^2 & 0 \\
0 & a^s bc & a^{s-1} c^2 b & 0
}.
$$
For any of the above matrices $M$, let $M_{XY}$ denote the $2\times 2$ submatrix in the bottom right corner, relating to the $X,Y$ subprocess. In all of the above cases, these coefficients relating to the $X,Y$ subprocess coincide, that is for any $k, s \geq 0$, $(A_1^{k+s} A_1^{k^\top})_{XY}
=(A_2^{k+s} A_2^{k^\top})_{XY}$, and since therefore 
$(\sum_{k = 0}^\infty A_1^{k+s} A_1^{k^\top})_{XY} 
= \sum_{k = 0}^\infty (A_2^{k+s} A_2^{k^\top})_{XY}$, 
it follows that $\E(S^{(1)}_{XY,t} S^{(1)^\top}_{XY,t-s}) = \E(S^{(2)}_{XY,t} S^{(2)^\top}_{XY,t-s})$ for all $s \geq 0$. 
\end{proof}

\subsection{Structural Causal Models and Interventions }\label{sec:interventions}
We provide a formal introduction to structural causal models (SCMs) and interventions, which motivates the notion of a causal effect. For a more detailed introduction, see \citet{pearl2009causality, Bongers2020}. 

An SCM consists of a tuple $\Pi=(\mathcal{S},P_\varepsilon)$ where $\mathcal{S}$ is a set of structural assignments 
and $P_\varepsilon$ describes the joint distribution of the error terms.
For a finite collection of variables $S^1, \ldots, S^d$, 
with structural assignments $\mathcal{S} \coloneqq \{f^1, \ldots, f^d\}$ and noise distribution $P_{\epsilon} = P^1\otimes \cdots\otimes P^d$, for each $j = 1, \ldots, d$, the structural equation of $S^j$ is 
\begin{equation*}
    S^j \coloneqq f^j(\PA_j, \epsilon^j),
\end{equation*}
where $\epsilon^j$ is distributed according to $P^j$ and the \emph{parents} $\PA_j$ is a subset of $\{S^1, \ldots, S^d\}\setminus\{S^j\}$. The SCM induces a corresponding graph $\mathcal{G}$ over nodes $\{1, \ldots, j\}$, where we draw an edge from $j'$ to $j$ if $S^{j'}\in\PA_j$. We assume that the parent sets are such that $\G$ is acyclic. 

Similarly, we interpret the VAR{($p$)} process described in \cref{eq:varp} as a structural causal model over an infinite number of nodes $\vecin{S_t}_{t\in\Z}$. In this case, the structural assignments are $S_t:=AS_{t-1}+\varepsilon_t$, $t\in\mathbb{Z}$.
Here, the error terms $\varepsilon_t$
are assumed to be i.i.d.\ over time, 
distributed according to $P_\varepsilon$. 
Furthermore, $P_\varepsilon$ is a product distribution, and thus the error terms are jointly independent. The SCM entails an observational distribution on the variables $[S_t]_{t\in \Z}$ which we denote by $P_S^{\Pi}$. 

Formally, an intervention on an SCM is a replacement of one or more of the structural assignments at one or more time points. 
Such a replacement induces a new SCM that we denote by $\Tilde{\Pi}=(\Tilde{\mathcal{S}},\Tilde{P_\varepsilon})$. 
An example of an intervention on the above VAR{($p$)} process is to fix the value of $X$ for some specific time point $t_0$ -- we write this intervention as $\operatorname{do}(X_{t_0}:=x)$. 
Under this intervention, for $t \neq t_0$, the process still satisfies the original SCM, including assumptions on the noise variables. 

The interventional distribution of $\Pi$ under this intervention is defined as  $P_S^{\Pi;\operatorname{do}(X_{t_0}:=x)}:=P_S^{\Tilde{\Pi}}$.
In general, the interventional distribution of $\Pi$ under an intervention is the 
distribution that is induced by the SCM $\Tilde{\Pi}$ obtained by replacing some of the structural assignments. 
We require that this distribution exists and is unique. 
Depending on the application at hand, 
several interventions on the process are useful, including, for example, changing the dynamics of one component for all time points \citep{Peters2020dynamics}.
In this work, we focus on an intervention at a particular time point. 
When performing such an intervention  
$\operatorname{do}(X_{t_0}:=x)$ on 
$X_{t_0}$, we have that 
$\frac{\partial}{\partial x}\E_{P_Y^{\Pi;\operatorname{do}(X_{t_0}:=x)}} [Y_{t_0+1}] = \alpha_{Y,X}$.
This  motivates calling 
$\beta = \alpha_{Y,X}$ the {direct} causal effect from $X$ to $Y$
{in our parametric model (while the effect may be mediated in another parametrization, e.g., at finer temporal resolution)}.
In several applications, the causal effect $\beta$ is of interest by itself because it yields insight into understanding the causal structure of the problem. 
The causal effect, however, also comes with another benefit: it is 
optimal for
prediction under intervention.
We discuss this point of view in \cref{sec:prediction_intervention}.

\subsection{Defining Multivariate Total Causal Effects}\label{sec:appendix-TCE}
In some cases, the effect we want to estimate may be more general than a single entry in one of the coefficient matrices $A_k$. 
In \cref{sec:def_notation} we define the total causal effect of a single variable $S^{i}_{t-l}$ on $S^{j}_{t}$ as
\begin{align*}
    \left(\sum_{\substack{1 \leq l_1, \ldots, l_m \leq p \\ l_1 + \cdots + l_m = l}} A_{l_1}\cdots A_{l_m}\right)_{j,i}, 
\end{align*}
where $A_l$ are the parameter matrices of the VAR($p$) process. Using the method of path coefficients \citep{wright1934method}, we now  provide a more general definition, where $\cX$ may contain multiple variables.
\begin{definition}[Path coefficients]
    Let $S$ be a VAR($p$) process,  
    let $Y_t\coloneqq S^{i_0}_{t}$ be a subprocess of $S$ and let $\cX_t = \vecin{{S_{t-l_1}^{i_1}}^\top, \ldots, {S_{t-l_m}^{i_m}}^\top}^\top$ be a collection of subprocesses of $S$. 
    For $j = 1, \ldots, m$, we define a \emph{$S_{t-l_j}^{i_j}$-causal path} to be a directed path  from $S_{t-l_j}^{i_j}$ to $Y_t$ in the full time graph of $S$ that does not intersect any other $S_{t-l_{j'}}^{i_{j'}}$, for $j'\neq j$. 
    For a $S_{t-l_j}^{i_j}$-causal path $\pi: S_{t-l_j}^{i_j} \stackrel{e_1}{\rightarrow} \cdots \stackrel{e_d}{\rightarrow}S_{t}^{i_0}$, we define the \emph{path coefficient} to be the product of linear coefficients along $\pi$, $c_\pi := \prod_{k=1}^d a_k$, where $a_k$ denotes the entry in the coefficient matrix $A_v$, for the lag $v$ corresponding to the edge $e_k$.
\end{definition}
We can now define the total causal effect.
\begin{definition}[Total Causal Effect]\label{def:causal_effect_varp}
    Let $S$ be a VAR($p$) process, let $Y_t = S^{i_0}_{t}$ be a subprocess of $S$ and let $\cX = \vecin{{S_{t-l_1}^{i_1}}^\top, \ldots, {S_{t-l_m}^{i_m}}^\top}^\top$ be a collection of subprocesses of $S$. 
    For $j = 1, \ldots, m$, the \emph{total causal effect}, $\beta^{j}$, of $S_{t-l_j}^{i_j}$ on $Y_t$ is the sum of path coefficients $c_\pi$ over all $S_{t-l_j}^{i_j}$-causal paths $\pi$ from $S_{t-l_j}^{i_j}$ to $Y_t$,
    $\beta^i := \sum_{\text{$S_{t-l_1}^{i}$-causal paths } p} c_p$.
    Similarly the total causal effect, $\beta$, of $\cX_t$ on $Y_t$ is the 
    concatenation of these, $\beta := \vecin{{\beta^1}^\top, \ldots, {\beta^m}^\top}^\top$.
\end{definition}

\section{Additional details for \texorpdfstring{\cref{sec:niv}}{}}
\subsection{Asymptotic variances for i.i.d.\ estimators }\label{sec:asymptotic-variance}
Drawing on existing results \citep{hall2005generalized}, we now provide formulas for the asymptotic variances of the NIV and CIV estimators. 
If a unique solution $(\beta, \alpha)$ to \cref{eq:niv-identifying-equation} exists, the asymptotic distribution of the $\hat{\beta}_{\NIV,T}$ estimator in the i.i.d.\ setting, discussed in \cref{sec:nuisance-iv},
with the weight matrix $W \coloneqq \E[\cI\cI^\top]^{-1}$ (which asymptotically is optimal, see \cref{sec:intro-civ}) is given by 
\begin{align*}
    \sqrt{T}(\hat{\beta}_{\NIV,T}-\beta)\xrightarrow[]{d}\mathcal{N}(0,\Sigma_1),
\end{align*} 
for $T \rightarrow \infty$, where the asymptotic variance $\Sigma_1$ is given by
\begin{align*}
    \Sigma_1 \coloneqq \big(\E (\tilde{\cX}\cI^\top)K^{-1} \E(\tilde{\cX}\cI^\top)^\top\big)^{-1},
\end{align*}
where $\tilde{\cX}\coloneqq\begin{matrix} [\cX^\top, \cZ^\top]^\top\end{matrix}$,
$K=\E((Y-\beta \cX - \alpha\cZ)^2)\E(\cI\cI^\top)$.
$\Sigma_1$ is a $(d_\cX+d_\cZ)\times(d_\cX+d_\cZ)$ matrix, with the top-left $d_\cX\times d_\cX$ sub-matrix describing the asymptotic variance of $\cX$.

Similarly, the asymptotic distribution of $\hat{\beta}_{\CIV, T}$ in the i.i.d.\ setting, discussed in \cref{sec:intro-civ}, with weight matrix $W \coloneqq \E[\var(\cI| \cB)]^{-1}$ is
\begin{align*}
    \sqrt{T}(\hat{\beta}_{\CIV,T}-\beta)\xrightarrow[]{d}\mathcal{N}(0,\Sigma_2),
\end{align*} 
where
\begin{align*}
    \Sigma_2 \coloneqq \big(\E[ \cov(\cX,\cI|\cB)]K^{-1} \E[\cov(\cX,\cI|\cB)^\top]\big)^{-1},
\end{align*}
and $K=\E[\var((Y-\beta \cX)^2|\cB)]\E[\var(\cI|\cB)]$.

\subsection{Asymptotic variances for estimators in time series} \label{sec:asymptotic-variance-ts}
Closed-form expressions for the asymptotic variances of the NIV and CIV estimators can also be found for the VAR{($1$)} process presented in \cref{sec:iv-for-time-series}, but these are slightly more involved than for the i.i.d.\ setting presented in \cref{sec:asymptotic-variance}. 
Assume the setting as described in~\cref{prop:ts-niv} and consider the NIV estimator $\hat{\beta}_{\NIV,T}$. We then have
\begin{align*}
    \sqrt{T}(\hat{\beta}_{\NIV,T}-\beta)\xrightarrow[]{d}\mathcal{N}(0,\Sigma_1),
\end{align*} 
where the asymptotic variance $\Sigma_1$ is given by
\begin{align*}
    \Sigma_1 \coloneqq \big(\E (\tilde{\cX}\cI^\top)K^{-1} \E(\tilde{\cX}\cI^\top)^\top\big)^{-1},
\end{align*}
where $\tilde{\cX}:=\begin{matrix} [X_{t-1}^\top,& Y_{t-1}^\top]^\top\end{matrix}$, $\cI:=\{I_{t-2},...,I_{t-m-1}\}$ and  $K:=\lim_{T\rightarrow\infty}Var\bigg(\frac{1}{\sqrt{T}}\sum_{t=1}^{T}(Y_t-\beta \cX - \alpha\cZ)\cI^{\top}\bigg)$ using the optimal choice of weight matrix, $W=K^{-1}$, 
see \citet[][Chapter~3]{hall2005generalized}.

Now for the CIV estimator, see \cref{prop:ts-civ}, the asymptotic distribution of $\hat{\beta}_{\CIV, T}$ is
\begin{align*}
    \sqrt{T}(\hat{\beta}_{\CIV,T}-\beta)\xrightarrow[]{d}\mathcal{N}(0,\Sigma_2),
\end{align*} 
where
\begin{align*}
    \Sigma_2 \coloneqq \big(\E[ \cov(X_{t-1},I_{t-2}|\cB_t)]K^{-1} \E[\cov(X_{t-1},I_{t-2}|\cB_t)^\top]\big)^{-1},
\end{align*}
and $K:=\lim_{T\rightarrow\infty}Var\bigg(\frac{1}{\sqrt{T}}\sum_{t=1}^{T}(r_{Y_t}-\beta r_{X_{t-1}})r_I^{\top}\bigg)$ with $r_i:=i-\E[i|\cB_t]$
using the optimal choice of weight matrix, $W=K^{-1}$,
see \citet[][Chapter~3]{hall2005generalized}.

\section{Additional details for \texorpdfstring{\cref{sec:iv-for-time-series}}{}}
{
\subsection{Failure of \texorpdfstring{$\iv{X_{t-1}}{Y_t}(I_{t-2})$}{}}
\label{app:failure-naive-iv}
In \cref{sec:iv-for-time-series}, we discuss that using $I_{t-2}$ as an instrument for the effect $X_{t-1}\rightarrow Y_t$ is not valid. This is simply a consequence of \cref{assump:civ-d-sep-tce} not being met, since $I_{t-2}$ correlates with $Y_t$ through other causal paths than $I_{t-2}\rightarrow X_{t-1} \rightarrow Y_t$. 
Here we formally prove that this procedure fails and derive the bias of the resulting estimate. 

\begin{restatable}[Failure of naive IV adaption]{proposition}{failureNaive}
\label{failure:naive_iv}
Consider a VAR($1$) process $S = \vecin{I_t^\top, X_t^\top, H_t^\top, Y_t^\top}^\top_{t\in\Z}$ satisfying \cref{assump:iv} with $d_I = d_X = d_H = d_Y = 1$.
If $\cov(X_{t-1},I_{t-2})\neq 0$ and $\alpha_{I,I}\alpha_{Y,Y}\neq 1$, the $\iv{X_{t-1}}{Y_t}(I_{t-2})$ estimator $\hat\beta$ converges in probability to 
\begin{equation*}
    (1 - \alpha_{I,I} \alpha_{Y,Y})^{-1} \beta.
\end{equation*}
Consequently, $\hat\beta$ is in general not consistent for the causal effect $\beta$ of $X_{t-1}$ on $Y_t$, unless $I$ or $Y$ do not have any autoregressive structure, that is, $\alpha_{I,I} = 0$ or $\alpha_{Y,Y}= 0$.
\end{restatable}

}

\subsection{Example of a distribution that does not satisfy the rank requirement in \texorpdfstring{\cref{prop:identifiability}}{}}\label{appendix:ts-niv-examples}
In \cref{se:meth2}, we have developed a criterion for identifiability that depends on the parameter matrix of a {VAR($1$)} process satisfying \cref{assump:iv}. We have showed in \cref{thm:almost_sure_identify} that if parameter matrices are drawn from a distribution with density with respect to Lebesgue measure, then the identifiability criterion holds almost surely. 
In this section, we provide an example of a parameter matrix that does not satisfy the criterion. 
\begin{example} \label{ex:not_identifiable}
	Consider the case where $d_X > 1, d_I=1,$ and
	$\alpha_{X,X}=\diag(c, \ldots, c)$ for a $c \in \R$.
	By \cref{prop:identifiability}, $\beta$ is not identifiable by NIV: this follows because $A_{XY}$ is a lower triangular matrix with $c, \ldots, c$ ($d_X$ times) and $\alpha_{Y,Y}$ on the diagonal, and the Jordan form $J$ is a diagonal matrix with the same diagonal entries. Hence there are $d_X$ Jordan blocks with the same eigenvalue $c$ so the causal effect $\beta$
	is not identified by NIV.
	On the contrary, when $\alpha_{X,X} = \diag(c_1, \ldots, c_{d_X})$ (see \cref{fig:ex-figure-identifiability}) where $c_i \neq c_j$ for all $i\neq j$, $\beta$ is identified by NIV if also 
	$\alpha_{Y,Y} \neq c_i$ for all $i$. 
	\begin{figure}
	    \centering
    \begin{tikzpicture}[>=latex,font=\sffamily]
        \node (X1) at (0, 2) {$X_{t-1}^{(1)}$};
        \node (X2) at (0, 1) {$X_{t-1}^{(2)}$};
        \node (Ypast) at (0, 0) {$Y_{t-1}$};
        \node (Y) at (3, 0) {$Y_t$};
    
        \draw[->] (X1) --node[midway,fill=white]{$c_1$} (Y);
        \draw[->] (X2) --node[midway,fill=white]{$c_2$} (Y);
        \draw[->] (Ypast) --node[midway,fill=white]{$\alpha_{Y,Y}$} (Y);
    \end{tikzpicture}
	    \caption{Subgraph of the full-time graph of the structure discussed in \cref{ex:not_identifiable} when $d_X = 2$. For simplicity, we do not draw nodes corresponding to the instrument, $I$, and the confounder, $H$.}
	    \label{fig:ex-figure-identifiability}
	\end{figure}
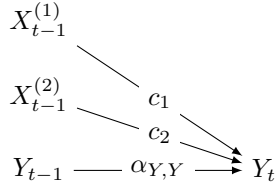
\end{example}
\begin{example}
In the case of $d_X = 1, d_I=1,$ and $\alpha_{X,I} \neq 0$, $\beta$ is identifiable. If, for example, $\alpha_{Y,Y} = \alpha_{X,X} =: \alpha$, we have, 
	\begin{align*}
		A_{XY} = \mat{\alpha & 0 \\ \beta & \alpha} = \underbrace{\mat{0 & 1 \\ \beta & 0}}_{=:P} \mat{\alpha & 1 \\ 0 & \alpha}\underbrace{\mat{0 & \frac{1}{\beta} \\ 1 & 0}}_{=P^{-1}}.
	\end{align*}
	This has only one Jordan block with algebraic multiplicity $m = 2$ and 
	\begin{equation*}
	    \left(P^{-1}\mat{\alpha_{X,I}\\0}\right)_{m} = \vec{1, 0}\mat{\alpha_{X,I} \\ 0} = \alpha_{X,I} \neq 0,
	\end{equation*}
	where $(\cdot)_{m}$ refers to the $m$-th entry,
	so by \cref{prop:identifiability}, $\beta$ is identifiable (and by similar arguments the same holds if $\alpha_{X,X} \neq \alpha_{Y,Y}$.
\end{example}

\subsection{Algorithm for Prediction under Interventions} \label{sec:apppredunderinterv}
This section contains \cref{alg:prediction}, the procedure for using a causal effect to predict under interventions, as discussed in \cref{sec:prediction_intervention}.
\begin{algorithm}
\begin{algorithmic}[1]
	\Statex \textbf{Input}: Causal parameter $\beta$, sample $\mathbf{X} = \vec{\mathbf{X}_1, \ldots, \mathbf{X}_{t-1}} \in \R^{d_X \times (t-1)}$, $\mathbf{Y} = \vec{\mathbf{Y}_1, \ldots, \mathbf{Y}_t} \in \R^{1 \times t}$, 
	interventional value $x$, lag parameters $m$ and $l$
	\Statex
	\State Compute the residual process $r_s \coloneqq \mathbf{Y}_{s+1} - \beta \mathbf{X}_{s}$ for $s = 1, \ldots, {t-1}$.
	\State For $s>\max(k,m)$, linearly regress 
	$r_s$ on $\{\mathbf{X}_{s-k}, k = 1, \ldots, m\}$ and $\{\mathbf{Y}_{s-j}, j = 0, \ldots, l\}$ to obtain coefficients $\hat\alpha_{Y,X}^{k}$ and $\hat\alpha_{Y,Y}^{j}$.
	\Statex 
	\Statex \textbf{Output}: 
	\State Prediction $\hat{\mathbf{Y}}_{t+1} = \beta x + \sum_{k = 1}^m \hat\alpha_{Y,X}^{k} \mathbf{X}_{t-k} + \sum_{j = 0}^l \hat\alpha_{Y,Y}^{j} \mathbf{Y}_{t-j}$.
\end{algorithmic}
\caption{Linear prediction under the intervention $\operatorname{do}(X_t := x)$
}
\label{alg:prediction}
\end{algorithm}

{
    \subsection{Identifying equations for time-inhomogeneous effects}\label{app:time-inhomogeneous-effect}
    Throughout the main part, we assume that the target parameter $\beta$ is a constant. We now show, that if instead $\{\beta_t\}_{t\in\Z}$ is a sequence of random variables (not necessarily i.i.d.), then we can apply our methodology to estimate $\E[\beta_t]$.
    For a VAR($1$) process $S$, we make the following modification of \cref{assump:iv}
    \begin{alt-alt-assumpenum}
        \setcounter{alt-alt-assumpenumi}{1}
        \item \label{assump:iv-time-varying} Assume that  $S$ satisfies \cref{assump:iv}
        except that for each $t\in\Z$, the structural assignment of $Y_t$ is $Y_t = \beta_t X_{t-1} + \alpha_{Y,H}H_{t-1} + \alpha_{Y,Y}Y_{t-1} + \epsilon_t^Y$, where $\{\beta_t\}_{t\in\Z}$ is a sequence of random variables. 
    \end{alt-alt-assumpenum}
    
    The following two results generalize \cref{prop:ts-civ,prop:ts-niv}, by stating that the identifying equations used in \cref{prop:ts-civ,prop:ts-niv} to identify $\beta$, identify $\E[\beta_t]$ in the case of a time inhomogeneous effect.
    \begin{restatable}[Identification with conditioning set and time inhomogeneous effect]{proposition}{tsCivInhomogeneous}
        \label{lemma:ts-civ-inhomogeneous}
        Consider a VAR($1$) process $S = \vecin{I_t^\top, X_t^\top, H_t^\top, Y_t^\top}^\top_{t\in\Z}$ satisfying \cref{assump:iv-time-varying}, ensuring that all quantities below are well-defined.
        Let either $\cB_t \coloneqq \{I_{t-3}\}$ or $\cB_t \coloneqq \{I_{t-3}, X_{t-2}, Y_{t-1}\}$ and let $b \coloneqq \E[\beta_t]$ be the expectation of $\beta_t$. 
        If $\beta_t \indep (X_{t-1}, I_{t-2}, \cB_t)$, the following two statements hold.
        (i)
        $b$ satisfies the CIV moment condition $\E[\cov(Y_t - b X_{t-1}, I_{t-2}|\cB_t)] = 0$.
        (ii) Furthermore, if $\E[\cov(X_{t-1}, I_{t-2} | \cB_t)]$ has rank $d_X$, then $b$ is identified by $\civ{X_{t-1}}{Y_t}(I_{t-2}| \cB_t)$.
    \end{restatable}
    \begin{restatable}[Identification with nuisance regressor and time inhomogeneous effect]{proposition}{tsNivInhomogeneous}
        \label{lemma:ts-niv-inhomogeneous}
        Consider a VAR($1$) process $S = \vecin{I_t^\top, X_t^\top, H_t^\top, Y_t^\top}^\top_{t\in\Z}$ satisfying \cref{assump:iv-time-varying}, ensuring that all quantities below are well-defined.
        Let $\cI_t \coloneqq \{I_{t-2},\ldots, I_{t-m-1}\}$ for an $m \geq 1$ and $\cZ_t \coloneqq \{Y_{t-1}\}$ and let $b \coloneqq \E[\beta_t]$ be the expectation of $\beta_t$. 
        If $\beta_t \indep (X_{t-1}, \cI_t)$, the following two statements hold.
        (i)
        There exists $\alpha \in \R$ such that $b$ satisfies the NIV moment condition $\E[\cov(Y_t - b X_{t-1}-\alpha \cZ_t, \cI_t)] = 0$.
        (ii)
        Further, if $\E[\vecin{X_{t-1}^\top,\cZ_t^\top}^\top\cI_t^\top]$ has rank $d_X + d_Y$, $b$ is identified by $\niv{X_{t-1}}{Y_t}(\cI_t, \cZ_t)$.
    \end{restatable}
    The proofs can be found in Appendices~\ref{sec:newproooof1} and~\ref{sec:newproooof2}, respectively.
    The above considers identification of a single time point. Further assumptions (such as observing each time point several times) may be required to obtain consistent estimators.
}

\section{Simulation details and additional experiments}
\subsection{Correlation matrix for errors in \texorpdfstring{\cref{plot:compare}}{}}\label{sec:error-correlation-matrix}
In \cref{plot:compare}, we plot the absolute errors for two CIV estimators and two NIV estimators for several different data generating mechanisms. Here, we print the Spearman correlations between the four %
estimators. 

When $n = 100$, the correlation is
\begin{equation*}
\begin{matrix}
    \textrm{\tiny{CIV}}_{\text{\tiny I,X,Y}} \\
    \textrm{\tiny{CIV}}_{\text{\tiny I}} \\
    \textrm{\tiny{NIV}}_{\textrm{\tiny 1 lag}} \\
    \textrm{\tiny{NIV}}_{\textrm{\tiny 3 lags}}
\end{matrix}
\begin{pmatrix}
   1.000 & 0.635 & 0.391 & 0.390 \\ 
  0.635 & 1.000 & 0.276 & 0.241 \\ 
  0.391 & 0.276 & 1.000 & 0.393 \\ 
  0.390 & 0.241 & 0.393 & 1.000 \\ 
\end{pmatrix}
\end{equation*}

When $n = 1{,}000$, the correlation is
\begin{equation*}
\begin{matrix}
    \textrm{\tiny{CIV}}_{\text{\tiny I,X,Y}} \\
    \textrm{\tiny{CIV}}_{\text{\tiny I}} \\
    \textrm{\tiny{NIV}}_{\textrm{\tiny 1 lag}} \\
    \textrm{\tiny{NIV}}_{\textrm{\tiny 3 lags}}
\end{matrix}
\begin{pmatrix}
  1.000 & 0.728 & 0.400 & 0.476 \\ 
  0.728 & 1.000 & 0.275 & 0.308 \\ 
  0.400 & 0.275 & 1.000 & 0.417 \\ 
  0.476 & 0.308 & 0.417 & 1.000 \\ 
\end{pmatrix}
\end{equation*}

When $n = 10{,}000$, the correlation is
\begin{equation*}
\begin{matrix}
    \textrm{\tiny{CIV}}_{\text{\tiny I,X,Y}} \\
    \textrm{\tiny{CIV}}_{\text{\tiny I}} \\
    \textrm{\tiny{NIV}}_{\textrm{\tiny 1 lag}} \\
    \textrm{\tiny{NIV}}_{\textrm{\tiny 3 lags}}
\end{matrix}
\begin{pmatrix}
  1.000 & 0.742 & 0.385 & 0.465 \\ 
  0.742 & 1.000 & 0.280 & 0.308 \\ 
  0.385 & 0.280 & 1.000 & 0.459 \\ 
  0.465 & 0.308 & 0.459 & 1.000 \\ 
\end{pmatrix}
\end{equation*}
Across sample sizes, all %
estimators are positively correlated, with the strongest correlation present between the two CIV estimators. Yet, for no other pair does the correlation exceed $0.5$, indicating that the remaining estimator pairs do not agree on which data generating mechanisms are 
more `difficult'.

\subsection{Varying instrument strength for a fixed data generating mechanism}\label{sec:exp-fixed-generating}
In \cref{sec:Sim_compare}, we simulate random data generating mechanisms, and consider the error in estimating causal effects. The resulting errors plotted in \cref{plot:compare} then reflect both finite sample fluctuations and fluctuations stemming from the fact that estimation is more difficult in some data generating mechanisms than in others. 

In this experiment, we consider a fixed data generating mechanism to compare only the finite sample error of the NIV and CIV estimators. We use the simulation setup described in \cref{sec:simulations}, except that the entries of the parameter matrix are not random, but fixed at
\begin{equation*}
    A \coloneqq
,
\end{equation*}
where $s$ is a parameter that controls the strength of the instrument to be either `Strong' ($s = 1$), `Medium' ($s=0.5$), `Weak' ($s = 0.1)$ or `Very Weak' ($s = 0.01$).

\begin{figure}[t]
    \centering
    \scalebox{0.8}{
%
     }
    \caption{{Absolute errors of CIV and NIV estimators of the causal effect in the experiment described in \cref{sec:exp-fixed-generating}. Each point corresponds to a single instantiation of a time series. The data are simulated using four different data generating mechanisms, with the strength of the instrument ranging from `Strong' to `Very Weak'.}}
    \label{fig:fixed-generating-mechanism}
\end{figure}
For each instrument strength, we sample a time series of $T = 1{,}000$ steps, apply the $\textrm{CIV}_{I,X, Y}$ and $\textrm{NIV}_{3 \textrm{ lags}}$ estimators from \cref{plot:compare} (left), and store the absolute deviation from the ground truth causal effect $\beta$. We repeat this $300$ times and plot the resulting errors in \cref{fig:fixed-generating-mechanism}.

As expected, the average errors for both CIV and NIV increase as the instruments get weaker. However, for this data generating mechanism, the weakening of the instruments particularly hurts the performance of the CIV estimator: While the performance of NIV and CIV are comparable using the `strong' instruments, the NIV estimator outperforms the CIV estimator in the majority of simulations with the `very weak' instruments. 

\subsection{Asymptotic behaviour with weak instruments in the asymptotic regime}\label{sec:weak-instruments-experiment}
The instrument process may sometimes only explain a small fraction of the variance in the covariate process, a situation known as `weak instruments' \citep{stock2002survey}. 
When coefficients are kept constant, as is the case for the asymptotics discussed in \cref{sec:asymptotic-variance-ts}, both i.i.d.\ IV estimators as well as the time series estimators, which this paper focuses on, are asymptotically normal; however, if we consider an asymptotic regime where the effect of the instruments on the covariates vanishes at rate $O(T^{-\frac{1}{2}})$, \citet{staiger1997instrumental} show that in i.i.d.\ data, estimators are not necessarily asymptotically normal, and may be ill-behaved even for large sample sizes. 
In this simulation study, we compare the bias and variance of NIV and CIV estimators to a misspecified IV estimator. We sample data from the same data generating process as in \cref{sec:simulations} (see `Data generating process'), except that we scale the entries $\alpha_{X,I}$ by $\frac{1}{2\sqrt{T}}$. 
We use the same estimators as the ones described in \cref{plot:compare}, and in addition an ordinary (misspecified) IV estimator, $\iv{X_{t-1}}{Y_t}(I_{t-2})$. 
In \cref{plot:weak-iv}, we plot the mean average error. As one would expect for GMM-based IV-methods \citep{stock2002survey}, we do not see evidence that the error converges to $0$, but rather remains approximately constant. 
Furthermore, the experiment does not suggest that our estimators are more affected by weak instruments than the ordinary IV estimator. While the $\NIV_{1\text{ lag}}$ estimator exhibits larger average errors, this aligns with its behavior under non-weak instrument conditions, as shown in \cref{plot:compare}, where we noted that it has heavy-tailed errors due to just-identification.

\begin{figure}[t]
    \centering
    {\scriptsize
\begin{tikzpicture}[x=1pt,y=1pt]
\definecolor{fillColor}{RGB}{255,255,255}
\path[use as bounding box,fill=fillColor,fill opacity=0.00] (0,0) rectangle (433.62,126.47);
\begin{scope}
\path[clip] ( 44.91, 30.69) rectangle (350.30,120.97);
\definecolor{drawColor}{gray}{0.92}

\path[draw=drawColor,line width= 0.3pt,line join=round] ( 44.91, 33.60) --
	(350.30, 33.60);

\path[draw=drawColor,line width= 0.3pt,line join=round] ( 44.91, 50.41) --
	(350.30, 50.41);

\path[draw=drawColor,line width= 0.3pt,line join=round] ( 44.91, 67.21) --
	(350.30, 67.21);

\path[draw=drawColor,line width= 0.3pt,line join=round] ( 44.91, 84.02) --
	(350.30, 84.02);

\path[draw=drawColor,line width= 0.3pt,line join=round] ( 44.91,100.83) --
	(350.30,100.83);

\path[draw=drawColor,line width= 0.3pt,line join=round] ( 44.91,117.64) --
	(350.30,117.64);

\path[draw=drawColor,line width= 0.6pt,line join=round] ( 44.91, 42.00) --
	(350.30, 42.00);

\path[draw=drawColor,line width= 0.6pt,line join=round] ( 44.91, 58.81) --
	(350.30, 58.81);

\path[draw=drawColor,line width= 0.6pt,line join=round] ( 44.91, 75.62) --
	(350.30, 75.62);

\path[draw=drawColor,line width= 0.6pt,line join=round] ( 44.91, 92.43) --
	(350.30, 92.43);

\path[draw=drawColor,line width= 0.6pt,line join=round] ( 44.91,109.23) --
	(350.30,109.23);

\path[draw=drawColor,line width= 0.6pt,line join=round] ( 80.15, 30.69) --
	( 80.15,120.97);

\path[draw=drawColor,line width= 0.6pt,line join=round] (138.87, 30.69) --
	(138.87,120.97);

\path[draw=drawColor,line width= 0.6pt,line join=round] (197.60, 30.69) --
	(197.60,120.97);

\path[draw=drawColor,line width= 0.6pt,line join=round] (256.33, 30.69) --
	(256.33,120.97);

\path[draw=drawColor,line width= 0.6pt,line join=round] (315.06, 30.69) --
	(315.06,120.97);
\definecolor{drawColor}{RGB}{248,118,109}

\path[draw=drawColor,line width= 0.6pt,line join=round] ( 80.15, 64.77) --
	(138.87, 62.24) --
	(197.60, 60.32) --
	(256.33, 60.69) --
	(315.06, 60.11);
\definecolor{drawColor}{RGB}{163,165,0}

\path[draw=drawColor,line width= 0.6pt,line join=round] ( 80.15, 69.28) --
	(138.87, 64.74) --
	(197.60, 62.77) --
	(256.33, 62.99) --
	(315.06, 69.14);
\definecolor{drawColor}{RGB}{0,191,125}

\path[draw=drawColor,line width= 0.6pt,line join=round] ( 80.15, 41.29) --
	(138.87, 39.02) --
	(197.60, 37.14) --
	(256.33, 35.75) --
	(315.06, 34.79);
\definecolor{drawColor}{RGB}{0,176,246}

\path[draw=drawColor,line width= 0.6pt,line join=round] ( 80.15,108.06) --
	(138.87,108.83) --
	(197.60,116.87) --
	(256.33,106.23) --
	(315.06,112.57);
\definecolor{drawColor}{RGB}{231,107,243}

\path[draw=drawColor,line width= 0.6pt,line join=round] ( 80.15, 66.67) --
	(138.87, 73.63) --
	(197.60, 64.48) --
	(256.33, 64.67) --
	(315.06, 63.84);
\end{scope}
\begin{scope}
\path[clip] (  0.00,  0.00) rectangle (433.62,126.47);
\definecolor{drawColor}{gray}{0.30}

\node[text=drawColor,anchor=base east,inner sep=0pt, outer sep=0pt, scale=  0.88] at ( 39.96, 38.97) {1};

\node[text=drawColor,anchor=base east,inner sep=0pt, outer sep=0pt, scale=  0.88] at ( 39.96, 55.78) {10};

\node[text=drawColor,anchor=base east,inner sep=0pt, outer sep=0pt, scale=  0.88] at ( 39.96, 72.59) {100};

\node[text=drawColor,anchor=base east,inner sep=0pt, outer sep=0pt, scale=  0.88] at ( 39.96, 89.39) {1000};

\node[text=drawColor,anchor=base east,inner sep=0pt, outer sep=0pt, scale=  0.88] at ( 39.96,106.20) {10000};
\end{scope}
\begin{scope}
\path[clip] (  0.00,  0.00) rectangle (433.62,126.47);
\definecolor{drawColor}{gray}{0.30}

\node[text=drawColor,anchor=base,inner sep=0pt, outer sep=0pt, scale=  0.88] at ( 80.15, 19.68) {100};

\node[text=drawColor,anchor=base,inner sep=0pt, outer sep=0pt, scale=  0.88] at (138.87, 19.68) {316};

\node[text=drawColor,anchor=base,inner sep=0pt, outer sep=0pt, scale=  0.88] at (197.60, 19.68) {1000};

\node[text=drawColor,anchor=base,inner sep=0pt, outer sep=0pt, scale=  0.88] at (256.33, 19.68) {3162};

\node[text=drawColor,anchor=base,inner sep=0pt, outer sep=0pt, scale=  0.88] at (315.06, 19.68) {10000};
\end{scope}
\begin{scope}
\path[clip] (  0.00,  0.00) rectangle (433.62,126.47);
\definecolor{drawColor}{RGB}{0,0,0}

\node[text=drawColor,anchor=base,inner sep=0pt, outer sep=0pt, scale=  1.10] at (197.60,  7.64) {Sample size};
\end{scope}
\begin{scope}
\path[clip] (  0.00,  0.00) rectangle (433.62,126.47);
\definecolor{drawColor}{RGB}{0,0,0}

\node[text=drawColor,rotate= 90.00,anchor=base,inner sep=0pt, outer sep=0pt, scale=  1.10] at ( 13.08, 75.83) {$\operatorname{error}(\hat\beta)$};
\end{scope}
\begin{scope}
\path[clip] (  0.00,  0.00) rectangle (433.62,126.47);
\definecolor{drawColor}{RGB}{0,0,0}

\node[text=drawColor,anchor=base west,inner sep=0pt, outer sep=0pt, scale=  1.10] at (366.80,110.93) {IV method};
\end{scope}
\begin{scope}
\path[clip] (  0.00,  0.00) rectangle (433.62,126.47);
\definecolor{drawColor}{RGB}{248,118,109}

\path[draw=drawColor,line width= 0.6pt,line join=round] (368.24, 97.13) -- (379.80, 97.13);
\end{scope}
\begin{scope}
\path[clip] (  0.00,  0.00) rectangle (433.62,126.47);
\definecolor{drawColor}{RGB}{163,165,0}

\path[draw=drawColor,line width= 0.6pt,line join=round] (368.24, 82.68) -- (379.80, 82.68);
\end{scope}
\begin{scope}
\path[clip] (  0.00,  0.00) rectangle (433.62,126.47);
\definecolor{drawColor}{RGB}{0,191,125}

\path[draw=drawColor,line width= 0.6pt,line join=round] (368.24, 68.22) -- (379.80, 68.22);
\end{scope}
\begin{scope}
\path[clip] (  0.00,  0.00) rectangle (433.62,126.47);
\definecolor{drawColor}{RGB}{0,176,246}

\path[draw=drawColor,line width= 0.6pt,line join=round] (368.24, 53.77) -- (379.80, 53.77);
\end{scope}
\begin{scope}
\path[clip] (  0.00,  0.00) rectangle (433.62,126.47);
\definecolor{drawColor}{RGB}{231,107,243}

\path[draw=drawColor,line width= 0.6pt,line join=round] (368.24, 39.31) -- (379.80, 39.31);
\end{scope}
\begin{scope}
\path[clip] (  0.00,  0.00) rectangle (433.62,126.47);
\definecolor{drawColor}{RGB}{0,0,0}

\node[text=drawColor,anchor=base west,inner sep=0pt, outer sep=0pt, scale=  0.88] at (386.75, 94.10) {CIV$_{I,X,Y}$};
\end{scope}
\begin{scope}
\path[clip] (  0.00,  0.00) rectangle (433.62,126.47);
\definecolor{drawColor}{RGB}{0,0,0}

\node[text=drawColor,anchor=base west,inner sep=0pt, outer sep=0pt, scale=  0.88] at (386.75, 79.65) {CIV$_{I}$};
\end{scope}
\begin{scope}
\path[clip] (  0.00,  0.00) rectangle (433.62,126.47);
\definecolor{drawColor}{RGB}{0,0,0}

\node[text=drawColor,anchor=base west,inner sep=0pt, outer sep=0pt, scale=  0.88] at (386.75, 65.19) {NIV$_{3 \textrm{ lags}}$};
\end{scope}
\begin{scope}
\path[clip] (  0.00,  0.00) rectangle (433.62,126.47);
\definecolor{drawColor}{RGB}{0,0,0}

\node[text=drawColor,anchor=base west,inner sep=0pt, outer sep=0pt, scale=  0.88] at (386.75, 50.74) {NIV$_{1 \textrm{ lag}}$};
\end{scope}
\begin{scope}
\path[clip] (  0.00,  0.00) rectangle (433.62,126.47);
\definecolor{drawColor}{RGB}{0,0,0}

\node[text=drawColor,anchor=base west,inner sep=0pt, outer sep=0pt, scale=  0.88] at (386.75, 36.28) {IV};
\end{scope}
\end{tikzpicture}
     }
    \caption{Average error (in log scale) of NIV, CIV and IV estimators for the causal effect of $X_{t-1}$ on $Y_t$,
    where the strength of the instrument vanishes at rate $O(T^{-1/2})$.
    Similar to the weak instrument setting in i.i.d.\ data, in this asymptotic regime our estimators do not converge to the true causal effect \citep[e.g.,][]{stock2002survey}. 
    }
    \label{plot:weak-iv}
\end{figure}

\section{Proofs}\label{sec:all_proofs}
\subsection{Proof of Conditional IV}  \label{app:proofciviid}
For completeness, we now prove our statement about CIV from \cref{sec:intro-civ}.
{Similar statements have been reported by \citet{brito2002generalized, pearl2009causality,henckel2021graphical} 
but they differ in their 
precise formulations. }
\begin{proposition} \label{prop:civiid}
Consider a linear SCM (see \cref{sec:interventions}) over variables $V$, and let $\cI, \cX, \cB, \{Y\} \subseteq V$ be disjoint collections of variables
from $V$, and let $\G$ be the corresponding DAG.
Assume that $\cI, \cX$ and $Y$ have zero mean 
and finite second moment and 
let $\beta$ be the causal coefficient with which $\cX$ enters the structural equation for $Y$
(some of the entries of $\beta$ can be zero, so not all variables in $\cX$ have to be parents of $Y$).
We consider
the following three requirements on $\cI, \cX, \cB$ and $Y$
\begin{civenum}
    \item $\cI$ and $Y$ are $d$-separated given $\cB$ in the graph $\G_{\cX\not\rightarrow Y}$, 
    that is the graph $\G$ where all direct edges from $\cX$ to $Y$ are removed,
    \item $\cB$ is not a descendant of 
    $\cX \cup Y$ in $\G$, and
    \item the matrix $\E[\cov(\cX, \cI | \cB)]$ has rank $d_\cX$, that is, full row rank.
\end{civenum}
If \cref{assump:civ-d-sep,assump:civ-descendants} are met, $Y - \beta\cX \indep \cI | \cB$, and in particular $\beta$ satisfies the \emph{CIV moment equation}
\begin{equation}\label{eq:civ-moment-equation-appendix}
    \E[\cov(Y - \beta \cX, \cI|\cB)] = 0.    
\end{equation}
If, additionally, \cref{assump:civ-relevance} is met,
$\beta$ is the unique solution to this equation,
\begin{equation*}
    \E[\cov(Y - b \cX, \cI|\cB)] = 0 \implies b = \beta.
\end{equation*}
\end{proposition}

\begin{proof}
Due to the additive, linear  structure of the SCM, we can write $Y$ as a     \begin{align*}
        Y = \beta \cX + \pi \cB + \gamma R + \epsilon^Y,
    \end{align*}
    where $R$ are those parents of $Y$ that are not in $\cX\cup\cB$ and $\pi \in \R^{d_Y \times d_\cB}, \gamma \in \R^{d_Y \times d_R}$
    are some coefficients.
    
    We claim (1) that any path from $\cI$ to $R$ is blocked by $\cB$ and (2) that $\E[\cov(\epsilon^Y, \cI|\cB)]=0$.
    It then follows from the global Markov property \citep{lauritzen1996graphical}, that $\E[\cov(R, \cI |\cB)] = 0$, and trivially also $\E[\cov(\cB, \cI|\cB)]=0$. 
    Hence, since $Y - \beta \cX = \pi \cB + \gamma R + \epsilon^Y$, it follows that $\E[\cov(Y - \beta \cX, \cI | \cB)] = \E[\cov(\pi \cB + \gamma R + \epsilon^Y, \cI | \cB)] = 0$.
    
    For (1), suppose for a contradiction that a path $\pi$ between $\cI$ and $R$ that is unblocked given $\cB$ exists. 
    Case 1: $\pi$ does not contain any edge from $\cX$ to $Y$. Then, 
    the path that concatenates $\pi$ with the corresponding edge from $R$ to $Y$ is an unblocked path (given $\cB$) in the graph, where the edges from $\cX$ to $Y$ are removed. This contradicts \cref{assump:civ-d-sep}.
    Case 2: $\pi$ contains an edge from $X \in \cX$ to $Y$. Then, 
    $\pi$ contains either the structure
    $X \rightarrow Y \leftarrow$ 
    or the structure 
    $X \rightarrow Y \rightarrow$. 
    The first case implies $\cB \cap \DE{Y} \neq \emptyset$, violating~\cref{assump:civ-descendants}. In the second case, 
    we either have that there is a directed path from $Y$ to $\cI$ that is unblocked by $\cB$, violating  \cref{assump:civ-d-sep}
    or, again that 
    $\cB \cap \DE{Y} \neq \emptyset$, violating~\cref{assump:civ-descendants}.
    
    For (2), we have that neither $\cB$ nor $\cI$ are descendants of $Y$: $\cB$ cannot be a descendant of $Y$ due to \cref{assump:civ-descendants}, and by \cref{assump:civ-d-sep}, $\cI$ can only be a descendant if $\cB$ is also a descendant, which is not possible. Every variable in the linear SCM can be rewritten as a function only of noise terms corresponding to ancestors. 
    Applying this to $\cB$ and $I$, we have  (since $Y$ is not an ancestor of neither $\cI$ nor $\cB$ and because $\epsilon^Y$ is independent of all other noise terms in the SCM) that $\epsilon^Y$ is independent of $(\cB, \cI)$ and it follows that $\E[\cov(\epsilon^Y, \cI|\cB)]=0$.
    
    When we additionally assume \cref{assump:civ-relevance}, the solution to \cref{eq:civ-moment-equation-appendix} is unique because the equation can be rewritten to
    \begin{align*}
        b\E[\cov(\cX, \cI|\cB)] = \E[\cov(Y, \cI|\cB)],
    \end{align*}
    and by the assumption of full row rank of $\E[\cov(\cX, \cI|\cB)]$, this can have at most one solution.
\end{proof}

\subsection{Proof of \texorpdfstring{\cref{thm:gmp}}{}}
\GMP*
\begin{proof}
    Our proof is inspired by \citet[][Sec.\ 6]{lauritzen1990independence}.
Given a time series (sub)graph 
    $\G$ over nodes $V$, we denote by $V_{[s,t]} := \{S^i_v | i \in \{1, \ldots, d\}, s\leq v \leq t, S^i_v \in V\} \subseteq V$ the nodes in $V$ between times $s$ and $t$ and denote by $\G_{[s,t]}$ the subgraph of $\G$, where only vertices in $V_{[s,t]}$ are included.
    Let further $s_0, t_0$ be the largest and smallest time points respectively such that $A\cup C\cup B \subseteq V_{[s_0, t_0]}$. 
    Let $q \in \mathbb{N}_{>p}$ be such that if two nodes in $V_{[s_0, t_0]}$ are $d$-connected (with empty conditioning set)
    in $\Gfull$, then there is a $d$-connecting path in
    $\G_{[s_0-q,t_0]}$.
    Define the set 
    \begin{align*}
            \mathcal{A} := \{n \in &\mathbb{N}\,|\, \text{for all } 
    \G^* \text{ that are graphs over nodes $V^*$ with time indices between } s_0-q \text{ and }  t_0 \text{ s.t.\ } \\  
    & \qquad \text{all edges in }\G^*\text{ point `forward in time', i.e., } \forall k\in \N,\text{ there is no edge }S^i_t \rightarrow S^j_{t-k}\text{ and} \\
    &\qquad\G^*\text{ is a subgr.\ of } \G_{f}, \text{ where }  \G_{f} \text{ is the full-time gr.\ induced by } \G^0 := \G^*_{[s_0-q,s_0-1]} \text{ and}\\
    &\qquad AN_{\G_f}(V^+)_{[s_0,t_0]} = V^+, \text{ where }   
    V^+ := V^{*}_{[s_0, t_0]}
    \text{ and}\\
    & \qquad |V^+| = n \text{ and} \\
    & \text{for all VAR{($p$)} processes whose structure is specified by } \G^0 \text{ and}\\
    & \text{for all } A^*, B, C^* \subseteq \G^* \text{ such that } \\
    & \qquad (A^+ \cup B \cup C^+) = V^+ (\text{where } A^+ := A^*_{[s_0,t_0]}, C^+ := C^*_{[s_0,t_0]}, \G^+ := \G^*_{[s_0,t_0]}),  \text{ and} \\
    & \qquad A^* = A^+ \cup (\PA_{\G^*}(A^+) \cap V^0) \text{ and } C^* = C^+ \cup (\PA_{\G^*}(C^+) \cap V^0)\\
    & \qquad{\text{where } V^0 \text{ are the nodes of }\G^0}\\
    & \text{we have }\\
    &\qquad A^* \perp_{\G^*} C^* | B \; \Rightarrow \;  A^* \indep C^* | B 
    \},
        \end{align*}
        where $\perp_{\G}$ indicates $d$-separation in $\G$.

    We show below, by induction, that $\mathcal{A} = \mathbb{N}$. This suffices to prove the statement of the theorem because of the following line of arguments.
    Let  
    $V^0 := V_{[s_0-q, s_0-1]}$ and 
    $V^+ := AN_{\Gfull}(A \cup C \cup B)_{[s_0, t_0]}$.
    Let $\G^*$ be the graph $\Gfull$ restricted to the nodes in $V^* := V^0 \cup V^+$. 
    Then, $A \perp_{\G^*} C | B$ (as $\G^*$ is a subgraph of $\Gfull$).   
    If $V^+ \neq (A \cup C \cup B)$, we enlarge $A$ and $C$ to the disjoint sets $A^+$ and $C^+$ such that 
    $A^+ \perp_{\G^*} C^+ | B$ and 
    $V^+ = A^+ \cup B \cup C^+$.
    Let us define 
    $A^* := A^+ \cup (\PA_{\G^*}(A^+) \cap V^0)$
    and 
    $C^* := C^+ \cup (\PA_{\G^*}(C^+) \cap V^0)$. Importantly, these two sets are disjoint (otherwise, $A^+$ and $C^+$ would have a joint parent not in $B$, violating $A^+ \perp_{\G^*} C^+ | B$).
    We then have that $A^* \perp_{\G^*} C^* | B$
    (Indeed, if there is an open path from a node in $a \in A^*$ to a node in $c \in C^*$, given $B$, then there is an open path between a node in $A^+$ (either $a$ itself or its child in $A^+$) to a node in $C^+$ (either $c$ itself or its child in $C^+$), violating $A^+ \perp_{\G^*} C^+ | B$). But then 
    $\mathcal{A} = \mathbb{N}$ implies $A^* \indep C^* | B$. And this implies that 
    $A \indep C | B$, as $A$ and $C$ are subsets of $A^*$ and $C^*$, respectively.

    Let us now prove that $\mathcal{A} = \mathbb{N}$ by, (1), proving 
    $1 \in \mathcal{A}$ and $2 \in \mathcal{A}$ and, (2), proving $n \in \mathcal{A}$ implies $n+1 \in \mathcal{A}$.

(1) We now prove that $1 \in \mathcal{A}$ and $2 \in \mathcal{A}$.\\
    The only non-trivial statement occurs when $A^* \neq \emptyset$ and $C^* \neq \emptyset$ and $B=\emptyset$. Because 
    $A^* \perp_{\G^*} C^*$, we have $AN_{\G^*}(A^*) \cap AN_{\G^*}(C^*) = \emptyset$. 
    This implies 
    $AN_{\Gfull}(A^*) \cap AN_{\Gfull}(C^*) = \emptyset$ because of the repetitive structure in a full time graph and the choice of $q$.

    (2) We now prove that $n \in \mathcal{A}$ implies $n+1 \in \mathcal{A}$.\\
    Assume $n\in \mathcal{A}$ and consider $\G^*$, $\G^0$, $\G^+$, $V^+$, $A^*$, $A^+$, $B$, $C^*$, $C^+$
    as described in set $\mathcal{A}$ with $|V^+|=n+1$. 
    Consider a node $\lambda \in V^+$ that is a sink node in $\G^*$.

First, assume that $\lambda \in A^+$. Then $\PA_{\G^*}(\lambda) \subseteq (A^*\setminus \{\lambda\}) \cup B$ (because $d$-separation would be violated if $C^* \cap \PA(\lambda) \neq \emptyset$). Thus, it follows that 
\begin{align} \label{eq:A1}
    \lambda \indep C^* | B \cup (A^* \setminus \{\lambda\}).
\end{align}
(Indeed, 
$\PA_{\Gfull}(\lambda) = \PA_{\G^*}(\lambda)$ and thus 
there exist a coefficient vector $\gamma \in \R^{|\PA_{\G^*}(\lambda)|}$ such that $\lambda = \gamma^\top \PA_{\G^*}(\lambda) + \epsilon^\lambda$; it then follows from the $MA(\infty)$
representation of $S$,
see \citet[][Sec.\ 11.3]{Brockwell1991} or \citet{hamilton1994time},
that $\lambda \indep C^* \cup (B \cup (A^* \setminus \{\lambda\}) \setminus \PA_{\G^*}(\lambda)) \, |\, \PA_{\G^*}(\lambda)$;
the claimed independence then follows with the weak union property.)
If $A^* = \{\lambda\}$, then this already implies 
$A^* \indep C^* | B$. Otherwise, we observe that
$A^*\setminus \{\lambda\} \perp_{{\G^*}^m} C^* |B$, 
where $\G^m$ denotes moralization of graph $\G$ \citep{lauritzen1996graphical}, as $d$-separation is equivalent to separation in the moralized graph.
But then, 
$A^*\setminus \{\lambda\} \perp_{({\G^*}^m)_{V^* \setminus \{\lambda\}}} C^* |B$, 
as this graph contains no more edges. 
And therefore, 
$(A^*\setminus \{\lambda\}) \perp_{(\G^*_{V^* \setminus \{\lambda\}})^m} C^* |B$
as, again, the graph contains no more edges. 
By the induction hypothesis $n \in \mathcal{A}$ and thus 
\begin{equation} \label{eq:A2}
A^* \setminus \{\lambda\} \indep C^* | B.
    \end{equation}
Combining~\eqref{eq:A1} and~\eqref{eq:A2} by the contraction property, it follows that 
$$
A^* \indep C^* | B.
$$

Second, assume that $\lambda \in C^+$. The argument follows in the same way as in the case $\lambda \in A^+$.

Third, assume that $\lambda\in B$ (these are all cases since $A^+ \cup B \cup C^+ = V^+$). Since $B$ separates $A^*$ and $C^*$ in $(\G^*)^m$, then $B\setminus \{\lambda\}$ also separates $A^*$ and $C^*$ in 
$({\G^*}^m)_{V^*\setminus \{\lambda\}}$ (as it has no more edges), 
and therefore $B\setminus \{\lambda\}$ also separates $A^*$ and $C^*$ in 
$(\G^*_{V^*\setminus \{\lambda\}})^m$,
since this graph, again, has no more edges. 
By the induction hypothesis $n\in \mathcal{A}$, so 
\begin{equation} \label{eq:Bind1}
    A^* \indep C^* | (B\setminus\{ \lambda\}).
\end{equation}
We now prove a second independence statement. We now make a case distinction
(a) Assume that 
$$
\PA_{\G^+}(\lambda) \cap A^* \neq \emptyset \quad \text{ or } \quad 
AN_{\G^*}(\PA_{\G^*}(\lambda)_{[s_0-q,s_0{-1}]}) \cap 
AN_{\G^*}(\PA_{\G^*}(A^*)_{[s_0-q,s_0{-1}]}) \neq \emptyset.
$$
Then, it follows that
\begin{equation} \label{eq:nocommonanc}
\PA_{\G^+}(\lambda) \cap  C^* = \emptyset \quad \text{ and }
AN_{\G^*}(\PA_{\G^*}(\lambda)_{[s_0-q,s_0{-1}]}) \cap 
AN_{\G^*}(\PA_{\G^*}(C^*)_{[s_0-q,s_0{-1}]}) = \emptyset.
\end{equation}
\begin{quote}
(Indeed, if the statement on the left-hand side would be false, then there is a $d$-connecting path between $A^*$ and $C^*$, given $B$: this goes from the element in $\PA_{\G^+}(\lambda) \cap C^*$ to $\lambda$ (which is in $B$) and then either 
to the element in $\PA_{\G^+}(\lambda) \cap A^*$
or to the common ancestor of $\PA_{\G^*}(\lambda)_{[s_0-q,s_0-1]}$ and $\PA_{\G^*}(A^*)_{[s_0-q,s_0-1]}$ and then to the corresponding element in $A^*$. If the statement on the right-hand side would be false, then we can use the same path but this time
going via the common ancestor of $\PA_{\G^*}(\lambda)_{[s_0-q,s_0{-1}]}$ and $\PA_{\G^*}(C^*)_{[s_0-q,s_0{-1}]}$.)
\end{quote}
But then it follows that 
\begin{equation} \label{eq:Bind2}
\lambda \indep C^* |A^* \cup (B\setminus \{\lambda\}).
\end{equation} 
\begin{quote}
(Indeed, 
noting that 
$\PA_{\G^+}(\lambda) \subseteq A^* \cup B \setminus \{\lambda\}$, we can replace the left-hand side by the
MA($\infty$) representation of 
$\PA_{\G^*}(\lambda)_{[s_0-q,s_0-1]}$.
For $C^*$, we repeatedly use the structural equations except for variables in $B \setminus \{\lambda\}$ 
or variables in $\PA_{\G^*}(C^*)_{[s_0-q,s_0-1]}$
(other variables will not occur: If there was a variable in $A^*$, for example, there would be a directed path from $A^*$ to $C^*$).
We then use the MA($\infty$) representation of 
$\PA_{\G^*}(C^*)_{[s_0-q,s_0-1]}$.
The statement then follows from the fact that
$\PA_{\G^*}(\lambda)_{[s_0-q,s_0-1]}$
and
$\PA_{\G^*}(C^*)_{[s_0-q,s_0-1]}$
do not have common ancestors, see \cref{eq:nocommonanc}.
\end{quote}
Combining~\eqref{eq:Bind1} and~\eqref{eq:Bind2} using the contraction property, it follows that $C^* \indep (\{\lambda\} \cup A^*) |(B\setminus \{\lambda\})$, and by the weak union property that 
$$
C^* \indep A^* |B. 
$$
(b) 
Now assume that 
$$
\PA_{\G^+}(\lambda) \cap A^* = \emptyset \quad \text{ and } \quad 
AN_{\G^*}(\PA_{\G^*}(\lambda)_{[s_0-q,s_0-1]}) \cap 
AN_{\G^*}(\PA_{\G^*}(A^*)_{[s_0-q,s_0-1]}) = \emptyset.
$$
Similarly as in case (a) it follows that 
\begin{equation} \label{eq:Bind3}
\lambda \indep A^* |C^* \cup (B\setminus \{\lambda\}).
\end{equation} 
Combining~\eqref{eq:Bind1} and~\eqref{eq:Bind3} using the contraction property, it follows that $(\{\lambda\} \cup C^*) \indep A^* |(B\setminus \{\lambda\})$, and by the weak union property that 
$$
C^* \indep A^* |B. 
$$
This concludes the proof.
\end{proof}

\subsection{Proof of \texorpdfstring{\cref{prop:nuisance-iv}}{}}
\nuisanceIV*
\begin{proof}
    By satisfaction of \cref{assump:civ-d-sep,assump:civ-descendants,assump:civ-relevance}, the causal effect $\tilde{\beta} = [\beta, \alpha]$
    of $\tilde{\cX} = \cX\cup\cZ$ on $Y$ is identified by the instrument $\cI$ and the conditioning set $\cB$
    by \cref{prop:civiid} in \cref{app:proofciviid}.
    In particular, also the sub-vector of the IV estimate corresponding to $\cX$ is identified.
\end{proof}

\subsection{Proof of \texorpdfstring{\cref{prop:neither-civ-nor-niv-optimal}}{}}
\neitherOptimal*
\begin{proof}
    We show this by considering two SCMs over 6 variables $S = [H, I, X, Y, Z, B]$ given by $S \coloneqq AS + \epsilon$ where $A$ is such that the resulting graph is acyclic and admits the graphical model in \cref{fig:3-structures} (right) and $\epsilon\sim\mathcal{N}(0, \Gamma)$; we provide two concrete choices for $A$ and $\Gamma$ below. 
    We consider both the $\civ{X}{Y}(I|B)$ and the $\niv{X}{Y}([I, B], Z)$ estimates of $\beta$, the causal effect of $X$ on $Y$, and provide two sets of parameters $(A^I, \Gamma^I)$ and $(A^{II}, \Gamma^{II})$ such that if $X$ is generated according to $(A^I, \Gamma^I)$, the CIV estimator has a lower asymptotic variance than the NIV estimator, and if $S$ is generated according to $(A^{II}, \Gamma^{II})$, the CIV estimator has a higher asymptotic variance than the NIV estimator. 
    \begin{align*}
        A^I \coloneqq 
        \begin{matrix}
            \text{\tiny H} \\ \text{\tiny I} \\ \text{\tiny X} \\ \text{\tiny Y} \\ \text{\tiny Z} \\ \text{\tiny B}
        \end{matrix}
        \begin{pmatrix}
            0 & 0 & 0 & 0 & 0 & 0\\ 
            0 & 0 & 0 & 0 & 0 & 1.185\\ 
            21.095 & 6.885 & 0 & 0 & 0 & -5.969\\ 
            -7.244 & 0 & 16.499 & 0 & -1.892 & 0\\ 
            1.921 & 0 & 0 & 0 & 0 & 2.62\\ 
            0 & 0 & 0 & 0 & 0 & 0
        \end{pmatrix}\,
        \Gamma^I \coloneqq 
        \begin{matrix}
            \text{\tiny H} \\ \text{\tiny I} \\ \text{\tiny X} \\ \text{\tiny Y} \\ \text{\tiny Z} \\ \text{\tiny B}
        \end{matrix}
        \begin{pmatrix}
            0.2 & 0 & 0 & 0 & 0 & 0\\ 
            0 & 1.2 & 0 & 0 & 0 & 0\\ 
            0 & 0 & 2.2 & 0 & 0 & 0\\ 
            0 & 0 & 0 & 1.2 & 0 & 0\\ 
            0 & 0 & 0 & 0 & 2.2 & 0\\ 
            0 & 0 & 0 & 0 & 0 & 0.2
        \end{pmatrix},
    \end{align*}
and
    \begin{align*}
        A^{II} \coloneqq 
        \begin{matrix}
            \text{\tiny H} \\ \text{\tiny I} \\ \text{\tiny X} \\ \text{\tiny Y} \\ \text{\tiny Z} \\ \text{\tiny B}
        \end{matrix}
        \begin{pmatrix}
            0 & 0 & 0 & 0 & 0 & 0\\ 
            0 & 0 & 0 & 0 & 0 & -2.918\\ 
            -22.439 & 3.519 & 0 & 0 & 0 & 4.282\\ 
            19.964 & 0 & 4.737 & 0 & 4.011 & 0\\ 
            0.884 & 0 & 0 & 0 & 0 & -7.97\\ 
            0 & 0 & 0 & 0 & 0 & 0
        \end{pmatrix}\,
        \Gamma^{II} \coloneqq 
        \begin{matrix}
            \text{\tiny H} \\ \text{\tiny I} \\ \text{\tiny X} \\ \text{\tiny Y} \\ \text{\tiny Z} \\ \text{\tiny B}
        \end{matrix}
        \begin{pmatrix}
            3.2 & 0 & 0 & 0 & 0 & 0\\ 
            0 & 1.2 & 0 & 0 & 0 & 0\\ 
            0 & 0 & 3.2 & 0 & 0 & 0\\ 
            0 & 0 & 0 & 2.2 & 0 & 0\\ 
            0 & 0 & 0 & 0 & 1.2 & 0\\ 
            0 & 0 & 0 & 0 & 0 & 2.2
        \end{pmatrix}.
    \end{align*}
    We can now use the formulas from \cref{sec:asymptotic-variance} to get the asymptotic variances, using that $\E[S] = 0$ and $\E[SS^\top] = (1-A)^{-1}\Gamma (1-A)^{-\top}$. 
    When the data generating mechanism is $(A^I, \Gamma^I)$, the asymptotic distributions are specified by $\sqrt{T}(\hat\beta_{\CIV}-\beta)\sim\mathcal{N}(0, 524.4)$ and $\sqrt{T}(\hat\beta_{\NIV}-\beta)\sim\mathcal{N}(0, 522.7)$. Furthermore,
    when the data generating mechanism is $(A^{II}, \Gamma^{II})$,
    the asymptotic distributions are $\sqrt{T}(\hat\beta_{\CIV}-\beta)\sim\mathcal{N}(0, 320.0)$ and $\sqrt{T}(\hat\beta_{\NIV}-\beta)\sim\mathcal{N}(0, 575.4)$, respectively.
\end{proof}

\subsection{Proof of \texorpdfstring{\cref{prop:dream-theorem}}{}}
\dreamTheorem*
\begin{proof}
    We first show part (i). Due to the additive, linear structure in \cref{assump:varp}, we can rewrite $Y$ as a linear combination of the parents of $Y$ in $\Gfull$ plus some additive noise, $Y = a_1p_1 + \ldots + a_m p_m + \epsilon^Y$, where $a_1, \ldots, a_m$ are coefficients and $p_1, \ldots, p_m$ are nodes from $\Gfull$. 
    Similarly, we can recursively decompose the parents into their parents (in $\Gfull$) and noise without replacing variables in $\cX \cup \cZ \cup \ND{\cX \cup \cZ}_{\Gfull} \cup \cB$, until the first time we have a decomposition 
    \begin{align*}
        Y = \beta \cX + \alpha \cZ + \pi \cB + \gamma R + \epsilon,
    \end{align*}
    where $R$ are $d_R$ variables from $\Gfull$ that are not in $\cB$ or on any
    directed path from $\cX \cup \cZ$ to $Y$
    in $\Gfull$, 
    \begin{equation} \label{eq:Rpropert}
            R \cap (\cB \cup \cX \cup \cZ \cup \DE{\cX \cup \cZ}_{\Gfull}) = \emptyset,
        \end{equation}
    ($R$ may include descendants of $\cB$) and $\epsilon$ are all the
    weighted noise variables
    accumulated when doing the decomposition in parents, $\beta$ and $\alpha$ are the total causal effects of $\cX$ and $\cZ$ on $Y$, respectively, for some coefficients $\pi \in \R^{1 \times d_\cB}, \gamma \in \R^{1 \times d_R}$.
    The coefficients in front of $\cX$ and $\cZ$ are indeed $\beta$ and $\alpha$, because the total causal effect is the product along all paths (in $\Gfull$) from $\cX\cup\cZ$ to $Y$ (see \cref{sec:appendix-TCE}), and by the assumption that $\cB \subseteq \ND{\cX\cup\cZ}_{\G_M}$ (which implies $\cB \subseteq \ND{\cX\cup\cZ}_{\Gfull}$), no directed path from $\cX\cup\cZ$ to $Y$ is blocked by $\cB$. 
    
    We claim (1)  that any path in $\Gfull$ 
    from $\cI$ to $R$ is blocked by $\cB$ 
    and (2) that $\E[\cov(\epsilon, \cI|\cB)]=0$.
    It then follows from \cref{thm:gmp}, that $\E[\cov(R, \cI |\cB)] = 0$, and trivially also $\E[\cov(\cB, \cI|\cB)]=0$. 
    Hence, since $Y - \beta \cX - \alpha \cZ = \pi \cB + \gamma R + \epsilon$, it follows that $\E[\cov(Y - \beta \cX - \alpha \cZ, \cI | \cB)] = \E[\cov(\pi \cB + \gamma R + \epsilon, \cI | \cB)] = 0$.
    
    For (1), suppose for a contradiction that there exist $i \in \cI$ and $m \in R$ and a path $p: i - v_1 - \cdots - v_n - m$
    in $\Gfull$ that is unblocked given $\cB$, where an indirected edge indicates a directed edge that can have any orientation. Since $p$ is unblocked given $\cB$,
    the non-colliders of $p$ are disjoint from $\cB$ and for every collider $v_k$ on $p$, there is a node $b^k \in \cB$ such that $b^k\in \DE{v_k}$ in $\Gfull$. 
    (Observe that the end node $m$ is not in $M$: It is not in $\{Y\}$ (since it was found among the ancestors of $Y$), and by construction of $R$, it is not in $\cB\cup\cX\cup\cZ$. 
    Also $m\notin\cI$: By the construction of $R$ through recursive rewriting as parents, there exists a directed path from $m$ to $Y$ that does not intersect $\cX\cup\cZ\cup\cB$; if $m\in\cI$, this path would violate \cref{assump:civ-d-sep-tce}.) 
    
    Let $w_1, \ldots, w_L$ be those vertices among $v_1, \ldots, v_n$ that appear in the marginalized graph $\G_M$ and let $w_0 \coloneqq i$ and $w_{L+1} \coloneqq Y$. We now, a) construct a path including the nodes $w_0, \ldots, w_{L+1}$ (and possibly some additional colliders, see below) from $\cI$ to $Y$ in $\G_M$; we then show that, b), this path is still unblocked, given $\cB$, if we remove 
    all edges as specified by \cref{assump:civ-d-sep-tce}, creating a contradiction to \cref{assump:civ-d-sep-tce}.

    For a), consider those $0 \leq k \leq L$ for which $w_{k}$ and $w_{k+1}$ are not directly connected in $p$. For such $k$, consider
the segment of $p$ (as a path in $\Gfull$) from $w_k - \cdots - w_{k+1}$, where $\cdots$ represent edges from $p$ that are not in $M$. If there are no colliders on this segment, by \cref{def:marginalzed_graph} at least one of the edges $w_k \rightarrow w_{k+1}, w_k \leftarrow w_{k+1}$ or $w_k \leftrightarrow w_{k+1}$ are present in $\G_M$. 
    If there is exactly one collider on $w_k - \cdots - w_{k+1}$ (as part of $p$ in $\Gfull$), the segment must be one of the following four options:
    \begin{align*}
        w_k \rightarrow\cdots \rightarrow &v_i \leftarrow\cdots\leftarrow w_{k+1},\\
        w_k \leftarrow \cdots \leftarrow v_{j_1} \rightarrow \cdots \rightarrow &v_i \leftarrow \cdots \leftarrow w_{k+1},\\
        w_k \rightarrow\cdots\rightarrow &v_i \leftarrow \cdots \leftarrow v_{j_2} \rightarrow \cdots \rightarrow w_{k+1},\\
        w_k \leftarrow \cdots \leftarrow v_{j_1} \rightarrow \cdots \rightarrow &v_i \leftarrow \cdots \leftarrow v_{j_2} \rightarrow \cdots \rightarrow w_{k+1},
    \end{align*}
    where $v_{j_1}, v_{j_2}$ also are nodes on $p$ and $\cdots$ is a short-hand for edges with the same orientation. But since $\DE{v_i}_{\Gfull}\cap\cB \neq \emptyset$, this implies that at least one of the following paths are present in $\G_M$:
    \begin{align*}
        &w_k \rightarrow b_{k,1} \leftarrow w_{k+1} \\
        &w_k \leftrightarrow b_{k,1} \leftarrow w_{k+1}\\
        &w_k \rightarrow b_{k,1} \leftrightarrow w_{k+1} \\
        &w_k \leftrightarrow b_{k,1} \leftrightarrow w_{k+1},
    \end{align*}
    where $b_{k,1} \in \cB$ (there is no node $a$ from $M$ on the path from $v_i$ to $b_{k,1}$, because if $a \in \cX \cup \cZ \cup \{Y\}$, \cref{assump:civ-descendants} would be violated, and if $a \in \cI$, this would constitute a path in $\Gfull$ from $Y$ to $\cI$ that is unblocked given $\cB$, using the same argument as for the original path).
    Similarly, if there are several colliders on $w_k - \cdots - w_{k+1}$, a path $w_k \rightarrow b_{k,1} \leftrightarrow\cdots\leftrightarrow b_{k,L} \leftarrow w_{k+1}$ (or one of the configurations $\rightarrow \cdots \leftrightarrow$, $\leftrightarrow\cdots \leftarrow$ or $\leftrightarrow\cdots\leftrightarrow$ as first and last edge) is present in $\G_M$, where $b_{k,1}, \ldots, b_{k,L} \in \cB$. 
    
    We now construct a path $p_M$ in $\G_M$ that is 
    $d$-connecting $\cI$ and $Y$, given $\cB$: For $k = 0, \ldots, L-1$, paste together the segments (in $\G_M$) from $w_k$ to $w_{k+1}$ including those possible colliders $b_{k,j}$ discussed above. 
    Further, add the edge $w_L \rightarrow y$ or $w_L \leftrightarrow y$, depending on the orientation of the edge $w_{L} - m$ in $p$. 
    If $w_k$ was a collider on $p$, it is also a collider on $p_M$.
    Since $p$ was unblocked in $\Gfull$, given $\cB$, $p_M$ is unblocked in $\G_M$, given $\cB$: consider a collider on $p_M$; either it is one of the $b_{k,j}$ (in this case, it does not block $p_M$) or it is one of the $w_k$ (in this case, $w_k$ is also a collider on $p$ and thus has a descendant in $\cB$, which is still a descendant of $w_k$ in $\G_M$; again, it does not block $p_M$).
    
    We now turn to b) and argue that the path is still unblocked in $\G_M$, given $\cB$, if we remove the edges specified in \cref{assump:civ-d-sep-tce}.
    First, assume for some $w_k \in \cX \cup \cZ$ that $p_M$ contained the segment $w_k \rightarrow w_{k+1}$. Then there would be some $k+2 \leq k' \leq L+1$ such that $w_{k'} \notin \DE{\cX\cup\cZ}_{\Gfull}$ (because otherwise $m$ would be a descendant of $\cX\cup \cZ$, violating~\eqref{eq:Rpropert}). But this would imply that $p_M$ has a collider, which is a descendant of $\cX\cup \cZ$ and an ancestor of $\cB$ (in $\G_M$), which is not possible by \cref{assump:civ-descendants}. Thus, $p_M$ cannot contain an edge $w_k \rightarrow w_{k+1}$ where $w_k \in \cX \cup \cZ$. 
    
    Second, assume that $p_M$ does not contain an edge $w_k \rightarrow w_{k+1}$ with $w_k \in \cX \cup \cZ$ and that for some $w_k\in\cX\cup\cZ$, $p_M$ contains an edge $w_{k-1} \stackrel{e}{\leftarrow} w_k$ that satisfies the criterion for removal under \cref{assump:civ-d-sep-tce}. We can choose $w_k$ such that $e$ is the first edge on $p_M$ satisfying the criterion and, by definition of the criterion, there exists a directed path $w_{k-1} \rightarrow u_1 \rightarrow\cdots\rightarrow u_t \rightarrow Y$ in $\G_M$ where $u_i \notin \cX\cup\cZ$ and $u_i \notin \cB$ (because otherwise $\cB \cap \DE{w_k}_{\G_M} \neq \emptyset$). Because $p_M$ is unblocked given $\cB$ in $\G_M$ and because no edge is removed on $\tilde{p}_M = w_0 - \cdots - w_{k-1} \rightarrow u_1 \rightarrow \cdots \rightarrow Y$, $\tilde{p}_M$ is a violation of \cref{assump:civ-d-sep-tce}, and therefore $p_M$ does not contain any edge $w_{k-1} \leftarrow w_k$ that satisfies the criterion for removal under \cref{assump:civ-d-sep-tce}. 
    
    Combining `First' and `Second', the path $p_M$ does not contain any edges that would be removed under the criterion in \cref{assump:civ-d-sep-tce}. Also, since $p_M$ is unblocked given $\cB$, any collider is an ancestor of some $b\in \cB$, and that collider is still an ancestor of $b\in\cB$ after removing edges satisfying the criterion in \cref{assump:civ-d-sep-tce} (otherwise $b\in\DE{\cX\cup\cZ}_{\G_M}$, violating \cref{assump:civ-descendants}). In conclusion, $p_M$ is also unblocked given $\cB$ in the graph where we remove edges satisfying the criterion in \cref{assump:civ-d-sep-tce} from $\G_M$. 
    This concludes the proof of (1).
    
    A similar argument proves (2), that is, $\E[\cov(\epsilon, \cI|\cB)]=0$:
    Each $\epsilon^i$ was accumulated as a noise variable of an ancestor (in $\Gfull$) of $Y$, $A^i \notin \cX \cup \cZ \cup \cB$; 
    by construction, a directed path from $\cZ\cup\cX$ through $A^i$ to $Y$ exists in $\Gfull$ that does not intersect $\cX \cup \cZ \cup \cB$, except at the first node of this path. 
    Hence, $\cB$ does not contain a descendant of $A^i$ in $\Gfull$ (because that would imply $\cB$ containing a descendant of $\cX\cup\cZ$ in $\G_M$, 
    violating \cref{assump:civ-descendants}). Also, $A^i$ is not an ancestor of any node in $\cI$ in $\Gfull$, 
    because that would imply an unblocked path from $\cI$ via $A_i$ to $Y$ in $\Gfull$ that does not contain any node in $\cB$ and therefore this corresponds to an unblocked path in $\G_M$, too.
    Using the MA($\infty$)-representation of $\cB$ and $\cI$, see \citet[][Sec.\ 11.3]{Brockwell1991} or \citet{hamilton1994time},
    $\epsilon^i$ is independent of $(\cB, \cI)$, and so $\E[\cov(\epsilon^i, \cI|\cB)] = 0$.
    This concludes the proof of the first part.
    
    Part (ii) follows because if the $(d_{\cX} + d_{\cZ})\times d_{\cI}$ matrix $\E[\cov(\tilde{\cX}, \cI|\cB)]$ has rank $d_{\cX} + d_{\cZ}$, then if a solution to the moment equation $\E[\cov(Y, \cI|\cB)] = \beta \E[\cov(\tilde{\cX}, \cI|\cB)]$ exists, it is unique. 

    For part (iii), let $\bar{\bX} \coloneqq [\begin{matrix}\bX{^\top} & \bZ{^\top}\end{matrix}]^\top$. By \cref{eq:civ-closed-solution}, we have to show that
\begin{align*}
    \hat\E[r_\bY r_\bI^\top] \, \, W\, \,  \hat\E[r_\bI r_{\bar{\bX}}^\top]\bigg(\hat\E[r_{\bar{\bX}} r_\bI^\top]\, \, W\, \, \hat\E[r_\bI r_{\bar{\bX}}^\top]\bigg)^{-1}\xrightarrow[]{P}\gamma
\end{align*}
with $\gamma \coloneqq \vecin{\beta^\top, \alpha^\top}^\top$. From \cref{eq:assump_moments} we have that empirical moments converge in probability to the population moment, and thus, using Slutsky's Theorem, we get that 
\begin{align*}
    \hat\E[r_\bY r_\bI^\top] \, \, W\, \,  \hat\E[r_\bI r_{\bar{\bX}}^\top]\bigg(\hat\E[r_{\bar{\bX}} r_\bI^\top]\, \, W\, \, \hat\E[r_\bI r_{\bar{\bX}}^\top]\bigg)^{-1}\xrightarrow[]{P}\E[r_{Y_t} r_{\cI_t}^\top] \, \, W\, \,  \E[r_{\cI_t} r_{\bar{\cX_t}}^\top]\bigg(\E[r_{\bar{\cX_t}} r_{\cI_t}^\top]\, \, W\, \, \E[r_{\cI_t} r_{\bar{\cX_t}}^\top]\bigg)^{-1},
\end{align*}  
where $r_{\bar{\cX}} = \bar{\cX} - \E[\bar{\cX}|\cB]$ and similarly for $r_Y$ and $r_\cI$.
We can rewrite $\E[r_{Y_t} r_{\cI_t}^\top]$ by adding and subtracting $\gamma r_{\bar{\cX_t}}$:
\begin{align*}
    \E[r_{Y_t} r_{\cI_t}^\top]&=\E[(r_{Y_t}-\gamma r_{\bar{\cX_t}})  r_{\cI_t}^\top]+\gamma\E[r_{\bar{\cX_t}} r_{\cI_t}^\top]\\&=0+\gamma\E[r_{\bar{\cX_t}} r_{\cI_t}^\top].
\end{align*}
The first term is zero due to the conditional uncorrelation established in (i) and we can thus conclude that 
\begin{align*}
    \E[r_{Y_t} r_{\cI_t}^\top] \, \, W\, \,  \E[r_{\cI_t} r_{\bar{\cX_t}}^\top]&\bigg(\E[r_{\bar{\cX_t}} r_{\cI_t}^\top]\, \, W\, \, \E[r_{\cI_t} r_{\bar{\cX_t}}^\top]\bigg)^{-1}\\
    &=\gamma\E[r_{\bar{\cX_t}} r_{\cI_t}^\top] \, \, W\, \,  \E[r_{\cI_t} r_{\bar{\cX_t}}^\top]\bigg(\E[r_{\bar{\cX_t}} r_{\cI_t}^\top]\, \, W\, \, \E[r_{\cI_t} r_{\bar{\cX_t}}^\top]\bigg)^{-1}\\&=\gamma.
\end{align*}
\end{proof}

\subsection{Proof of \texorpdfstring{\cref{failure:naive_iv}}{}}
\failureNaive*
\begin{proof}
Since $\epsilon_t^Y$ and $H_{t-1}$ are both independent of $I_{t-2}$ (for $H_{t-1}$, this follows from \cref{assump:iv,thm:gmp}), it follows that 
\begin{align*}
    \E[Y_{t}I_{t-2}] &= \alpha_{Y,Y}\alpha_{I,I}\E[Y_{t-1}I_{t-3}] + \beta \E[X_{t-1} I_{t-2}] \\
    \implies \E[Y_{t}I_{t-2}] &= (1-\alpha_{Y,Y}\alpha_{I,I})^{-1}\beta \E[X_{t-1} I_{t-2}],
\end{align*}
where in the last step we use that by covariance stationarity we have $\E[Y_{t-1}I_{t-3}]=\E[Y_{t}I_{t-2}]$. 
The $\iv{X_{t-1}}{Y_t}(I_{t-2})$ moment equation is $\E[(Y_t - b X_{t-1})I_{t-2}] = 0$, which has the solution (because $d_I=d_X=d_Y=1$)
\begin{align*}
    b=\frac{\E[I_{t-2}Y_{t}]}{\E[I_{t-2}X_{t-1}]}.
\end{align*}
By plugging in the expression for $\E[Y_tI_{t-2}]$ above, we get $b = (1-\alpha_{Y,Y}\alpha_{I,I})^{-1}\beta$.
\end{proof}

\subsection{Proof of \texorpdfstring{\cref{prop:ts-civ}}{}}
\tsCiv*
\begin{proof}
By \cref{prop:dream-theorem}, it suffices for part (i) to show that \cref{assump:civ-d-sep-tce,assump:civ-descendants} are satisfied for $\cX_t \coloneqq \{X_{t-1}\}, \cI_t \coloneqq \{I_{t-2}\}$, $\cB_t$, and $Y_t$ in the marginalized graph $\G_{M_t}$ with $M_t \coloneqq \cX_t \cup \cI_t \cup \cB_t \cup \{Y_t\}$ (with $\cB_t$ being either of the two sets from the theorem), see \cref{fig:thm2_marginalized} \textit{left} and \textit{middle}. 
For either choice of $\cB_t$, $\cB_t$ is not a descendant of $X_{t-1}$ and $Y_t$, so \cref{assump:civ-descendants} is satisfied (see \cref{fig:thm2_marginalized} \textit{left} and \textit{middle}).
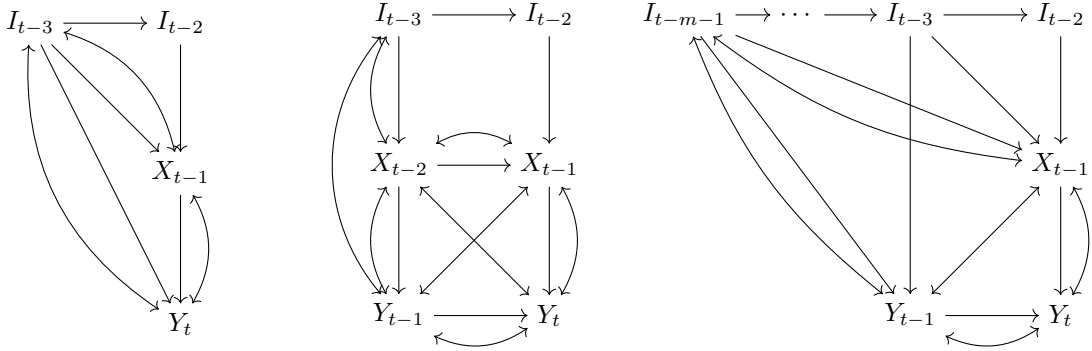
\begin{figure}
    \centering
    \begin{subfigure}{0.27\linewidth}
        \begin{tikzpicture}
        \node (I) at (2,4) {$I_{t-2}$};
        \node (X) at (2,2) {$X_{t-1}$};
        \node (Y) at (2,0) {$Y_{t}$};
        \node (I-) at (0,4) {$I_{t-3}$};
        \draw[->] (I) edge (X);
        \draw[->] (X) edge (Y);
        \draw[->] (I-) edge (I);
        \draw[<->] (X) edge[bend left=30] (Y);
        \draw[<->] (I-) edge[bend left=30] (X);
        \draw[->] (I-) edge (X);
        \draw[->] (I-) edge (Y);
        \draw[<->] (I-) edge[bend right=30] (Y);
        \end{tikzpicture}
    \end{subfigure}
    \begin{subfigure}{0.27\linewidth}
        \begin{tikzpicture}
        \node (I) at (2,4) {$I_{t-2}$};
        \node (X) at (2,2) {$X_{t-1}$};
        \node (Y) at (2,0) {$Y_{t}$};
        \node (I-) at (0,4) {$I_{t-3}$};
        \node (X-) at (0,2) {$X_{t-2}$};
        \node (Y-) at (0,0) {$Y_{t-1}$};
        \draw[->] (I) edge (X);
        \draw[->] (X) edge (Y);
        \draw[->] (I-) edge (I);
        \draw[->] (I-) edge (X-);
        \draw[->] (X-) edge (Y-);
        \draw[->] (Y-) edge (Y);
        \draw[->] (X-) edge (X);
        \draw[<->] (I-) edge[bend right=40] (Y-);
        \draw[<->] (I-) edge[bend right=30] (X-);
        \draw[<->] (Y-) edge[bend right=30] (Y);
        \draw[<->] (X-) edge[bend right=30] (Y-);
        \draw[<->] (X) edge[bend left=30] (Y);
        \draw[<->] (X-) edge[bend left=30] (X);
        \draw[<->] (X-) edge (Y);
        \draw[<->] (Y-) edge (X);
        \end{tikzpicture}
    \end{subfigure}
    \begin{subfigure}{0.4\linewidth}
        \begin{tikzpicture}
        \node (I) at (2,4) {$I_{t-2}$};
        \node (X) at (2,2) {$X_{t-1}$};
        \node (Y) at (2,0) {$Y_{t}$};
        \node (I-) at (0,4) {$I_{t-3}$};
        \node (I--) at (-3,4) {$I_{t-m-1}$};
        \node (dots) at (-1.5,4) {$\cdots$};
        \node (Y-) at (0,0) {$Y_{t-1}$};
        \draw[->] (I) edge (X);
        \draw[->] (X) edge (Y);
        \draw[->] (I-) edge (I);
        \draw[->] (I-) edge (Y-);
        \draw[->] (Y-) edge (Y);
        \draw[->] (I-) edge (X);
        \draw[<->] (Y-) edge[bend right=30] (Y);
        \draw[<->] (X) edge[bend left=30] (Y);
        \draw[<->] (Y-) edge (X);
        \draw[->] (I--) edge (dots);
        \draw[->] (dots) edge (I-);
        \draw[->] (I--) edge (Y-);
        \draw[->] (I--) edge (X);
        \draw[<->] (I--) edge[bend right=15] (X);
        \draw[<->] (I--) edge[bend right=15] (Y-);
        \end{tikzpicture}
    \end{subfigure}     \caption{(\textit{left}) Marginalization of the full time graph to nodes $I_{t-2}, I_{t-3}, X_{t-1}$ and $Y_t$.   (\textit{middle}) Marginalization of the full time graph to nodes $I_{t-2}, X_{t-1}$ and $Y_t$ and their lagged values. (\textit{right}) Marginalization to $m$ instrument nodes $I_{t-2}, \ldots, I_{t-m-1}$, and $X_{t-1}, Y_t,$ and $Y_{t-1}$.}
    \label{fig:thm2_marginalized}
\end{figure}
To show that \cref{assump:civ-d-sep-tce} holds, we argue that every path from $I_{t-2}$ to $Y_t$ is blocked by $I_{t-3}$ in $\G_{M{_t}(X_{t-1}\not\rightarrow Y_t)}$, the graph obtained from $\G_{M_t}$ by removing the directed edge from $X_{t-1}$ to $Y_t$. 
For either graph, we have that any path from $I_{t-2}$ to $Y_t$ either contains the non-collider $I_{t-3}$ or the collider $X_{t-1}$. 
Since $I_{t-3}$ is in the conditioning set $\cB_t$ (for either definition of $\cB_t$) and $(\{X_{t-1}\}\cup \DE{X_{t-1}})\cap \cB_t=\emptyset$, any path from $I_{t-2}$ to $Y_t$ in $\G_{M{_t}(X_{t-1}\not\rightarrow Y_t)}$ is blocked by $\cB_t$. 

Parts (ii) and (iii) follow directly from \cref{prop:dream-theorem}.
\end{proof}

\subsection{Proof of \texorpdfstring{\cref{prop:ts-niv}}{}}
\tsNiv*
\begin{proof}
{Let $\vecin{\beta,\alpha}$ be the total causal effect of $\vecin{X_{t-1}, \cZ_t}$ on $Y_t$.} By \cref{prop:dream-theorem}, it suffices for part (i) to show that \cref{assump:civ-d-sep-tce,assump:civ-descendants} are satisfied for $\cX_t \coloneqq \{X_{t-1}\}, \cZ_t, \cI_t$, and $Y_t$ in the marginalized graph $\G_{M_t}$ with $M_t \coloneqq \cX_t \cup \cI_t \cup \cZ_t \cup \{Y_t\}$ (see \cref{fig:thm2_marginalized} \textit{right}). 
Since $\cB=\emptyset$, \cref{assump:civ-descendants} is trivially satisfied. 
It remains to argue that \cref{assump:civ-d-sep} 
\cref{assump:civ-d-sep-tce}
is satisfied, that is, that $\cI_t$ is $d$-separated from $Y_t$ in the marginalized graph
$\G_{M_t(X_{t-1},Y_{t-1} \not \rightarrow Y_t)}$ obtained from $\G_{M_t}$ by removing the edges $Y_{t-1}\rightarrow Y_t$ and $X_{t-1}\rightarrow Y_t$. Let $s\in\{t-m-1,...,t-2\}$. Every path from $I_s$ to $Y_t$ must go through either the collider $\rightarrow X_{t-1}\leftarrow$ or the collider $\rightarrow Y_{t-1}\leftarrow$ and since the conditioning set is empty, those paths are blocked. 

Parts (ii) and (iii) follow directly from \cref{prop:dream-theorem}.
\end{proof}

\subsection{Proof of \texorpdfstring{\cref{prop:identifiability}}{}}
\label{sec:proof-jordan-forms}
\subsubsection*{A brief review of Jordan canonical forms}
If $M$ is an arbitrary square matrix of size $d \times d$, there exists a unique (up to row or column permutations) square invertible matrix $Q$ of the same dimension such that $M = QJQ^{-1}$ where $J$ is a $d \times d$ block diagonal matrix 
\begin{equation} \label{eq:jordann}
    J = J_{m_1}(\lambda_1) \oplus \ldots \oplus J_{m_k}(\lambda_k):= \text{diag}(J_{m_1}(\lambda_1), \ldots, J_{m_k}(\lambda_k))
\end{equation}
with each \emph{Jordan block} $J_{m_i}(\lambda_i)$ being an $m_i \times m_i$ matrix having one value $\lambda_i$ on the diagonal and ones on the superdiagonal (and zeros elsewhere): that is, for all $m \in \mathbb{N}_{>0}$,
\begin{equation*}
    J_m(\lambda) \coloneqq \mat{
    \lambda & 1 & & \\ 
    & \lambda & \ddots  & \\ 
    & & \ddots & 1 \\
    & & & \lambda}.
\end{equation*}
We sometimes write $J_m$ instead of $J_m(\lambda)$ and
call \cref{eq:jordann} the Jordan canonical form.
Jordan forms and the involved matrices satisfy the following properties \citep{horn1985matrix}.
\begin{itemize}
	\item Let $N_m$ (we simply write $N$ if the dimension is obvious) be the canonical nilpotent matrix of degree $m$, that is the $m\times m$-matrix with ones in the superdiagonal and zeroes elsewhere. Then $J_m(\lambda) = \lambda 1 + N$ and by the binomial formula $J^n_m = \sum_{i=0}^n \mat{n \\ i}\lambda^{n-i} N^{i}$. 
	\item Every diagonal value of a Jordan block is an eigenvalue of $M$ and for every eigenvalue $\lambda$ of $M$, there is at least one Jordan block with diagonal $\lambda$. There may however be more than one Jordan block for the same eigenvalue. 
	\item The geometric multiplicity of an eigenvalue $\lambda$ is the number of corresponding Jordan blocks
	\item The algebraic multiplicity of an eigenvalue $\lambda$ is the sum of the sizes $m_i$ of the corresponding Jordan blocks. 
	\item If $M$ is diagonalizable, all Jordan blocks are of size one, 
	which is equivalent to
	the algebraic and geometric multiplicities being equal.
\end{itemize}
\subsubsection*{Some Lemmata for the proof of \texorpdfstring{\cref{prop:identifiability}}{}}
We say that a vector $v$ in a $d$-dimensional vector space is \emph{cyclic of the $d\times d$ matrix $J$} if $v, Jv, \ldots, J^{d-1}v$ constitute a basis for the vector space. 
\begin{lemma}\label{lemma:cyclic_jordan}
    Let $J = J_{m_1}(\lambda_1) \oplus \ldots \oplus  J_{m_k}(\lambda_k)$ be a block Jordan form over $\mathbb{C}$ for a square matrix $J$.
    If two or more blocks have the same eigenvalue, no vector $v \in \mathbb{R}^{\sum_{i=1}^k m_i}$ is cyclic of $J$.
\end{lemma}
\begin{proof}
It suffices to consider the case $J = J_{m_1}(\lambda) \oplus J_{m_2}(\lambda)$ where without loss of generality $m_1 \geq m_2$. For $J_{m_1}(\lambda) = \lambda 1 + N_{m_1}$ and $J_{m_2}(\lambda) = \lambda 1 + N_{m_2}$ the degree $m_1$ minimal polynomial $p(x) = (x-\lambda)^{m_1}$ annihilates $J$ such that $p(J) = 0$. Consequently $J^{m_1}$ can be written as a linear combination of $J^0, \ldots, J^{m_1 - 1}$. In particular $J^0v, \ldots, J^{m_1 + m_2 - 1}v$ cannot be linearly independent. 
\end{proof}

\begin{lemma}\label{lemma:cyclic_vector_jordan}
    Let 
    $J = J_{m_1}(\lambda_1) \oplus \ldots \oplus  J_{m_k}(\lambda_k)$ be a block Jordan form over $\mathbb{C}$ for a square matrix $J$,
    with each block corresponding to a distinct eigenvalue $\lambda_i$. Then $v \in \mathbb{C}^{\sum_{i=1}^k m_i}$ is a cyclic vector for $J$ if and only if for each $d = 1,\ldots,k$ the entry $v_{\sum_{i=1}^d m_i}$ is non-zero.
\end{lemma}
\begin{proof}
We first show by contraposition that if $v$ is cyclic for $J$, the corresponding entries will be non-zero. 
If it does not hold that for each $d = 1,\ldots,k$ the entry $v_{\sum_{i=1}^d m_i}$ is non-zero, we may, without loss of generality, assume that the last entry of $v$ is zero, such that $v = [u, 0]^\top$ for suitable $u\in \mathbb{C}^{-1 + \sum_{i=1}^k m_i}$. Denote $\lambda = \lambda_k$ the eigenvalue corresponding to the last Jordan block $J_{m_k}(\lambda_k)$; observe that the bottom row of $J^n$ is $[0, \ldots, 0, \lambda^n]$ for any power $n$, and so the last entry of $J^n v$ is $0$ for every $n$ and consequently the matrix $[J^0 v, J^1 v, \ldots, J^{(\sum_{i=1}^k m_i)-1} v]$ has a $0$-row. Consequently, $v$ is not cyclic of $J$. This shows that if $v$ is cyclic, the entries $v_{\sum_{i=1}^d m_i}$ are non-zero.

Now we show the other implication by induction over $k$. Assume first that $k=1$, i.e., $J= J_{m}(\lambda)$ consists of a single Jordan block. 
A vector $v = [v_1, \ldots, v_m]^\top$ is cyclic of $J$ if $v_m \neq 0$. Indeed, consider coefficients $a_0, \ldots, a_{m-1}$ such that $0 = \sum_{n=0}^{m-1}a_n J^{n}v$. Recall that $J^{n} = \sum_{i=0}^n \mat{n \\ i}\lambda^{n-i} N^{i}$, which implies
\begin{equation*}
    0 = \sum_{n=0}^{m-1}a_n J^{n}v 
    = \sum_{n=0}^{m-1}a_n \left(\sum_{i=0}^n \mat{n \\ i}\lambda^{n-i} N^{i}\right)v 
    = \sum_{i=0}^{m-1}\left(\sum_{n=i}^{m-1} a_n\mat{n\\i}\lambda^{n-i}\right)N^{i}v,
\end{equation*}
where in the final equality, we swap the order of summation, using that the pairs $(n,i)$ where $n \in \{0, \ldots, m-1\}$ and $i \in \{0, \ldots, n\}$, are the same as the pairs $(n,i)$ where $i \in \{0, \ldots, m-1\}$ and $n \in \{i, \ldots, m-1\}$.
Since $v_m \neq 0$ the collection $N^0 v, \ldots, N^{m-1}v$ are linearly independent: they form an upper-triangular matrix with $v_m$ on the diagonal. This implies that, in particular, the coefficient on $N^{m-1}v$ must be $0$. But this coefficient equals $\sum_{n=m-1}^{m-1} a_n\mat{n\\m-1}\lambda^{n-(m-1)} = a_{m-1}$, and so $a_{m-1}=0$. Substituting this into the coefficient on $N^{m-2}v$, one obtains $a_{m-2} = 0$ and so forth. 
Therefore, $a_n = 0$ for all $n$ and thus $J^0v, \ldots, J^{m-1}v$ are linearly independent, so $v$ is cyclic of $J$ if $v_m \neq 0$. 

Next assume that the induction hypothesis holds for any matrix with $k$ Jordan blocks $J = J_{m_1}(\lambda_1) \oplus \ldots \oplus J_{m_k}(\lambda_k)$ with distinct eigenvalues for each block and for %
all $v = [v_1, \ldots, v_{\sum_{i=1}^k m_i}]^\top$ where for every $d = 1,\ldots,k$: $v_{\sum_{i=1}^d m_i} \neq 0$. Now consider the additional Jordan block $D = J_{m_{k+1}}(\lambda_{k+1})$ where $\lambda_{k+1}\neq \lambda_1, \ldots \lambda_k$ and the vector $u$ whose last entry $u_{m_{k+1}}$ is non-zero, and let $\tilde J = J \oplus D$, $\tilde v = [v^\top, u^\top]^\top$.

Define the polynomial $p$ of degree $\sum_{i=1}^k m_i$ by $p(x) = (x-\lambda_1)^{m_1}(x-\lambda_2)^{m_2}\cdots(x-\lambda_k)^{m_k}$. Observe that $p(\lambda_{k+1}) \neq 0$ and so $p(D)$ is an upper triangular matrix with $p(\lambda_{k+1})$ on the diagonal. Hence the last entry of the vector $p(D)u$ is $p(\lambda_{k+1})u_{m_{k+1}}$ which is non-zero, and so $p(D)u$ is cyclic of $D$ (by the initial step of the induction proof). Further observe that $p$ annihilates each of the previous blocks because $J_{m_i}(\lambda_i) = \lambda_i 1 + N$ so $(J_{m_i}(\lambda_i) - \lambda_i 1)^{m_i} = N^{m_i} = 0$. Consequently,
\begin{align*}
p(\tilde{J}) &= p(J_{m_1}(\lambda_1)) \oplus \ldots \oplus p(J_{m_k}(\lambda_k)) \oplus p(D) \\
&= 0 \oplus \ldots \oplus 0 \oplus p(D),
\end{align*}
so that
\begin{align}\label{eq:annihilating_polynomial}
    p(\tilde{J}) \tilde{v} = [0^\top, \ldots, 0^\top, (p(D)u)^\top]^\top.
\end{align}
Now to show $\tilde{v}$ is cyclic of $\tilde{J}$, we take any vector $x \in \mathbb{C}^{\sum_{i=1}^k m_i}$ and $y \in \mathbb{C}^{m_{k+1}}$. Our aim is to show that $[x^\top, y^\top]^\top$ is in the span of $\tilde{J}^0 \tilde{v}, \ldots, \tilde{J}^{(\sum_{i=1}^{k+1} m_i) - 1}\tilde{v}$. Since $v$ is cyclic of $J$, $x$ can be expressed as a linear combination of $J^0v, \ldots, J^{\sum_{i=1}^k m_i-1}v$. Taking the same linear combination of $\tilde{J}^0\tilde{v}, \ldots, \tilde{J}^{\sum_{i=1}^k m_i-1} \tilde{v}$ yields the vector $[x^\top,z^\top]^\top$ for some $z \in \mathbb{C}^{m_{k+1}}$. Since $p(D)u$ is cyclic of $D$, we can write $y - z$ as a linear combination of $D^0p(D)u, D^1p(D)u, \ldots, D^{m_{k+1}-1}p(D)u$. It follows from \cref{eq:annihilating_polynomial} that by taking the same linear combination of $\tilde{J}^0p(\tilde{J})\tilde{v},\ldots, \tilde{J}^{m_{k+1}-1} p(\tilde{J})\tilde{v}$ 
one obtains $[0^\top, (y-z)^\top]^\top$. Since $p$ is a polynomial of degree $\sum_{i=1}^k m_i$, it follows that both $[x^\top, z^\top]^\top$ and $[0^\top, (y-z)^\top]^\top$ lie in the span of $\tilde{J}^0 \tilde{v}, \ldots, \tilde{J}^{\sum_{i=1}^{k+1} m_i -1}\tilde{v}$, and so does $[x^\top, y^\top]^\top$. 
Since $x$ and $y$ were arbitrary, the entire space is spanned, completing the induction step.
\end{proof}

\subsubsection*{Proof of \texorpdfstring{\cref{prop:identifiability}}{}}
\jordanIdentifiability*
\begin{proof}
First observe that 
\begin{equation*}
    \mat{X_t \\ Y_t} = A_I I_{t-1} + A_{XY}\mat{X_{t-1}\\Y_{t-1}} + \underbrace{\mat{\nu_X\\ \nu_Y} H_{t-1} + \epsilon^{X,Y}_t}_{\text{uncorrelated to I}}
\end{equation*}
and consequently: 
\begin{equation*}
\E\left[\mat{X_t\\ Y_t}\, I_t\right] = \E\left[\left(A_I I_{t-1} + A_{XY}\mat{X_{t-1}\\ Y_{t-1}}\right)\, \alpha_{I,I} I_{t-1} \right].
\end{equation*}
From this, we obtain
\begin{equation*}
    \E\left[\mat{X_t\\ Y_t}\, I_t\right] = \underbrace{\E[I_t^2]}_{=:v_I}\alpha_{I,I} \underbrace{(1-\alpha_{I,I}A_{XY})^{-1}}_{=:B^{-1}}A_I.
\end{equation*}
This expression is justified as $B$ is invertible. (Indeed, if $\alpha_{I,I} = 0$, this is trivial. If $\alpha_{I,I}\neq 0$, since $e_1^\top A_1 = \alpha_{I,I} e_1^\top$, where $A_1$ is coefficient matrix assumed in \cref{assump:iv} and $e_1 = \vecin{1, 0, \ldots, 0}^\top$ is the first unit vector, $\alpha_{I,I}$ is an eigenvalue of $A_1^T$ and thus of $A_1$ and in particular, by \cref{assump:varp},
it has absolute value strictly smaller than $1$. $B$ is degenerate if and only if $A_{XY} - \frac{1}{\alpha_{I,I}}$ is, but this would imply that $\frac{1}{\alpha_{I,I}}$ would be an eigenvalue of $A_{XY}$, but since the eigenvalues of $A_{XY}$ are also eigenvalues of $A_1$ (if $A_{XY}v = \lambda v$, then $A_1 \vecin{0, v^\top}^\top = \lambda\vecin{0, v^\top}^\top$) and belong to the interior of the unit circle, $\frac{1}{\alpha_{I,I}}$ cannot be an eigenvalue of $A_{XY}$.)

By performing the same expansion for $\E[\vecin{X_t^\top, Y_t^\top}^\top I_{t-j}]$ for $j\geq 1$ and plugging in the above, we obtain: 
\begin{align*}
\E\left(\mat{X_t\\ Y_t}\, I_{t-1}\right) &= A_Iv_I + A_{XY}\E\left(\mat{X_{t-1}\\ Y_{t-1}}\, I_{t-1}\right)= v_I B^{-1} A_I \notag\\
\E\left(\mat{X_t\\ Y_t}\, I_{t-2}\right) &= v_I \left[A_{XY} B^{-1} + \alpha_{I,I} 1\right] A_I \notag\\
\E\left(\mat{X_t\\ Y_t}\, I_{t-3}\right) &= v_I\left[A_{XY}^{2}B^{-1} + \alpha_{I,I} A_{XY}  + \alpha_{I,I}^2 1\right] A_I \notag
\end{align*}
and in general:
\begin{align}\label{eq:columns}
\E\left(\mat{X_t\\ Y_t}\, I_{t-1-j}\right) = v_I \left[ A_{XY}^j B^{-1} + \sum_{k=0}^{j-1} \alpha_{I,I}^{j-k} A_{XY}^k\right] A_I.
\end{align}
The columns (denote the $j$'th column by $\textrm{col}_j$) of $\Sigma:=\E[\vecin{X_{t-1},Y_{t-1}} \cI_t^\top]$
are exactly those given by \cref{eq:columns}. If we deduct $\alpha_{I,I} \text{col}_{j-1}$ from $\text{col}_{j}$ we obtain: 
\begin{align*}
	&v_I\left[\left(\sum_{k=0}^{j-1}\alpha_{I,I}^{j-k}A_{XY}^k + A_{XY}^jB^{-1}\right)-	\alpha_{I,I}\left(\sum_{k=0}^{j-2}\alpha_{I,I}^{j-1-k}A_{XY}^k + A_{XY}^{j-1}B^{-1}\right)\right]A_I \\
	&=v_I\left[\alpha_{I,I} A_{XY}^{j-1}+A_{XY}^{j}B^{-1} - \alpha_{I,I} A_{XY}^{j-1}B^{-1} \right]A_I\\
	&= v_I(1-\alpha_{I,I}^2)A_{XY}^{j} B^{-1}A_I.
\end{align*}
Since deducting columns from each other does not change the determinant, we can create a simpler matrix, $\Sigma_{\text{equiv}}$, with the same determinant: for $j \in\{2, \ldots, k\}$ we deduct $\alpha_{I,I}\text{col}_{j-1}$ from $\text{col}_j$ (starting with the largest $j$, that is first deducting $\alpha_{I,I}\text{col}_{k-1}$ from $\text{col}_k$, etc.), and obtain
\begin{equation*}
    \Sigma_{\text{equiv}}=v_I\vec{A_{XY}^0 B^{-1}A_I, \quad (1-\alpha_{I,I}^2)A_{XY}^1 B^{-1}A_I, \quad \ldots, \quad  (1-\alpha_{I,I}^2)A_{XY}^{d_X} B^{-1}A_I}.
\end{equation*}
By the Laplace expansion, removing $(1-\alpha_{I,I}^2)$ terms appearing in all but the first column scales the determinant by a factor $\frac{1}{(1-\alpha_{I,I}^2)^{d_X}}$, but it will not change its invertibility (from the requirement on the eigenvalues in \cref{assump:varp}, it follows that $1-\alpha_{I,I}^2 > 0$).
The same applies to $v_I = \E[I_t^2]$.
Hence, $\Sigma$ is invertible if and only if
\begin{equation*}
    \Sigma_{\text{equiv},2} \coloneqq\vec{A_{XY}^0 B^{-1}A_I, \quad A_{XY}^1 B^{-1}A_I, \quad \ldots, \quad A_{XY}^{d_X} B^{-1}A_I}
\end{equation*}
is invertible. Now observe that $B^{-1}$ commutes with $A_{XY}^j$. This follows because $BA_{XY} = (1-\alpha_{I,I} A_{XY})A_{XY} = A_{XY}(1-\alpha_{I,I} A_{XY}) = A_{XY}B$. This implies $B^{-1}A_{XY} = A_{XY}B^{-1}$, because for any matrix $M$ where $MB = BM$, it follows that
$$M = M B B^{-1} = B M B^{-1} \implies B^{-1}M = B^{-1}BMB^{-1} = MB^{-1}.$$
This implies that
\begin{equation*}
    \Sigma_{\text{equiv},2} = B^{-1} \underbrace{\vec{A_{XY}^0 A_I, \quad \ldots, \quad A_{XY}^{d_X} A_I}}_{=:\Sigma_{\text{equiv},3}}.
\end{equation*}
Since $B^{-1}$ is invertible, it has non-zero determinant, and again invertibility of $\Sigma$ is equivalent to invertibility of $\Sigma_{\text{equiv},3}$. 

Let $A_{XY} = QJQ^{-1}$ be the Jordan block factorization. Observe that 
\begin{align*}
	\left\{A_{XY}^0 A_I, \ldots, A_{XY}^{d_X} A_I\right\} = Q\left\{J^0 Q^{-1}A_I, \ldots, J^{d_X} Q^{-1}A_I\right\}.
\end{align*}
And finally, since $Q$ is invertible, invertibility of $\Sigma$ is equivalent to $Q^{-1}A_I$ being cyclic of $J$.
According to \cref{lemma:cyclic_jordan} if two or more Jordan blocks have the same eigenvalue, no vector can be cyclic, so in particular not $Q^{-1}A_I$. 
If on the contrary no eigenvalue is shared across Jordan blocks (equivalently, the geometric multiplicity of every eigenvalue is $1$), it follow from \cref{lemma:cyclic_vector_jordan} that $Q^{-1}A_I$ is a cyclic vector of $J$ 
if and only if the vector $Q^{-1}A_I$ is non-zero in the entries indexed by $\sum_{i=1}^d m_i$ for all $d = 1, \ldots, k$. Writing $A_I = Qa$ in the basis of the columns of $Q$ for some coefficient vector $a \in \mathbb{C}^{d_X+1}$, this means $Q^{-1}A_I$ is a cyclic vector for $J$ exactly when the coefficients $a_{\sum_{i=1}^d m_i}$ are non-zero for all $d = 1, \ldots, k$. This concludes the proof.
\end{proof}

\subsection{Proof of \texorpdfstring{\cref{thm:almost_sure_identify}}{}}
\almostSure*
\begin{proof}
    We check the conditions of \cref{prop:identifiability}. Since the entries are drawn from a density with respect to Lebesgues measure, the eigenvalues are almost surely distinct. Thus, taking into account the sparsity pattern,
     $A_{XY}$ can almost surely be diagonalized and the corresponding Jordan form has blocks of size one, all with distinct eigenvalues. 
    
    Also with probability one, $w=Q^{-1}\mat{\alpha_{X, I}\\0}$ does not have any zeroes:
    $Q$ is determined from $A_{XY}$ (so $Q$ depends only on $\alpha_{X,X}, \alpha_{Y,Y}$ and $\beta$), and so the probability that $[\alpha_{X,I}^\top, 0]^\top$ is orthogonal to any of the rows of $Q^{-1}$ is $0$.
\end{proof}
\subsection{Proof of \texorpdfstring{\cref{prop:identify_more_than_one_instrument}}{}}
\moreThanOneInstrument*
\begin{proof}
    Although the instruments $I^{(i)}, j \neq i$ are observed, we may treat them as latent, being part of the latent process $\tilde H_t \coloneqq (H_t, I_t^{(i, i \neq j)})$. By i), $I^{(j)}$ is independent of $\tilde H$. By ii), \cref{prop:ts-niv}, and \cref{prop:identifiability}, $\beta$ is identifiable in the reduced process $(I^{(j)}, X, Y)$, and the solution is therefore also unique in the full system $(I, X, Y)$.
\end{proof}

\subsection{Proof of \texorpdfstring{\cref{thm:CIV_general}}{}}
\tsCivgen*
\begin{proof}
To obtain $\E[\cov((Y_t - \beta X_{t-1})I_{t-2}^\top |\cB_t)] = 0$,
we show that
\begin{align*}
    Y_t-\beta X_{t-1}\indep I_{t-2}\,|\,\cB_t.
\end{align*}
Using that $Y_t-\beta X_{t-1}=\alpha_{Y,Y}Y_{t-1}+g(\varepsilon_t^Y,H_{t-1})$ it suffices to show that $Y_{t-1}\indep I_{t-2}|\cB_t$ and $(\varepsilon_t^Y,H_{t-1})\indep I_{t-2}|\cB_t$ since $g$ is measurable. For the first conditional independence we use that $Y_{t-1},\cB_t\subseteq\ND{I_{t-2}}$ and $\PA{(I_{t-2})}\subseteq\cB_t$ to conclude $Y_{t-1}\perp_d I_{t-2}|\cB_t$ in $\Gfull$ (Indeed, any path from $I_{t-2}$ to $Y_{t-1}$ leaves $I_{t-2}$ either through a parent of $I_{t-2}$ or must contain a collider that is a descendant of $I_{t-2}$.) By the global Markov property, contained in \cref{assump:CIV_general}, this implies $Y_{t-1}\indep I_{t-2}|\cB_t$. For the second conditional independence, we 
show $\varepsilon_t^Y\indep I_{t-2}|(\cB_t,H_{t-1})$ and $H_{t-1}\indep I_{t-2}|\cB_t$ and use the contraction property of conditional independence to obtain $(\varepsilon_t^Y,H_{t-1})\indep I_{t-2}|\cB_t$. Now, $H_{t-1}\indep I_{t-2}|\cB_t$ holds by the global Markov property since $H_{t-1}\in\ND{I_{t-2}}$ and $\PA{(I_{t-2})}\subseteq\cB_t$. To show $\varepsilon_t^Y\indep I_{t-2}|(\cB_t,H_{t-1})$, 
we use that by \cref{assump:CIV_general} $\varepsilon_t^Y$ is independent of any finite subset of $\ND{Y_t}$ in $\Gfull$. We have that $\cB_t \cup \{H_{t-1}\} \cup \{I_{t-2}\} \subseteq \ND{Y_t}$ 
and thus by weak union we get $\varepsilon_t^Y\indep I_{t-2}|(\cB_t,H_{t-1})$ as desired. This proves part (i).

Part (ii) follows because the moment equation $\E[\cov((Y_t - \beta X_{t-1})I_{t-2}^\top|\cB_t)] = 0$ 
is the same as in \cref{prop:ts-civ}, and so the rank requirement for identifiability is also the same.

To show part (iii), we have to show that 
\begin{align*}
    \hat\E[r_{\bY_t} r_{\bI_{t-2}}^\top] \, \, W\, \,  \hat\E[r_{\bI_{t-2}} r_{{\bX_{t-1}}}^\top]\bigg(\hat\E[r_{{\bX_{t-1}}} r_{\bI_{t-2}}^\top]\, \, W\, \, \hat\E[r_{\bI_{t-2}} r_{{\bX_{t-1}}}^\top]\bigg)^{-1}\xrightarrow[]{P}\beta.
\end{align*}
This is analogous to the argument in the proof of \cref{prop:dream-theorem}, except for that the convergence of empirical moments is now guaranteed by \cref{assump:empirical_moments} (instead of \cref{assump:varp}). Hence, using Slutsky's Theorem and rewriting $Y_t$ as in \cref{prop:dream-theorem}, gives the desired convergence.
\end{proof}

\subsection{Proof of \texorpdfstring{\cref{thm:NIV_general}}{}}
\tsNivgen*
\begin{proof}
By \cref{eq:structural-Y}, we have that $Y_t-\beta X_{t-1}-\alpha_{Y,Y} Y_{t-1}=g(\epsilon_t^Y, H_{t-1})$. Furthermore, by \cref{assump:inst_indep_of_noise} $(\varepsilon_t^Y,H_{t-1})\indep \cI_t$. Combining this (and using measureability of $g$) we obtain 
\begin{align*}
    Y_t-\beta X_{t-1}-\alpha_{Y,Y} Y_{t-1} \indep \cI_t.
\end{align*} 
Thus $\E[(Y_t-\beta X_{t-1}-\alpha_{Y,Y}Y_{t-1})\cI_t^\top]=0$ for $\alpha=\alpha_{Y,Y}$, and part (i) follows. 

Part (ii) follows because the moment equation $\E[(Y_t-\beta X_{t-1}-\alpha_{Y,Y}Y_{t-1})\cI_t^\top]=0$ is the same as in \cref{prop:ts-niv}, and so the rank requirement for identifiability is also the same.

To show part (iii), let
$\bar{\bX}_{t-1}:=[\begin{matrix}\bX_{t-1}{^\top},&\bY_{t-1}{^\top}\end{matrix}]^\top$. We have to show that
\begin{align*}
    \hat\E[\bY_t \bm{\cI}_t^\top] \, \, W\, \,  \hat\E[\bm{\cI}_t {\bar{\bX}_{t-1}}^\top]\bigg(\hat\E[{\bar{\bX}_{t-1}} \bm{\cI}_t^\top]\, \, W\, \, \hat\E[\bm{\cI}_t {\bar{\bX}_{t-1}}^\top]\bigg)^{-1}\convP\gamma
\end{align*}
with $\gamma:=[\begin{matrix}\beta ,& \alpha_{Y,Y}\end{matrix}]$. \cref{assump:empirical_moments} ensures convergence of the empirical moments to population moments, and using Slutsky's Theorem in combination with the expression for $Y_t$, the statement follows as in the proof of \cref{prop:dream-theorem}. 
\end{proof}

\subsection{Proof of \texorpdfstring{\cref{prop:optimal_prediction}}{}}
\optimalPrediction*
\begin{proof}
Recall that $\beta$ and $\alpha_{Y,Y}, \alpha_{Y,H}$ denote the causal effects from $X_t, Y_t$ and $H_t$, respectively, to $Y_{t+1}$. We have
\begin{align*}
&\min_{a, b, c}
\E_{do(X_{t} := x)} 
\left(Y_{t+1} - \sum_{j=0}^\ell a_j Y_{t-j} - bX_{t} - \sum_{k=1}^m c_k X_{t-k}\right)^2\\
=\,& 
\min_{a, b, c}
\E_{do(X_{t} := x)} 
\left(\{\beta X_t + \alpha_{Y,Y}Y_t + \alpha_{Y,H}H_t + \epsilon^{Y}_{t+1}\} - \sum_{j=0}^\ell a_j Y_{t-j} - bX_{t} - \sum_{k=1}^m c_k X_{t-k}\right)^2\\
=\,& 
\min_{a,b,c}
\E_{do(X_{t} := x)} 
(\beta X_{t} - b X_{t})^2\\
& \qquad
+ 
\E_{do(X_{t} := x)} (\beta X_{t} - b X_{t})
\left(\alpha_{Y,Y} Y_{t} + \alpha_{Y,H} H_{t} + \epsilon^{Y}_{t+1} - \sum_{j=0}^\ell a_j Y_{t-j} - \sum_{k=1}^m c_k X_{t-k}\right)\\
& \qquad
+
\E_{do(X_{t} := x)} \left(\alpha_{Y,Y} Y_{t} + \alpha_{Y,H} H_{t} 
 + \epsilon^{Y}_{t+1} - \sum_{j=0}^\ell a_j Y_{t-j} - \sum_{k=1}^m c_k X_{t-k}\right)^2\\
=\,& 
\min_{a,b,c}
\E_{do(X_{t} := x)} 
(\beta X_{t} - b X_{t})^2\\
& \qquad
+ 
\E_{do(X_{t} := x)} (\beta X_{t} - b X_{t})
\left(\alpha_{Y,Y} Y_{t} + \alpha_{Y,H} H_{t} + \epsilon^{Y}_{t+1} - \sum_{j=0}^\ell a_j Y_{t-j} - \sum_{k=1}^m c_k X_{t-k}\right)\\
& \qquad
+
\E \left(\alpha_{Y,Y} Y_{t} + \alpha_{Y,H} H_{t} 
 + \epsilon^{Y}_{t+1} - \sum_{j=0}^\ell a_j Y_{t-j} - \sum_{k=1}^m c_k X_{t-k}\right)^2\\
=\,& 
\min_{a,b,c}
\E_{do(X_{t} := x)} 
(\beta X_{t} - b X_{t})^2
+ 
(\beta x - b x)\E
\left(\alpha_{Y,Y} Y_{t} + \alpha_{Y,H} H_{t} + \epsilon^{Y}_{t+1} - \sum_{j=0}^\ell a_j Y_{t-j} - \sum_{k=1}^m c_k X_{t-k}\right)\\
& \qquad
+
\E \left(\alpha_{Y,Y} Y_{t} + \alpha_{Y,H} H_{t} 
 + \epsilon^{Y}_{t+1} - \sum_{j=0}^\ell a_j Y_{t-j} - \sum_{k=1}^m c_k X_{t-k}\right)^2\\
=\,& 
\min_{a,c} \E\left(Y_{t+1} - \beta X_t - \sum_{j=0}^\ell a_j Y_{t-j} - \sum_{k=1}^m c_k X_{t-k}\right)^2.
\end{align*}
Here, the third and fourth equality signs hold because 
the joint distribution of the variables 
$H_{t}$, $\epsilon^{Y}_{t+1}$, $Y_{t}, \ldots, Y_{t-\ell}$, and $X_{t-1}, \ldots, X_{t-m}$
is the same under the observational and the 
 interventional distribution -- as the variables are all non-descendants of $X_t$.
Further, the minimum is obtained for
$b = \beta$ and $a$ and $c$ being the coefficients after
(linearly) projecting 
$Y_{t+1} - \beta X_t$
on the space spanned by 
 $Y_t, \ldots, Y_{t-\ell}$
 and
$X_{t-1}, \ldots, X_{t-m}$.
\end{proof}

{
\subsection{Proof of \texorpdfstring{\cref{lemma:ts-civ-inhomogeneous}}{}} \label{sec:newproooof1}
\tsCivInhomogeneous*
\begin{proof}
    The full-time graph of $S$ is the same as that in \cref{prop:ts-civ}, and therefore it follows by the same proof as of \cref{prop:ts-civ} that $\E[\cov(Y_t - \beta_t X_{t-1}, I_{t-2}|\cB_t)] = 0$. From this, one can derive:
    \begin{align*}
        0 &= \E[\cov(Y_t - \beta_t X_{t-1}, I_{t-2}|\cB_t)] \\
        &= \E[\cov(Y_t, I_{t-2}|\cB_t)] -  \E[\cov(\beta_t X_{t-1}, I_{t-2}|\cB_t)] \\
        &= \E[\cov(Y_t, I_{t-2}|\cB_t)] -  \E[\beta_t]\E[\cov(X_{t-1}, I_{t-2}|\cB_t)] \\
        &= \E[\cov(Y_t - \E[\beta_t] X_{t-1}, I_{t-2}|\cB_t)].
    \end{align*}
    Here the third equality follows because
    \begin{align*}
        \E[\cov(\beta_t X_{t-1}, I_{t-2}|\cB_t)] & = \E\bigg[\E[\beta_t X_{t-1}I_{t-2}^\top |\cB_t] - \E[\beta_t X_{t-1}|\cB_t]\E[I_{t-2}^\top |\cB_t]\bigg]\\
        &=\E\bigg[\E[\beta_t|\cB_t]\E[ X_{t-1}I_{t-2}^\top |\cB_t] - \E[\beta_t |\cB_t]\E[X_{t-1}|\cB_t]\E[I_{t-2}^\top |\cB_t]\bigg] \\
        &=\E\bigg[\E[\beta_t]\E[ X_{t-1}I_{t-2}^\top |\cB_t] - \E[\beta_t]\E[X_{t-1}|\cB_t]\E[I_{t-2}^\top |\cB_t]\bigg] \\
        &=\E[\beta_t]\E[\cov(X_{t-1}, I_{t-2}|\cB_t)],
    \end{align*}
    where we use the independence $\beta_t\indep(X_{t-1}, I_{t-2},\cB_t)$ to conclude $\beta_t \indep (X_{t-1}, I_{t-2}) |\cB_t$ and $\beta_t \indep X_{t-1} |\cB_t$ (using the `Weak Union' property, \citet{lauritzen1996graphical}) and $\beta_t \indep \cB_t$. 
    
    This shows statement \textit{(i)}, and statement \textit{(ii)} follows by the same arguments as the proof of \cref{prop:ts-civ}.
\end{proof}

\subsection{Proof of \texorpdfstring{\cref{lemma:ts-niv-inhomogeneous}}{}} \label{sec:newproooof2}
\tsNivInhomogeneous*
\begin{proof}
    The full-time graph of $S$ is the same as that in \cref{prop:ts-niv}, and therefore it follows by the same proof as of \cref{prop:ts-niv} that there exists $\alpha$ such that $\cov(Y_t - \beta_t X_{t-1} - \alpha \cZ_t, \cI_t) = 0$. From this, one can derive:
    \begin{align*}
        0 &= \cov(Y_t - \beta_t X_{t-1} - \alpha \cZ_t, I_{t-2}) \\
        &= \cov(Y_t  - \alpha \cZ_t, I_{t-2}) - \cov(\beta_t X_{t-1}, I_{t-2}) \\
        &= \cov(Y_t  - \alpha \cZ_t, I_{t-2}) - \E[\beta_t] \cov(X_{t-1}, I_{t-2}) \\
        &= \cov(Y_t - \E[\beta_t] X_{t-1} - \alpha \cZ_t, I_{t-2}).
    \end{align*}
    Here, the third equality follows because
    \begin{align*}
        \cov(\beta_t X_{t-1}, I_{t-2}) &= \E[\beta_t X_{t-1} I_{t-2}^\top] - \E[\beta_t X_{t-1}]\E[I_{t-2}^\top] \\
        &= \E[\beta_t]\E[X_{t-1} I_{t-2}^\top] - \E[\beta_t]\E[X_{t-1}]\E[I_{t-2}^\top] \\
        &= \E[\beta_t]\cov(X_{t-1}, I_{t-2}).
    \end{align*}
    
    This shows statement \textit{(i)}, and statement \textit{(ii)} follows by the same arguments as the proof of \cref{prop:ts-niv}.
\end{proof}
}

\end{document}